\documentclass{article}
\usepackage[left=1in,right=1in,top=1in,bottom=1in]{geometry}
\usepackage{amsfonts,amsmath,mathrsfs,bbm,amsthm, amssymb, enumitem,comment,apptools, stmaryrd}
\usepackage{xcolor}
\usepackage{authblk}
\usepackage{appendix}
\usepackage{enumitem}
\usepackage{graphicx}
\usepackage[hidelinks]{hyperref}
\usepackage{url}
\usepackage{tikz-cd}
\usepackage{bbm}
\usepackage{pdfpages}
\definecolor{codegray}{rgb}{0.5,0.5,0.5}
\usepackage{listings}
\lstdefinestyle{mystyle}{
  keywordstyle=\color{blue},
  numberstyle=\tiny\color{codegray},
  basicstyle=\ttfamily\footnotesize,
  captionpos=b,                   
}
\usepackage{orcidlink}
\DeclareMathOperator{\diag}{diag}
\DeclareMathOperator{\loc}{loc}

\DeclareMathOperator{\id}{id}
\DeclareMathOperator{\sgn}{sgn}
\DeclareMathOperator{\spn}{span}
\DeclareMathOperator{\supp}{supp}
\DeclareMathOperator{\tr}{Tr}
\DeclareMathOperator{\spec}{spec}
\DeclareMathOperator{\res}{res}
\newcommand{\HS}{\mathrm{HS}}
\newcommand{\one}{\mathbbm{1}}

\newcommand{\acts}[1]{\overset{#1}{\curvearrowright}}

\newcommand{\C}{\mathbb C}
\newcommand{\Z}{\mathbb Z}
\newcommand{\A}{\mathcal{A}}
\newcommand{\B}{\mathcal{B}}
\newcommand{\M}{\mathbb{M}}
\newcommand{\N}{\mathbb N}

\DeclareMathOperator{\diam}{diam}

\renewcommand{\P}{\mathbb{P}}

\newcommand{\E}{\mathbb{E}}
\newcommand{\ca}{C(\Omega, \mathcal{A})}
\newcommand{\caz}{C(\Omega, \mathscr{A}_{\mathbb{Z}})}
\DeclareMathOperator{\Aut}{Aut}

\newcommand{\<}{\langle}
\renewcommand{\>}{\rangle}
\newcommand{\ket}[1]{|#1\>}
\newcommand{\bra}[1]{\<#1|}

\newcommand{\calq}{\mathcal{Q}}
\newcommand{\cale}{\mathcal{E}}

\let\emptyset\varnothing

\newtheorem{thm}{Theorem}[section]
\newtheorem{thmx}{Theorem}

\newtheorem{prop}[thm]{Proposition}
\newtheorem{lem}[thm]{Lemma}
\newtheorem{cor}[thm]{Corollary}
\newtheorem*{thm*}{Theorem}
\theoremstyle{definition}
\newtheorem{define}[thm]{Definition}
\newtheorem{exmp}[thm]{Example}
\newtheorem{rmk}[thm]{Remark}
\newtheorem{ass}{Assumption}
\newtheorem{note}[thm]{Notation}
\newtheorem{claim}{Claim}[thm]

\numberwithin{equation}{section}
\numberwithin{thm}{section}
\begin{document}
\title{\huge\textbf{Finitely Correlated States Driven by Topological Dynamics}}
\author{Eric B. Roon\orcidlink{0009-0000-7566-5734}\thanks{\url{rooneric@msu.edu}}\quad\&\quad Jeffrey H. Schenker\orcidlink{0000-0002-1171-7977}\thanks{\url{schenke6 @ msu.edu}}\\Department of Mathematics\\ Michigan State University\\ East Lansing, MI., U.S.A.}
\date{\today}                     
\setcounter{Maxaffil}{0}
\renewcommand\Affilfont{\itshape\small}
\maketitle

\begin{abstract}
    Let $(\Omega, \P)$ be a standard probability space and let $\vartheta:\Omega \to \Omega$ be a measure preserving ergodic homeomorphism. Let $\mathcal{A}$ be a $C^*$-algebra with a unit and let $\mathcal{A}_\Z$ be the quasi-local algebra associated to the spin chain with one-site algebra $\mathcal{A}$. Equip $\A_\Z$ with the group action of translation by $k$-units, $\tau_k\in Aut(\A_{\Z})$ for $k\in \Z$. We study the problem of finding a disordered matrix product state decomposition for disordered states $\psi(\omega)$ on $\mathcal{A}_\Z$ with the covariance symmetry condition $\psi(\omega) \circ \tau_k = \psi(\vartheta^k \omega)$. This can be seen as an ergodic generalization of the results of Fannes, Nachtergaele, and Werner \cite{FannesNachtergaeleWerner}. To reify our structure theory, we present a disordered state $\nu_\omega$ obtained by sampling the AKLT model \cite{AKLT} in parameter space. We go on to show that $\nu_\omega$ has a nearest-neighbor parent Hamiltonian, its bulk spectral gap closes, but it has almost surely exponentially decaying correlations, and finally, that $\nu_\omega$ is time-reversal invariant with a Tasaki index of $-1$ almost surely. 
\end{abstract}

\tableofcontents

\section{Introduction}
Studying the effects of disorder on quantum many-body systems has been a topic of great interest in recent years yielding a vast array of results (e.g. see \cite{Abanin_et_al, BolsDeRoeck, DeRoeckHuveneerset_al, HamzaSimsStolz, NachtergaeleReschke}) and has taken great inspiration from the disordered single-body literature. For example, the study of zero-velocity Lieb-Robinson bounds is analogous to the study of (complete) Anderson localization \cite{AizenmanWarzel, HamzaSimsStolz} (see the books and surveys \cite{AizenmanWarzel, BugerolLacroix, CarmonaLacroix, Stollmann, Stolz}). 

Understanding the structure of disordered states in the quantum spin system context is a crucial task on the way to gaining deeper insight about disordered interacting systems and their thermodynamic properties. Therefore, obtaining access to a structure theory for disordered quantum spin states would be a large step forward in understanding disordered quantum systems on the whole.  In general, however, understanding disordered states has many challenges, even for the single-body case, and there is not a clear analogy from single-particle dynamics to inform what methods one might try in the interacting many-body case. 

In the deterministic many-body setting, great progress has been made in the last three decades by utilizing the structure of the underlying lattice to reduce the properties of a translation invariant bulk state to properties witnessed by a smaller ancilla. Specifically, the \emph{matrix product states} formalism, pioneered by \cite{FannesNachtergaeleWerner} in the early 1990's, allows one to model important properties of a deterministic, finitely-correlated, translation-invariant state on a quantum spin chain via linear dynamics on a finite-dimensional ancilla. Out of the MPS theory grew the so-called tensor network theory as a multi-dimensional generalization of this important class of states \cite{BridgemanChubb, Cirac_et_al}.

Recent steps toward a theory of \emph{disordered} MPS have been made in \cite{MovassaghSchenker_PRX, MovassaghSchenker}, and \cite{NelsonRoon}. These works focus on constructing MPS-like states that are  \emph{translation covariant}, i.e., translation invariant in distribution. The aim of this paper is to show, conversely, that disordered states with finite correlations and translation covariance admit a kind of disordered MPS factorization. Furthermore, we show that states with finite correlations are weak* dense in the set of all translation covariant states. Therefore, we establish a robust machinery to study TCVS that parallels the development of MPS for translation invariant states. 

As an example, we perform an in-depth investigation into a disordered deformation of the famed AKLT state \cite{AKLT}, which is well-known to belong to a parameterized family of MPS representations which admit a nearest-neighbor parent Hamiltonian where the MPS is the unique gapped ground state. Additionally, the AKLT model has a nontrivial Time-Reversal symmetry index. In our disordered deformation, we show that there is a disordered nearest-neighbor parent Hamiltonian for which our disordered state is the ground state and that there is a nontrivial time-reversal index almost surely. However, in contrast with the original AKLT state, the correlation structure of the disordered-AKLT type state is such that the bulk gap necessarily vanishes for \textit{any} finite range, translation covariant disordered parent Hamiltonian.  We believe that the low-lying excited states of this model represent a ``quasi-gap'' somewhat analogous to a mobility gap in disordered Fermion systems. We comment on this further in the open problems section~\ref{sec:conclusion} below. 

\subsection{Background and Relation to Other Works}
Since their inception in the early 1990's \cite{FannesNachtergaeleWerner}, MPS have served as an important tool to construct and approximate ground states of local Hamiltonians on  infinite quantum spin chains. These models provide a rich insight into the structure of quantum states in one spatial dimension, and furthermore such MPS serve as useful approximates of more general states \cite{FannesNachtergaeleWerner_abundance, VerstraeteCirac}. After an observation by Vidal \cite{Vidal} that finite-chain MPS are computationally efficient, their utility for simulating quantum many body systems in finite chains cemented the MPS formalism as a standard tool for the modern condensed matter theorist. As such, a large volume of work has been devoted to understanding MPS and their properties since then. See the non-exhaustive list \cite{Beaudry_et_al, Cirac_et_al, FannesNachtergaeleWerner_abundance, FannesNachtergaeleWerner_entropy, FannesNachtergaeleWerner, FannesNachtergaeleWerner_pure,   FernandezGonzalez_et_al, Heikkinen, Nachtergaele, Perez-Garcia_et_al, Souissi}. 

Recently, various authors have begun to explore \emph{disordered} MPS on both the finite and infinite chain, \cite{Chen_et_al, LancienPerezGarcia, MovassaghSchenker_PRX, MovassaghSchenker,   Souissi}, as well as random tensor networks (see e.g. \cite{Hayden_et_al, HuNechita}). Much of the random MPS literature thus far aims to understand probability distributions on MPS \textemdash \ that is, a translation-invariant MPS selected at random \textemdash \ and their relevant properties up to some large probability \cite{JauslinLemm, LancienPerezGarcia}. Another facet of the probabilistic interpretation of these works is their connection to Hidden Markov Models (see \cite{Fanizzaet_al, Souissi, SouissiAndolsi, SouissiAndolsi_II}).

Recent work by Movassagh and Schenker introduced the \emph{ergodic matrix product state} \cite{MovassaghSchenker_PRX, MovassaghSchenker} which is \emph{not} translation invariant. By contrast to questions of typicality for translation-invariant states, Movassagh and Schenker aimed to study a homogeneous system with extensive disorder, resulting in a thermodynamic limit which does not have translation symmetry. A generalization of \cite{MovassaghSchenker} to the case when the on-site algebra is a kind of infinite-dimensional operator algebra was completed by Nelson and Roon in \cite{NelsonRoon}. One of the key properties that both sets of work centered on is the notion of \emph{translation covariance} with respect to an ergodic map. More specifically, if $(\Omega, \P)$ denotes a probability space equipped with a measure-preserving ergodic map $\vartheta:\Omega \to \Omega$, and $\omega \mapsto \psi_\omega^{[-N, N]}$ is a disordered state on the interval $[-N, N] \cap \Z$ which is generated by an ergodic MPS representation, then in the thermodynamic limit $N\to \infty$ we are interested in studying states which satisfy translation covariance:
    \begin{equation}\label{eqn:translation_covariance}
        \psi^{\Z}_{\omega} \circ \tau_k = \psi^{\Z}_{\vartheta^k \omega}\,.
    \end{equation} 
In words, equation~(\ref{eqn:translation_covariance}) says that translation of an observable leads to a corresponding shift in the disorder with \emph{statistically} identical expectation values, but not necessarily \emph{trajectory-wise} identical values. Indeed the construction of the relevant thermodynamic limiting state in both \cite{MovassaghSchenker, NelsonRoon} shows~(\ref{eqn:translation_covariance}) is an emergent property from the local, ergodic MPS data. This sets ergodic MPS apart from the states studied in \cite{JauslinLemm, LancienPerezGarcia}, since there, disorder is introduced when picking the MPS representation, but the same representation is repeated throughout the lattice. Translation co-variance leads to a distinct type of global symmetry for disordered states. A quick note on the terminology: in the works \cite{EkbladMorenoNadalesRoonSchenker, RoonSchenker_BGA} the authors use the term `ergodic,' to align more closely with the Schr\"odnger operators literature, however in this work our results in Theorem~\ref{thmx:Fundamental} and~\ref{thmx:w*dense} rely only on covariance with respect to a homeomorphism, so the latter term is more accurate.

Both \cite{MovassaghSchenker} and \cite{NelsonRoon} took inspiration from \cite{FannesNachtergaeleWerner}, but each is only a one-way construction of a translation covariant state from local data. By comparison, the authors in \cite{FannesNachtergaeleWerner} give a complete classification of translation invariant states in terms of the linear dynamics on an auxiliary $C^*$-algebra. In the case this auxiliary is finite-dimensional, one recovers the notion of a $C^*$-Finitely Correlated State (abbreviated $C^*$FCS). In general, $C^*$-FCS are the mixed state versions of MPS in the thermodynamic limit (cf. Proposition 3.1 of  \cite{FannesNachtergaeleWerner}).  

Part of the appeal of MPS in the first place is that they have a relatively simple description and properties of the thermodynamic limit can be calculated using the matrices from which the state is comprised. For example, topological indices have been used to study the symmetry properties of MPS \cite{Ogata_z2,Ogata_cmp, Ogata_cdm,  Ogata_icm, Pollman_et_al_z2, Pollmann_et_al_entanglement, Tasaki_prs, Tasaki_jmp}. Of particular interest to us is the \emph{Tasaki index} originating in \cite{Tasaki_prs} which is a $\Z_2$-valued index, $\sigma_{TR}$, that in the deterministic MPS case separates the AKLT state from the trivial phase \cite[Corollary 2]{Tasaki_prs}. This index is defined for time-reversal invariant states, where the symmetry is implemented by a local spin flip $S_j^{(\alpha)} \longmapsto -S^{(\alpha)}_j$ for all $j\in \Z$ and $\alpha =1,2,3$. Tasaki calculated his index by taking the large-volume limit of expectations of a nonlocal unitary operator $T_L$ called the \emph{Affleck-Lieb Twist} operator defined below in \eqref{eqn:Affleck-Lieb_Twist}.

It is currently an open question as to whether or not Tasaki's index agrees with the index defined by Ogata in \cite{Ogata_z2}, while it is known that Ogata's index is mathematically equivalent to the index defined in Pollmann et al. \cite{Pollman_et_al_z2}. The upshot of the aforementioned works is that they each show that $\sigma_{TR} = -1$ for the AKLT ground state, distinguishing it from the ``trivial phase'' of states with $\sigma_{TR}=1$. Another way symmetries of MPS have been studied is with the appearance of nontrivial string order parameters \cite{PG_StringOrder}, which we do not investigate in this work.

\subsection{Main results}
We restrict our attention to models of disorder which can be described by a suitable probability space. Our working assumption for the rest of the paper is as follows. 

\begin{ass}\label{Assumption:cptH}
The underlying probability space $(\Omega, \mathcal{F}, \P)$ is a compact Hausdorff space, $\mathcal{F}$ is the Borel $\sigma$-algebra generated by open sets, and $\P$ is a Radon measure. Furthermore, $\Omega$ is equipped with a measure-preserving homeomorphism $\vartheta:\Omega \to \Omega$. 
\end{ass}
\begin{rmk}
For example, one could take $\Omega=K^{\mathbb{Z}}$ with $K$ a compact set,  $\vartheta$ the bilateral shift, and $\mathbb{P}=\bigotimes_{j\in \mathbb{Z}} \mu$ with $\mu$ fixed Radon measure on $K$.  This would model independent, identically-distributed disorder at each site $j\in \mathbb{Z}$.
\end{rmk}
That is, we aim to investigate the structure of translation covariant states which are `driven' by a topological dynamical system. This is actually quite a general assumption as we discuss in Section~\ref{sec:prob} in the preliminaries. More specifically, the Jewett-Krieger theorem (see~\ref{thm:JK_ergodic} below) shows that (up to a.e. equal versions) every ergodic dynamical system on a \emph{standard} probability space can be modeled by a probability space satisfying Assumption~\ref{Assumption:cptH}.  See Lemma~\ref{lem:IID_cts} below.

To continue our discussion, let us briefly recall the construction of a quantum spin chain. Let $\mathcal{A}$ be a $C^*$-algebra with a unit $\one$, and consider the integer lattice $\Z$. To each $z\in \Z$ we associate an isomorphic copy $\mathcal{A}_{\{z\}}$ of $\mathcal{A}$.  For each finite subset $\Lambda \subset \Z$, form the local algebra 
\[
    \mathcal{A}_{\Lambda} = \bigotimes_{z\in \Lambda} \mathcal{A}_{\{z\}}\,.
\] For most models of interest, the one-site algebra is a matrix algebra $\A_{\{x\}}\cong \M_n$, for which the usual tensor product is sufficient. However, our results in section~\ref{sec:Structure_Theorem} hold in more generality, so we note that in the case when $\A_{\{x\}}$ is a general $C^*$-algebra, the tensor product is taken to be the minimal tensor product (see \cite{BrownOzawa}).  By taking $\Lambda \uparrow \Z$ and the norm closure, one forms the quasi-local algebra $\mathcal{A}_\Z$ for which $\mathcal{A}_{\Lambda} \subset \mathcal{A}_\Z$ for all $\Lambda \Subset \Z$ (see the Section~\ref{subsubsec:quasilocal} in the preliminaries for more details). The weakly* compact set of states on $\mathcal{A}_\Z$ is written $\mathcal{S}(\mathcal{A}_{\Z})$. 

In our analysis, we consider the bi-partition of the system into left $\mathcal{L}$ and right-hand $\mathcal{R}$ sides, corresponding to observables supported on $\{\dots, -2, -1\}$ and $ \{0, 1, 2,\dots\}$ respectively. Note that $\mathcal{A}_{\Z} = \mathcal{L} \otimes \mathcal{R}$. Given a state $\psi \in \mathcal{S}(\mathcal{A}_{\Z})$ and an element $b\in \mathcal{L}$ one obtains a linear functional on $\mathcal{R}$ by taking $$ a \mapsto \psi(b\otimes a) \ .$$

Consider a function
    \begin{equation}
    \begin{split}
        \psi:&\Omega \to \mathcal{S}(\mathcal{A}_{\Z})\,,\\
        &\omega \longmapsto \psi_\omega\,,
    \end{split}
    \end{equation} which is weakly* continuous in $\omega$. We regard $\psi$ as a disordered state supported on the bulk $\mathcal{A}_\Z$.

\begin{thmx}[Theorem~\ref{thm:Fundamental} Informal]\label{thmx:Fundamental}
    Let $\psi_\omega$ be a weakly* continuous function of states. Suppose further that $\psi_\omega$ is translation covariant, in the sense that 
    \[
        \psi_\omega\circ \tau = \psi_{\vartheta \omega}\,,
    \] where $\tau$ is the shift on $\A_{\Z}$. Then the following are equivalent: 
    \begin{enumerate}[label = \alph*.)]
        \item The linear space $\{ \psi(b\otimes ( \, \cdot \, ) ) \ : \ b\in \mathcal{L}\}\subset \mathcal{R}^*$  is finite-dimensional almost surely. 
        \item There exists an auxiliary bundle $\mathcal{B} = \bigsqcup_{\omega\in \Omega} \mathcal{B}_\omega$ where each $\B_\omega$ is a finite dimensional vector space with an associated family of linear maps $\{E_{a,\omega}:B_\omega \to B_{\vartheta^{-1}\omega} \colon a\in \mathcal{A}\}$, a linear functional $\varrho_\omega \in \B_\omega^*$, and a distinguished element $e(\omega) \in \B_\omega$ such that the following factorization holds:
        \begin{equation}\label{eqn:fundamental_factor}
            \psi_\omega(a_m \otimes a_{m+1}\otimes \cdots \otimes a_{n}) = \varrho_{\vartheta^{m-1}\omega}\circ E_{a_m, \vartheta^m \omega} \circ \cdots E_{a_n, \vartheta^n\omega}(e(\vartheta^n\omega))\,.
        \end{equation}
    \end{enumerate}
\end{thmx} This result can be viewed as a  disordered generalization of Proposition 2.1 in \cite{FannesNachtergaeleWerner} and is the starting point for our analysis. In fact, under Assumption~\ref{Assumption:cptH}, we show there is a topology on $\mathcal{B}$ making the aforementioned \textbf{transfer apparatus} $(E_{a,\omega}, \varrho_\omega, e(\omega))$ continuous in $\omega$ (see Theorem~\ref{thm:Fundamental} below). If a state $\psi_\omega$ satisfies the conditions in Theorem~\ref{thmx:Fundamental}, we say it has \textbf{small correlations}.

To prove Theorem~\ref{thmx:Fundamental}, we utilize the theory of Banach bundles which was pioneered by J. M. G. Fell in the 1980's. These are the Banach-space analog of vector bundles \cite{Husemoller, Gierz}, and allow us to define a class of dynamical maps which depend continuously on a disorder parameter. Banach bundles have slightly more flexibility than vector bundles since they do not require local triviality as part of the definition \cite{DiximierDouady,  DupreGillette, FellDoran}. A key result is the Fell-Doran theorem (see Theorem~\ref{thm:Banach_Bundle} below), which allows one to construct a Banach bundle from a disjoint set of Banach spaces once a vector space of sections is specified \cite{FellDoran}. In fact the transfer maps defined in part b.) of Theorem~\ref{thmx:Fundamental} are \emph{homomorphisms} of Banach bundles in the sense of \cite{Gutman_1,GutmanKoptev}. The general structure of Banach bundles is still an active topic of investigation (see e.g. \cite{Chirvasitu}).

Our construction of the Banach bundle $\B$ follows that of \cite{FannesNachtergaeleWerner}, where they start with a trivial transfer apparatus factorization and obtain results if the fibers are finite dimensional.  However for us, the finite-dimensionality assumed in part a.) of Theorem~\ref{thmx:Fundamental} seems to be necessary for the Banach bundle structure to hold with the specific equivalence relation we use to define $\B_\omega$. In general if the fibers $\B_\omega$ are infinite dimensional, the best we can say is that they form a \emph{normed quotient bundle} in the sense of \cite{ResendeSantos}.

Another key result of \cite{FannesNachtergaeleWerner} is that MPS are weakly*-dense in the convex set of translation-invariant states. Our next main result is that the same is true for translation covariant states. 

\begin{thmx}[Theorem~\ref{thm:w*_dense} Informal]\label{thmx:w*dense}
     Let $\mathfrak{C}$ denote the set of translation covariant states on $\A_\Z$ and denote the subset of those translation covariant states with small correlations by $\mathfrak{K}\subset \mathfrak{C}$. One has that $\mathfrak{K}$ is weakly-* dense in $\mathfrak{C}$ uniformly in $\omega$. 
\end{thmx}

In fact, more structure is present as there exists an ordering on the $\B_\omega$'s for which all maps of the form $E_{a^*a, \omega}$ are completely positive. This makes each $\B_\omega$ into an \emph{operator system}, and in fact $\B_\omega$ is an \emph{operator system quotient} in the sense of \cite{KavrukPaulsenTodorov}. Interestingly, one can prove a variant of Theorem~\ref{thmx:Fundamental} for which all the fibers are $C^*$-algebras by considering the GNS representations of $\mathcal{A}_\Z$ with respect to $\psi_\omega$ indexed by the disorder parameter. We show a similar decomposition to equation~(\ref{eqn:fundamental_factor})  holds on this $C^*$-algebraic bundle $\{\pi_\omega(\mathcal A_{\Z})\}_{\omega \in \Omega}$ (see Proposition~\ref{prop:C*Bundle} below). At this stage, it is not clear which presentation of the underlying bundle is more advantageous.

To reify our theoretical results, we now present a disordered model that admits a parent Hamiltonian  which is almost surely gapless in the thermodynamic limit. Before we describe our model, we note that due to the work of \cite{FernandezGonzalez_et_al} in the deterministic case, one can always construct a frustration-free parent Hamiltonian for an injective MPS whose gap closes in the bulk. However unlike \cite{FernandezGonzalez_et_al}, our disordered parent Hamiltonian arises by following the typical construction as in \cite{FannesNachtergaeleWerner, Perez-Garcia_et_al}. By contrast the authors of \cite{FernandezGonzalez_et_al} construct the so-called uncle Hamiltonian as a perturbation of the usual parent Hamiltonian. 

For our construction, we make use of the following one-parameter family of isometries $V_{\alpha} : \C^2 \to \C^3 \otimes \C^2$ 
    \[
        V_{\alpha}^* = \begin{bmatrix}
                           0 & \cos(\alpha) & -\sin(\alpha) & 0 & 0 & 0 \\
                           0 & 0 & 0 & \sin(\alpha) & -\cos(\alpha) & 0
                        \end{bmatrix}; \quad \alpha\in [0, \pi)\,,
    \] which is introduced in \cite{FannesNachtergaeleWerner} in order to  generalize the AKLT model \cite{AKLT}. We construct a state as follows. Set $E_{a,\omega} = V_\omega^*(a \otimes (\,\cdot \,))V_{\omega}:\M_2 \to \M_2$ for $a\in \M_3$ the $3\times 3$ matrices. Then, for any interval $[-m, n]$ define 
        \[
            \nu^{[m,n]}_{\omega} (a_{m}\otimes \cdots \otimes a_n) = \frac12 \tr \{ E_{a_m, \omega_m} \circ \cdots \circ E_{a_n, \omega_n}(\one_{2\times 2})\}\,,
        \] where $\omega = (\omega_j)_{j\in \Z}\in [0, \pi)^{\Z}$. By Takeda's theorem (see Theorem~\ref{thm:Takeda} below), these finite volume data define a global state $\nu^\Z_{\omega}$ for each $\omega\in [0, \pi)^\Z$. By equipping $[0, \pi)^{\Z}$ with a product measure, we can use this space to model a bi-infinite sequence of independent, identically distributed random parameters $(\omega_j:j\in \Z)$ taking values in $[0, \pi)$, such that $\omega_j = (S\omega)_{j-1}$ where $S$ is the shift operator. We then have the following result. 

\begin{thmx}[Theorem~\ref{thm:IID_AKLT_Gapless}, Informal]\label{thmx:Gapless_Hamiltonian}
    Suppose $(\omega_j: j\in \Z)$ is a family of IID random variables taking values in $[0, \frac{\pi}{4}]$ such that $\P[\omega_0=0]=0$ but $\P[\omega_0 <\delta] >0$ for all $\delta>0$. Then the following hold for the disordered state $\nu^\Z_\omega$ as constructed above. 
    \begin{enumerate}[label = \roman*.)]
        \item There is a nearest-neighbor parent Hamiltonian $\{h_{j, j+1}(\omega)\}_{j\in \Z}$ for which $\nu^{\Z}_\omega$ is the bulk ground state of the GNS Hamiltonian.  
        \item Furthermore, $\nu^{\Z}_\omega$ is $\P$-a.s. pure.
        \item With probability one, there is a sequence of intervals $I_k \subset \Z$ with $|I_k| \ge k$ and such that correlations of $\nu^{\Z}_{\omega}$ do not decay over $I_k$.
        \item Therefore, the GNS Hamiltonian $H^\Z_\omega$ associated to $\nu^{\Z}_{\omega}$ is almost surely gapless. In particular, \emph{any} translation covariant finite ranged Hamiltonian for which $\nu_\omega$ is the ground state must have a vanishing spectral gap. 
        \item Nevertheless, with probability one, $\nu^{\Z}_{\omega}$ has exponentially decaying correlations.
    \end{enumerate} 
\end{thmx}

\begin{rmk}
    The parent Hamiltonians $h_{j, j+1}(\omega)$ we construct are projection-valued and also translation co-variant with respect to the disorder in the sense that  \[
    h_{j,j+1}(\omega) = h_{j+1, j+2}(\vartheta^{-1}\omega)\,.
    \] Furthermore they are frustration free for $\nu_\omega$ in the sense that $\nu_\omega(h_{j,j+1}(\omega)) \equiv 0$ for all $j\in \Z$. 
\end{rmk}

 Owing to the translation covariance, we conclude that the spectrum of the GNS Hamiltonian for $\nu^{\Z}_\omega$ is deterministic in Corollary~\ref{cor:spec_const} below (see \cite{RoonSchenker_BGA} for a more general statement). If there were a nonzero bulk gap, the exponential clustering theorem \cite{NachtergaeleSims} would imply that the two point correlator $|\nu^{\Z}_\omega(S_{x_k}^z S_{x_k+k}^z)|$ is decaying exponentially fast uniformly at all scales with a necessarily deterministic rate. The key step in the proof that the bulk-gap closes, is to isolate rare regions where the values of $\omega$ are close to zero for long stretches. In such a rare region, we prove a lower bound on the correlation function $|\nu^{\Z}_\omega(S_{x_k}^z S_{x_k+k}^z)|\ge f(k)$ for which $\limsup_{k\to \infty} f(k) = 1$. This violates the bound provided by the exponential clustering theorem and thus contradicts the existence of a bulk gap. 

Based on the transfer apparatus, we find explicit formulas for vector linear functionals $\langle \Gamma_L| \cdot |\Gamma_L\>$ on $\A_{[-L, L+1]}$ which approximate $\nu^{\Z}_\omega$ as $L\to \infty$ almost surely. Given these expansions, we show explicitly that $\nu^{\Z}_\omega$ has non-trivial Tasaki index. Recall \cite{Ogata_z2}, time-reversal symmetry is an automorphism of $\A_{\Z}$ sending local spin matrices $S^w_m \mapsto -S^{w}_m$ for all $w\in\{x,y,z\}$ and $m\in \Z$.  The state-of-the-art methods for matrix product states in the deterministic case, involve computing a unitary that locally implements the symmetry \cite{Ogata_cdm,Ogata_icm,Ogata_z2, Tasaki_book}. However, these results rely on a strong uniqueness theorem for injective matrix product states, namely the Fundamental Theorem of Injective Matrix product states \cite{Cirac_et_al}, which does not yet have a disordered analog. 

Therefore, we turn to the computational methods of \cite{Tasaki_prs}, which uses the Affleck-Lieb Twist Operator defined as 
\begin{equation}\label{eqn:Affleck-Lieb_Twist}
    T_L := \bigotimes_{x=-L}^{L+1}\exp\left\{-2\pi i \frac{x+L}{2L+1}\,S^{z}_x \right\}\in \A_{[-L,L+1]}\,.\end{equation} 
In the deterministic case of a gapped ground state $\phi$ with time-reversal symmetry, \cite{Tasaki_prs} argued that the sign of the expectation value $\mathscr{O}(\phi):=\lim_{L\to \infty}\sgn(\phi(T_L)) = \pm 1$, and this determines a $\Z_2$-index. Namely, \cite{Tasaki_prs} shows that the AKLT state has an index value equal to $-1$. By contrast, we show the following.

\begin{thmx}[Theorem~\ref{thm:Z2-index}, Informal]\label{thmx:TasakiIndex}
    Let $\nu^\Z_\omega$ be the disordered deformation of the AKLT model described in Theorem~\ref{thmx:Gapless_Hamiltonian}. Then $\nu^{\Z}_\omega$ is invariant under time-reversal almost surely. Furthermore,
    \[
       \P\left[ \lim_{L\to \infty} \nu^{\Z}_\omega(T_L) = -1\right] =1\,.
    \] That is, the value of the index $\mathscr{O}(\nu^{\Z}_\omega)$ as defined by \cite{Tasaki_prs}, converges to $-1$ for $\P$-a.e. realization of the disorder. 
\end{thmx}

\subsection{Remarks and Open Questions}\label{sec:conclusion}
In this work, we build on the results of \cite{MovassaghSchenker} to provide a complete structure theory for the emerging field of Ergodic Matrix Product States. We expect our mathematical structure will serve as a convenient language to describe abstract results about this new class of disordered states, and will help provide a framework from which deeper results can be wrought. We have also exhibited a  deformation of the famed AKLT state which cannot be a gapped ground state of any local Hamiltonian with the same `statistical' translation symmetry. Nevertheless, due to the strict positivity of the ergodic quantum process which implements translation inside the state, we find that it is still exponentially clustering almost surely. Among the many interesting properties of this disordered state, we find that it retains some understanding of the symmetry protection index enjoyed by its deterministic gapped cousins. 

While our results are new, we note that there are echoes of our work to be found in the literature about models of independent Fermion states in a mobility gap \cite{AizenmanGraf}. Namely, such states retain exponential clustering despite closure of the gap, and they also have deterministic indices which can differentiate phases. Our example suggests the definition of a new class of states we shall describe as \textbf{quasi-gapped ground states}. 

Let us recall that gapped ground states of quasi-local interactions obey exponential clustering in the following way \cite{NachtergaeleSims}. Suppose $a_i$ and $a_j$ are single-site observables located at sites $i$ and $j$ respectively, and $\psi$ is a deterministic gapped ground state of a quasilocal interaction (i.e. the interaction is gapped and has a finite $F$-norm). Then, after some rewriting, Theorem 2 of \cite{NachtergaeleSims} says 
\begin{equation}\label{eqn:exponential_clustering}
    |\psi([a_i,a_j])| \le 2 \|a_i\| \|a_j\| e^{-\mu\,(|i-j| - \ell_0)}\,,  
\end{equation} 
where $\mu>0$ is a function of the gap size. Here we have suggestively absorbed the pre-factors appearing in the original estimate from \cite{NachtergaeleSims} into the constant $\ell_0 >0$, which we call the \emph{onset length} for decay. Observe that $\ell_0$ is a lower bound on the minimal distance between $i$ and $j$ so that decay of correlations at exponential rate $\mu$ can be observed; indeed, for $|i-j|< \ell_0$, eq. \eqref{eqn:exponential_clustering} is weaker than the trivial estimate $|\psi([a_i,a_j])| \le 2 \|a_i\| \|a_j\|$. Thus we can replace eq. \eqref{eqn:exponential_clustering} with the stronger bound
\begin{equation}\label{eqn:exponential_clustering_positive_part}
    |\psi([a_i,a_j])| \le 2 \|a_i\| \|a_j\| e^{-\mu\,(|j| - \ell_0)_+}\,,  
\end{equation} 
where $(x)_+ = \frac{1}{2}(x+|x|)$ denotes the positive part of a real number $x$.  Note that although the state $\psi$ and the interaction were not assumed to be translation invariant, the assumption of a gap provides a translation invariant upper bound on correlations.

By contrast, Ergodic Matrix Product States --- which our group has recently shown are always ergodic ground states of a canonical parent Hamiltonian \cite{EkbladMorenoNadalesRoonSchenker} --- obey a kind of randomized exponential clustering \cite[Theorem 2]{MovassaghSchenker} where now there is a \emph{random onset length} $L_0(\omega)$, that is finite almost surely. In symbols, if $\psi_\omega$ is an ergodic MPS ground state, then, there exists $\mu>0$ and $L_0(\omega)$ so that, with probability one,
\begin{equation}\label{eqn:quasigap}
   |\psi_\omega ( [a_0, a_j])| \le 2\|a_0\| \|a_j\| e^{ - \mu\, (|j| - L_0(\omega))_+}  
\end{equation}
for all single site observables $a_0$ and $a_j$ located at $0$ and $j$, respectively. Note that we have localized this estimate to $i=0$.  Indeed, a key feature of ergodic MPS is that, by the translation covariance, the onset length may be spatially dependent. That is, with probability one,
\begin{equation}\label{eqn:quasigap_translated}
      |\psi_\omega ( [a_i, a_j])| \le 2\|a_{i}\| \|a_{j}\| e^{ - \mu\, (|i-j| - L_0(\vartheta^i\omega))_+}.
\end{equation}
for all single site observables $a_{i}$ and $a_{j}$ located at sites $i$ and $j$ respectively.

We venture the preliminary definition that an ergodic ground state $\psi_\omega$ is \emph{quasigapped} if \eqref{eqn:quasigap} holds but \eqref{eqn:exponential_clustering} fails for all \emph{deterministic} onset lengths, where a full definition should require considering correlations between any two local observables, not just those supported at a single site. In words, quasigapped states retain almost sure exponential clustering in the presence of a closed bulk spectral gap, but allow for varying onset-length along the chain.

In fact our results hint at the possibility of two types of translation covariant models. In our recent pre-print \cite{RoonSchenker_BGA}, we show that the bulk spectrum of any global Hamiltonian which is translation covariant is necessarily deterministic. It would appear, then, that there ought to be deformed AKLT models where the gap actually remains open almost surely, and some where the gap will close. In fact, we provide a class of such gapped examples in the upcoming work \cite{EkbladMorenoNadalesRoonSchenker}. We conjecture that the distinction can be detected in the clustering properties of the state in question.

As yet several very interesting questions remain open.  
\begin{enumerate}
 \item Does a theory of Symmetry Protected Topological phases hold for quasigapped states? Namely, if one begins with a quasigapped state and connects via homotopy to another quasigapped state which does not close the \emph{quasi}-gap, can we say these states belong to the same phase?  If a state belongs to, say, a parameterized phase in some neighborhood $(a,b)$ with a phase transition at $b$, will disordered sampling of the MPS outside of $(a,b)$ destroy the phase information? 
    \item In the single body literature, there is a robust machinery around the notion of a mobility gap, and our disordered states seem to be many-body analogues of topological insulators. If this analogy is correct, can we use information about the Banach bundle to deduce phase information about the state? Do the excited states of such ergodic MPS models obey a kind of spatial localization like those in a mobility gap?
    \item Are there examples of ergodic MPS whose small correlations bundle is nontrivial? If so, does this nontriviality portend a closure of the spectral gap in the thermodynamic limit? 
    \item What are the minimal conditions on a small correlations bundle that will ensure the resulting state is pure almost surely? Are these conditions analogous to those in \cite{FannesNachtergaeleWerner}, and if so, are the pure states weak-* dense like they were shown to be in the MPS case \cite{FannesNachtergaeleWerner_abundance}? 
    
\end{enumerate}
We mention these with the hope to inspire the reader and other members of the quantum spin systems community to investigate this interesting family of disordered states.
\subsection{Organization of the Paper}
We will now outline the structure of our paper for the reader's convenience. 

In section~\ref{sec:Prelims}, we collect results which are more or less tacitly known by experts in the community, but which we shall use frequently. This section is purely for the reader's convenience and she/he/they may skip it if desired. In Section~\ref{sec:Structure_Theorem} we discuss our main results and describe the structure theory of translation covariant states in terms of Banach bundles. In Section~\ref{sec:application}, we study a version of the parameterized Hamiltonian introduced in \cite{FannesNachtergaeleWerner} where the parameter is taken to be an IID random variable at each site. The resulting bulk state is shown to have an almost surely gapless parent Hamiltonian while simultaneously being described by our Banach bundle structure. Furthermore, this disordered state is shown to be invariant under time-reversal symmetry with an almost sure Tasaki index equal to $-1$. We conclude with an appendix~\ref{apdx:mathematica} wherein we record our computational results.  

\section*{Acknowledgments}
EBR would like to thank the Operator Algebras Reading Group of Michigan State University for patiently listening to his expository talk on Banach bundles. He also thanks the other members of Prof. Schenker's research group Messrs. Owen Ekblad and Eloy Moreno-Nadales for the productive and enjoyable Thursday-morning discussions. The work of JHS was supported by the National Science Foundation under Grant No. 2153946. We would also like to thank Dr. Jacob Shapiro, Dr. Bruno Nachtergaele, and Dr. Daniel Spiegel for the helpful discussions and suggestions which improved this work. 

\section{Preliminaries}\label{sec:Prelims}
In this section, we recall some preliminary facts that we will use frequently throughout the rest of the exposition. Some of the results of this section are standard, while others are more specialized to our purposes. We aim to provide proofs when necessary, and otherwise we will provide detailed references for the reader's convenience. 

\subsection{Probability Theory}\label{sec:prob}
We will denote by $(\Omega, \mathcal{F}, \P)$ a probability space with $\sigma$-algebra $\mathcal{F}$ and probability measure $\P$. A map $T:\Omega \to \Omega$ is \textit{measure preserving} whenever $A\in \mathcal{F}$ implies $\P(T^{-1}(A))=\P(A)$ and is further called \textit{ergodic} if $T^{-1}(A)\subset A$ implies $\P(A) \in \{0, 1\}$. In general, not every measure preserving map has a measure preserving inverse. If $T$ is a bijection and both $T$ and $T^{-1}$ are measure-preserving, then $T$ is called an \textit{automorphism} of $(\Omega, \mathcal{F}, \P)$. 

\begin{rmk}\label{rmk:ergodic_inverse}
    If $T$ is an ergodic automorphism of $\Omega$, then $T^{-1}$ is necessarily ergodic. This can be seen by noting that for any $A\in \mathcal{F}$ with $T(A)\subset A$, the set $B:= \bigcup_{n=1}^\infty T^{-n}(A)$ satisfies $T^{-1}(B)\subset B$ and $\P[A] = \P[B]\in \{0,1\}$ by continuity from below. 
\end{rmk}

Throughout this paper, we assume that $\Omega$ is a compact Hausdorff space and that $\mathcal{F}$ is the Borel $\sigma$-algebra.  Furthermore, we will assume a given measure-preserving homeomorphism $T$ on $\Omega$; that is, $T$ is a continuous bijection with a continuous inverse. There are two standard examples of this case which motivate our study.

\begin{exmp}\label{exmp:Quasiperiodic_rotation}
Let $\Omega = [0,\pi)$ the unit circle equipped with the normalized Lebesgue measure $\P = \frac{dx}{\pi}$. It is well known (see \cite[Example 6.1.3]{Durrett}) that for any irrational number $\alpha\in (0,\pi)$, there is an associated measure preserving ergodic map $\vartheta:\omega \mapsto \omega + \alpha \mod{\pi}$ which is a homeomorphism with an ergodic inverse. 
\end{exmp}

Another standard example comes in the form of independent identically distributed random variables taking values in a compact Hausdorff space. 

\begin{lem}\label{lem:IID_cts}
    Let $(\Omega, \mathcal{F}, \P)$ be a (not necessarily compact Hausdorff) probability space. Let $\{X_z:z\in \mathbb{Z}\}$ be a bi-infinite sequence of IID random variables on $\Omega$. Then, there exist a compact Hausdorff probability space $(S, \mathcal{G}, \mathbb{Q})$ and a family of random variables $\{Y_z:z\in \Z\}$ and an ergodic automorphism $\vartheta \in \Aut(S)$ so that the following hold
        \begin{enumerate}[label = \alph*.)]
            \item Each $\{Y_z : z\in \Z\}$ are IID.
            \item The automorphism $\vartheta$ is a homeomorphism.
            \item For all $z\in \Z$ one has $Y_z = Y_{z-1}\circ \vartheta = Y_{z+1}\circ \vartheta^{-1}$.
            \item The finite dimensional distributions agree 
            \[
                \mathbb{Q}[Y_{z_1}\le x_1, Y_{z_2}\le x_2, \dots, Y_{z_n}\le x_n] = \P[X_{z_1}\le x_1, X_{z_2}\le x_2, \dots, X_{z_n}\le x_n] \ .
            \]
        \end{enumerate} 
\end{lem}
\begin{rmk}\label{rmk:IID_cts}
    In words, the previous lemma states that we can always find a model for a sequence of IID random variables in which the underlying probability space is a compact Hausdorff space with an ergodic homeomorphism.  The lemma may be proved by pushing forward the probability measure $\mathbb{P}$ to the sequence space $\dot{\mathbb{R}}^\mathbb{Z}$ where $\dot{\mathbb{R}}$ denotes the one-point compactification of $\mathbb{R}$. 
\end{rmk}

A much more powerful result holds which characterizes measure-preserving transformations on the Lebesgue space $([0,1], dx)$ in terms of ergodic homeomorphisms. Recall the following from \cite[Corollary 12.30]{Eisner_et_al} or \cite[Section 4.4]{Petersen}.

\begin{thm}[Jewett-Krieger Theorem]\label{thm:JK_ergodic}
    If  $T$ is an ergodic, measure preserving transformation on a Lebesgue space $(X, \mathscr{B}, \mu)$ where $\mathscr{B}$ is the Borel $\sigma$-algebra and $\mu$ is a probability measure, then $(X,T,\mu)$ is isomorphic (up to sets of measure zero) to an ergodic homeomorphism of the Cantor set. 
\end{thm}

The proof of this is deep and difficult, so we will not recount it here. However, in the context of our results, the above theorem says that one can interpret the measure theoretic data of a large array of probabilistic models in terms of topological dynamics, so that (in some sense) topological dynamics are an appropriate setting for the discussion of disorder. 
\subsection{Banach Bundles}\label{sec:Bundles}
In this section we will recall some facts about Banach Bundles which will be relevant to our quotient construction in section~\ref{sec:Bundles} below. More specifically, we find that if the fibers are finite dimensional, our family of quotient spaces forms a Banach bundle in Theorem~\ref{thm:quotients_are_bbundled} below. Thus it is necessary to recount the relevant details.

The theory of Banach bundles goes back to at least the 1960's with \cite{DiximierDouady}, and the book by Fell and Doran \cite{FellDoran} serves as an informative reference. As there is some variance in the basic definitions of these objects, we point out that our working definition is that of Fell and Doran given in \cite[Definition 13.1]{FellDoran} (also written as (F)-Banach Bundles in \cite{Chirvasitu, DupreGillette}). 

\subsubsection{Definition and Properties}
Let $\Omega$ be an Hausdorff topological space and let $B$ be an Hausdorff space equipped with a continuous, open surjection $\pi: B\to \Omega$. If $U\subset \Omega$ is open, the subset $B_U:= \pi^{-1}(U)$ is the \textit{restriction} of $B$. The pre-image $B_x := \pi^{-1}(\{x\})$ is called the \textit{fiber} with base-point $x$.   We say that $B$ is a \textit{Banach Bundle} whenever each $B_x$ is equipped with the structure of a complex Banach space and the following conditions are satisfied: 
\begin{enumerate}[label = (\roman*)]
    \item The mapping $e \mapsto \|e\|$, the Banach space norm of $e$ in $B_{\pi(e)}$, is a continuous map $B \to [0, \infty)$.
    \item The binary operation $(e,f) \mapsto e+f$ is a continuous map from $\Delta(B^2):=\{(e,f): \pi(e) = \pi(f)\}$ to $B$, where the domain $\Delta(B^2)$ is taken with the subspace topology inherited from the product topology on $B^2$.
    \item For each $\lambda \in \mathbb{C}$, the mapping $m_\lambda :B \to B$ via $e\mapsto \lambda e$ is continuous. 
    \item For every $x\in X$, the sets of the form 
        \[
            V_{U, \epsilon}:= \{e\in B \colon \pi(e) \in U, \|e\|<\epsilon\} \quad \text{so that } x\in U\subset \Omega \text{ is open and }\epsilon >0,
        \] form a neighborhood basis around the origin $0_x$ of the fiber $B_x$. 
\end{enumerate} 

\begin{rmk}
Item (iv) is equivalent to the following: 

\indent (iv)' If $x\in X$ and $(e_{j})_{j\in \mathcal{J}}\subset B$ is a net so that $\pi(e_j) \to x$ and $\|e_j\| \to 0$, then $e_j \to 0_x$ in $B$. 

\end{rmk}

We recall that a \textit{section} of $\pi$ is a continuous function $f:\Omega\to B$ so that $\pi \circ f = \id_\Omega$.  A Banach bundle is said to \textit{have enough sections} if for every $e\in B$, there is a section $f$ so that $f(\pi(e)) = e$. We denote the set of continuous sections by $\Gamma(B)$. Let us recall a fundamental result.

\begin{thm}[Theorem C.16 \cite{FellDoran}]
    Let $B$ be a Banach bundle over a locally compact Hausdorff space $\Omega$. Then $B$ has enough continuous sections. 
\end{thm}

When equipped with the uniform norm, the linear space $\Gamma(B)$ of continuous sections of $(B, \pi, \Omega)$ forms a Banach space equipped with a natural multiplication by the algebra of $\C$-valued functions $C(\Omega)$. Moreover, if $f\in \Gamma(B)$ and $\lambda \in C_0(\Omega)$, one clearly has $\|\lambda f\|_{\Gamma(B)} \le \|\lambda\|_{C(\Omega)} \cdot \|f\|_{\Gamma(B)}$. This makes $\Gamma(B)$ into a \textit{Banach $C(\Omega)$-module} \cite[Chapter 2]{DupreGillette} (see also \cite{Chirvasitu}). 

\begin{lem}\label{lem:bundle_convergence}
    Let $B$ be a Banach bundle over a compact Hausdorff space $\Omega$. Then, a net $(e_j)_{j\in \mathcal{J}}$ converges to a point $e$ in the topology of $B$ if and only if the net of base-points $\pi(e_j) \to \pi(e)$ in $\Omega$ and one of the following equivalent criteria are met 
        \begin{enumerate}[label = (\roman*)]
            \item For every $f\in \Gamma(B)$ one has 
            \[
                \|e_j - f(\pi(e_j))\| \to \|e- f(\pi(e))\|\,,
            \] in $\mathbb{R}$. 
            \item Let $W\subset \Gamma(B)$ be a dense $C(\Omega)$-sub-module. For every $w\in W$, one has 
            \[
                \|e_j - w(\pi(e_j))\| \to \|e- w(\pi(e))\|\,.
            \]
        \end{enumerate}
\end{lem}  
\begin{proof}
    $(i)$: The forward direction follows from continuity. For the converse, note there is some $g\in \Gamma(B)$ with $g(\pi(e)) = e$. Now, $\|e - g(\pi(e))\| = 0$, and by continuity, this means 
    $\|e_j - g(\pi(e_j))\| \to 0.$ By property (iv)', this implies that $e_j - g(\pi(e_j)) \to 0_{\pi(e)}$, from which we conclude that $\lim e_j = \lim g(\pi(e_j)) = e$, since addition is continuous. $(iii)$: This follow by a straightforward $\epsilon/3$ argument. 
\end{proof}

We now recall a useful lemma. The original argument seems to go back to \cite{DiximierDouady}, but we record the proof below in English for the convenience of the reader.
\begin{lem}
    Let $B$ be a Banach bundle over $\Omega$ with a space of continuous sections $\Gamma$. Then, the set 
        \[
            \Omega_{\ge n}:=\{x \colon \dim B_x \ge n\} \,,
        \] is open for each $n\ge 0$
\end{lem}
\begin{proof}
    Suppose $x_0 \in \Omega_n$ and find $s_1,\dots s_n \in \Gamma$ so that $\{s_j(x_0)\}_{j=1}^n$ is linearly independent in $B_{x_0}$. Consider the map 
        \[
            s(y, \vec{c}) := \sum_{j=1}^n c_j \cdot s_j(y)\,,
        \] which is a jointly continuous map from $X\times \C^n \to B$. Then, the function $\eta(y) = \inf_{\|\vec{c}\|=1}\|s(y, \vec{c})\|$ is a continuous mapping from $X \to [0, \infty)$. Since the unit sphere of $\C^n$ is compact, $\eta(x_0)>0$ by linear independence. Therefore, there is a neighborhood $V\ni x_0$ so that $\eta(y) >0$ for all $y\in V$. Equivalently, $\{s_j(y)\}_{j=1}^n$ are linearly independent in $B_y$. Thus $V\subset \Omega_n$ as required.
\end{proof} 

In other words, the above lemma says that the dimension function $x \mapsto \dim B_x$ is lower semi-continuous in $x$. The next theorem is a simplified version of the Stone-Weierstra{\ss} theorem for Banach bundles Theorem 4.2 in \cite{Gierz}. 
\begin{thm}[Stone-Weierstra{\ss} theorem for Banach Bundles]
    Let $\Omega$ be a compact Hausdorff space and let $B$ be a Banach bundle over $X$. A $C(\Omega)$-submodule $W$ of the continuous sections $\Gamma(B)$ is dense in $\Gamma(B)$ if and only if for every $x\in \Omega$, the set $\{s(x) : s\in W\}\subset B_x$ is dense.
\end{thm}
\begin{proof}
    For the forward direction, note that since $\Omega$ is compact, it has enough sections \cite[Theorem C.11]{FellDoran}. So for any $e\in B_x$, there is a continuous section $s\in \Gamma(B)$ so that $s(x) = e$. Therefore, for any $\epsilon >0$ one may find $w\in W$ so that $\|w(x) - s(x)\| \le \|w - s\| <\epsilon$ as required. 

    Conversely, pick $s\in \Gamma(B)$. Fix $x_0 \in \Omega$. Since $W$ is fiber-wise dense, there exists a section $w_{x_0} \in W$ so that $\|w_0(x_0) - s(x_0)\| < \epsilon$. Since the mapping $X\ni x\mapsto \|w(x) - s(x)\|$ is continuous, there is a neighborhood $U_{x_0}\ni x_0$ so that for all $y\in U_{x_0}$ one has $\|w_0(y) - s(y)\|<\epsilon$. \textit{Nota bene} $\{U_{x} : x\in \Omega\}$ (constructed as above) forms an open cover of $\Omega$ from which we can extract a finite sub-cover $\{U_{x_1}, U_{x_2}, \dots, U_{x_k}\}$. Find a partition of unity $\{f_1:X \to [0,1]\colon 1\le j \le k\}\subset C(\Omega)$ subordinate to this sub-cover, and define $w = \sum_{1}^k f_j w_{x_j}\in W$ since $W$ is a $C(\Omega)$-module. 

    Hence, for any $y\in X$, 
        \begin{align*}
            \|w(y) - s(y)\| \le \sum_{j=1}^k f_j(y) \|w_j(y) - s_j(y)\| &\le \sum_{\substack{j_*:\\ y\in U_{x_{j_*}}}} |f_{j_*}(y)| \|w_{j_*}(y) - s_{j_*}(y)\|< \epsilon \sum f_j(y) = \epsilon.\qedhere
        \end{align*} 
\end{proof}

Moreover, \cite{FellDoran} give a general means of constructing Banach bundles from families of Banach spaces. 
\begin{thm}[Theorem 13.18 \cite{FellDoran}]\label{thm:Banach_Bundle}
    Let $\Omega$ be an Hausdorff space. Let $A$ be a set equipped with a surjection $\varrho :A \to \Omega$, so that $A_x := \varrho^{-1}(\{x\})$ is a complex Banach space. Suppose $\Gamma$ is a complex linear space of sections of $\varrho$ so that 
    \begin{enumerate}[label = (\alph*)]
        \item for all $f\in \Gamma$, the mapping $x\mapsto \|f(x)\|_{A_x}$ is continuous,
        \item for all $x\in X$, the set $\{f(x)\}_{f\in \Gamma} \subset A_x$ is dense. 
    \end{enumerate} 
Then, there exists a unique topology $\tau$ making $A$ a Banach bundle so that every element of  $\Gamma$ is a continuous section of $\varrho$. In particular, a topological basis for $\tau$ is given by sets of the following form: 
    \begin{equation}
        W(f, U, \epsilon) = \{a\in A \colon \varrho(a) \in U,\, \|a- f(\varrho(a))\|<\epsilon\}\,,
    \end{equation} 
where $\epsilon >0$, $U\subset \Omega$ is an open set, and $f\in \Gamma$. 
\end{thm}

We will sometimes refer to Theorem~\ref{thm:Banach_Bundle} as the \textit{Fell-Doran Theorem} for simplicity. 

 Recall that a Banach bundle is \textit{locally trivial} if for every point $x\in \Omega$, there exists a neighborhood $U\ni x$ and a fixed Banach space $F$, so that the restriction $B_U$ is homeomorphic to $F\times U$ in the product topology. In contrast to the typical definition of a vector bundle, Banach bundles make \emph{no assumptions} about local triviality, which leads to some interesting behavior.

 \begin{exmp}
     Let $\Omega = [0,1]$ be the unit interval equipped with the usual topology. Define $V_x=\C$ for all $x\in [0,1]\setminus \{1/2\}$ and $V_{1/2} = \{0\}$. Let $\Gamma$ be the subspace of $C([0,1])$ with $f(1/2) = 0$. Then by the Fell-Doran Theorem (Theorem~\ref{thm:Banach_Bundle}) there is a unique topology on $\mathcal{V}:=\{V_x\}_{x\in [0,1]}$ making $(\mathcal{V}, \pi, [0,1])$ into a Banach bundle with $\Gamma$ as a dense $C([0,1])$-submodule of the continuous sections by the Stone-Weierstra{\ss} theorem for Banach bundles. In fact, the resulting bundle can be realized as the set of continuous functions on $[0,1]$ that vanish at $\frac{1}{2}$. 
 \end{exmp}

\subsubsection{Homomorphisms}
We now recapitulate the theory of mappings between Banach bundles. It seems that the first formulation of notion of a \emph{homomorphism of Banach bundles} is due to \cite{Gutman_1}, and these were used in \cite{Gutman_1,GutmanKoptev} to investigate the notion of a \textit{dual bundle}. In view of Theorem~\ref{thm:Banach_Bundle}, one might hypothesize that the natural definition of the `dual' of a Banach bundle $X$ would be to take the dual spaces of each fiber $\{X_\omega^*: \omega \in \Omega\}$ together with the natural surjection, and topologize it via the homomorphisms into the trivial bundle which form a natural class of sections thereof. However, as is shown in Examples 2.13 of \cite{GutmanKoptev}, even in the case that the fibers of a given bundle are finite-dimensional, it is \textit{not} the case that the point-wise norm of every homomorphism is continuous -- a requirement of Theorem~\ref{thm:Banach_Bundle}.  

In general, a drop in dimension allows for the operator norm of a linear functional to have a discontinuity in its base points. A pointwise construction of a dual bundle is actually quite subtle for this reason. However, for practical purposes it will suffice to study \textit{homomorphisms} between bundles. In particular, we relax the requirement that the linear functionals form a bundle in and of themselves, to requiring that they are homomorphisms into the trivial bundle.

We introduce a slightly modified version of the definition of a homomorphism. 
\begin{define}\label{def:bundle_hom}
    Let $X$ and $Y$ be Banach bundles over the same topological base space $\Omega$, with modules of sections $\Gamma(X)$ and $\Gamma(Y)$ respectively. Suppose further that $\Omega$ is equipped with an homeomorphism $\vartheta:\Omega \to \Omega$. A function $H: \Omega \to \bigsqcup_{\omega \in \Omega} B(X_\omega, Y_{\vartheta \omega})$ via $\omega \mapsto H_\omega$ is a ($\vartheta$-)\textbf{covariant homomorphism} of $X$ to $Y$ if the associated mapping of the bundles $H:X\to Y$ via $e\mapsto H_{\pi(e)}e\in Y_{\vartheta \omega}$ is continuous. If $\vartheta$ is the identity map, then $H$ is simply referred to as a \textit{homomorphism}. 
\end{define}

\begin{note}
    The $C(\Omega)$-module of covariant homomorphisms will be written $\hom^\vartheta(X,Y)$, with the $\vartheta$-dropped if we are referring to the identity map on $\Omega$. If the pointwise norm function $\omega \mapsto \|H_\omega\|_{B(X_\omega, Y_{\vartheta\omega})}$ is bounded, then we say $H$ is a \textit{bounded homomorphism}. The submodule of bounded homomorphisms will be written $\hom^\vartheta_b(X,Y)$. Note that every covariant isometric isomorphism of Banach bundles is a covariant homomorphism. 
\end{note}

There is a natural action of the homomorphisms $\hom^\vartheta(X,Y)$ on the Banach modules of continuous sections $\Gamma(X)$ and $\Gamma(Y)$, as we note in the following lemma. 

\begin{lem}\label{lem:hom_action_sections}
    Let $\Omega$ be a compact Hausdorff space equipped with an homeomorphism $\vartheta$. Let $(X, \pi_X, \Omega)$ and $(Y, \pi_Y, \Omega)$ be two Banach bundles over $\Omega$ with $C(\Omega)$-modules of sections $\Gamma(X)$ and $\Gamma(Y)$ respectively. Let $H:\Omega \to \bigsqcup_{\omega\in \Omega}B(X_\omega, Y_{\vartheta \omega})$ be a $\vartheta$-covariant homomorphism. Then, for every section $\sigma \in \Gamma(X)$, the function $H\cdot \sigma: \Omega \to Y$ via $(H \cdot \sigma)(\omega) := H_\omega \sigma(\omega) \in Y_{\vartheta \omega}$ is continuous.
\end{lem}
\begin{proof}
    Let $(\omega_{\lambda})_{\lambda \in \Lambda}\subset \Omega$ be a net converging to $\omega$. Note that 
    \[
        \pi_Y(H_{\omega_\lambda}\sigma(\omega_\lambda)) = \vartheta \circ \pi_X(\sigma(\omega_\lambda)) = \vartheta \omega_\lambda \,,
    \] which converges to $\vartheta \omega = \pi(H_\omega \sigma(\omega))$. Now, let $f\in \Gamma(Y)$ be a continuous section, and note that in this case, the net $f(\vartheta (\omega_\lambda)) \to f(\vartheta( \omega))$, and similarly $\sigma(\omega_\lambda) \to \sigma(\omega)$. Whence $\|H_{\omega_\lambda} \sigma(\omega_\lambda) - f(\vartheta \omega)\| \to \|H_\omega \sigma(\omega) - f(\vartheta \omega)\|$ since $H$ is continuous in the bundle topology. Thus, $H\cdot \sigma$ is a continuous mapping into $Y$. 
\end{proof}

The following lemma uses the central idea in the proof of \cite[Theorem 1.4.5]{Gutman_1} to show that every covariant homomorphism of bundles over a compact Hausdorff space has a uniformly bounded point-wise norm. 

\begin{lem}
        Let $\Omega$ be a compact Hausdorff space equipped with an homeomorphism $\vartheta$. Let $(X, \pi_X, \Omega)$ and $(Y, \pi_Y, \Omega)$ be two Banach bundles over $\Omega$ with $C(\Omega)$-modules of sections $\Gamma(X)$ and $\Gamma(Y)$ respectively. Let $H\in \hom^\vartheta(X,Y)$. Then, $H$ is bounded, i.e.,  \[\|H\|_\infty := \sup_{\omega \in \Omega} \|H_\omega\|_{B(X_\omega, Y_{\vartheta \omega})}<\infty.\]
\end{lem}
\begin{proof}
    Recall that $\Gamma(X)$ and $\Gamma(Y)$ form Banach spaces when endowed with the uniform norm. By Lemma~\ref{lem:hom_action_sections}, the assignment $\omega \mapsto \|H_\omega \sigma(\omega)\|$ is continuous for every $\sigma\in \Gamma(X)$. Therefore, $\sup_{\omega} \|H_\omega\sigma(\omega)\| := K_\sigma <\infty$ since $\Omega$ is compact. 

    Setting $T_\omega:\Gamma \to \ell^1(\bigsqcup Y_\omega)$, via $T_\omega \sigma := H_\omega \sigma(\omega)\in Y_{\vartheta\omega}$, the above observation yields 
        \[\sup_{\omega} \|T_\omega \sigma\|_{Y_{\vartheta\omega}}= K_\sigma <\infty\,, \] 
    for all $\sigma \in \Gamma$. Therefore, by the uniform boundedness principle (\cite[Theorem 2.6]{Rudin}), $M:=\sup_{\omega} \|T_\omega\|<\infty$. We are now ready to estimate $\|H_\omega e\|$ for $e\in X_\omega$. Our aim is to express $\|H_\omega e\|$ in terms of $T_\omega \sigma(\omega)$, for a suitable section $\sigma$. Recall that $\Omega$ being compact implies $X$ has enough sections: for all $e\in X_\omega$, there is some $\sigma_e\in \Gamma$ with $\sigma(\omega) = e$.  We need to modify $\sigma_e$ a bit in order to control its norm. 
    
    Let $\epsilon >0$. By  Urysohn's Lemma (\cite[Theorem 33.1]{Munkres}), there is a continuous function $f_\epsilon:\Omega \to [0,1]$ so that for $A = \{\pi(e)\}$ and $B = \{\omega: \|\sigma_e(\omega)\| \ge (1+\epsilon)\|e\|\}$ one has $f(A) = \{1\}$ and $f(B) = \{0\}$. Then, taking $\sigma_\epsilon = f\cdot \sigma_e$, one has $\|\sigma_\epsilon\|_\infty \le (1+\epsilon)\|e\|$, and $\sigma_\epsilon(\pi(e)) = e$. Then, 
        \[
            \|H_\omega e\|_{Y_{\vartheta \omega}} = \|H_\omega \sigma_\epsilon (\omega)\| = \|T_\omega \sigma_e\| \le M \|\sigma_e\|_\infty \le M (1+\epsilon) \|e\|\,,
        \] and we are done. 
\end{proof}

The following is a specialization of \cite[Theorem 1.4.4]{Gutman_1}
\begin{lem}
    Let $\Omega$ be a compact Hausdorff space equipped with an homeomorphism $\vartheta$. Let $(X, \pi_X, \Omega)$ and $(Y, \pi_Y, \Omega)$ be two Banach bundles over $\Omega$ with $C(\Omega)$-modules of sections $\Gamma(X)$ and $\Gamma(Y)$ respectively. Let $H:\Omega \to \bigsqcup_{\omega\in \Omega}B(X_\omega, Y_{\vartheta \omega})$. The following are equivalent:
    \begin{enumerate}[label = (\roman*)]
        \item The function $H$ is a $\vartheta$-covariant homomorphism. 
        \item For every section $\sigma \in \Gamma(X)$, the functions $H\cdot \sigma: \Omega \to Y$ via $(H \cdot \sigma)(\omega) := H_\omega \sigma(\omega) \in Y_{\vartheta \omega}$ are continuous.
        \item For every section $\sigma$ in a dense $C(\Omega)$-submodule $W\le \Gamma(X)$, the function $H\cdot \sigma$ is a continuous. 
    \end{enumerate}  
\end{lem}
\begin{proof}
    $(i) \Rightarrow (ii):$ This is Lemma~\ref{lem:hom_action_sections}. $(ii)\Rightarrow (iii):$ This is obvious. $(iii) \Rightarrow (i)$: Let $(e_\lambda)_{\lambda \in \Lambda}\subset X$ be a net converging to $e_0$ in the bundle topology. Let $\omega_{\lambda}:= \pi(e_{\lambda}) \to \omega_0 := \pi(e_0)$ be the associated net of base points which we know converges. Let $\epsilon >0$. Since $X$ has enough continuous sections, we may find $\sigma_0\in \Gamma(X)$ so that $\sigma_0(\omega_0) = e$. Furthermore, as $W$ is a dense sub-module of sections, we may find $w\in W$ with $\|w - \sigma_0\|_{\infty}<\frac{\epsilon}{3\|H\|_\infty}$. 

    Now, let $f\in \Gamma(Y)$. Note that $\| f\circ \vartheta - H\cdot w \|$ is a continuous function in $\omega$, and that moreover by the reverse triangle inequality,
    \[
        \left| \|f(\vartheta \omega_0) - H_{\omega_0}w(\omega_0)\| - \|H_{\omega_0}e - f(\vartheta \omega_0)\| \right| \le \| H_{\omega_0}\| \|e - w(\omega_0)\| = \| H_{\omega_0}\| \|\sigma_0(\omega_0) - w(\omega_0)\|\le \frac{\epsilon}{3}\,.
    \] Therefore, for any $\eta >0$ there is $\lambda$ large enough so that
    \begin{equation}\label{pf:hom_continuity}
        \left| \|f(\vartheta \omega_\lambda) - H_{\omega_\lambda}w(\omega_\lambda)\| - \|H_{\omega_0}e - f(\vartheta \omega)\| \right| \le \frac{\epsilon}{3} + \eta\,.
    \end{equation}

The following bound holds by combining the triangle inequality and the reverse triangle inequality: 
    \begin{align*}
        \left| \|H_{\omega_\lambda}e_\lambda - f(\vartheta \omega) \| - \|H_\omega e - f(\vartheta \omega) \| \right|&\le \left| \|H_{\omega_\lambda}e_\lambda - f(\vartheta \omega_\lambda) \| - \|f(\vartheta \omega_{\lambda}) - H_{\omega_\lambda} \sigma_0(\omega_{\lambda}) \| \right|\\
        &\text{\hspace{8mm}}+\left| \|f(\vartheta \omega_{\lambda}) - H_{\omega_\lambda} \sigma_0(\omega_{\lambda}) \| - \|f(\vartheta \omega_{\lambda}) - H_{\omega_\lambda}w(\omega_\lambda) \| \right|\\
        &\text{\hspace{8mm}}+\left| \|f(\vartheta \omega_{\lambda}) - H_{\omega_{\lambda}}w(\omega_\lambda)\| - \| H_{\omega}w(\omega) - f(\vartheta \omega)\| \right|\\
        &\le \|H\|_{\infty}\left( \|e_\lambda - \sigma_0(\omega_{\lambda})\| +\|\sigma_0 - w\|_{\infty}\right) +\frac{\epsilon}{3} + \eta\\
        &\le \epsilon +\eta \text{ for $\lambda$ large enough.}
    \end{align*} Note we have used~(\ref{pf:hom_continuity}) to obtain the estimate on the third summand in the second inequality. Since $\eta$ and $\epsilon$ were chosen independently, this completes the proof.
\end{proof}

\subsection{$C^*$-algebras and Quasi-local Algebras}\label{sec:C*algs}

Let us recall here that a $C^*$-algebra, $\mathcal{A}$ is a complete normed algebra with norm $\|\cdot \|$ satisfying the $C^*$-condition $\|a^*a\| = \|a\|^2$. We say that $\A$ is unital if it possesses a unit element $\one$, and we say that $\mathcal{A}$ is separable if, with the topology generated by $\|\cdot \|$, $\mathcal{A}$ contains a countable, dense subset.

In view of the following theorem, we never loose generality by assuming a $C^*$-algebra is contained in the bounded operators on a Hilbert space. This is particularly useful for some constructions, but we aim to present our results abstractly wherever possible.

\begin{thm}[Theorem 7.11 \cite{Conway_OT}]
    Every $C^*$-algebra has a faithful representation $\pi$ into the bounded operators on a Hilbert space $\mathcal{H}$.
\end{thm}

In fact, more is true once one specifies a reference state. 

\begin{thm}[Gelfand-Naimark-Segal Theorem]
    Let $A$ be a unital $C^*$-algebra and $\psi$ be a state on $A$. There exists a Hilbert space $\mathcal{H}$, a $*$-representation $\pi$, and a normalized cyclic vector $\xi\in \mathcal{H}$ so that $\psi(a) = \<\xi, \pi(a)\xi\>$. 
\end{thm}

Let $\mathcal{A},\mathcal{B}$ be two unital $C^*$-algebras acting on the Hilbert spaces $H$ and $K$ respectively. One may form their algebraic tensor product in the standard way: 
\[
    \mathcal{A}\odot \mathcal{B} = \left\{\sum_{k=1}^n a_k \otimes b_k \colon n \text{ is finite}\right\}\,,
\]  (where the sums are taken modulo the usual bilinear relations). which is easily checked to be a $*$-subalgebra of $B(H\otimes K)$. We recall the minimal (a.k.a. spatial) tensor product of $\mathcal{A}$ and $\mathcal{B}$ is $\mathcal{A}\otimes \mathcal{B}= \overline{\mathcal{A}\odot \mathcal{B}}^{\|\cdot \|}$. The norm of an arbitrary element $x\in \mathcal{A}\odot \mathcal{B}$ is therefore given by 
    \[
        \|x\|_{\min} = \left\| \sum_{k=1}^n a_k \otimes b_k\right \|\,.
    \] A fundamental result is that the value of this norm is independent of the choice of faithful representations of $\mathcal{A}$ and $\mathcal{B}$ \cite[Proposition 3.3.11]{BrownOzawa}. In particular, the norm on the minimal tensor product is a (in fact, the smallest) \textit{cross norm} in the sense that $\|a\otimes b\|_{\min} = \|a\|_\mathcal{A} \|b\|_\mathcal{B}$ which also satisfies the $C^*$-property \cite[Theorem IV.4.19]{TakesakiI}. Moreover, $A\otimes_{\min} B \cong B\otimes_{\min} A$. 
    
Following the convention of Kadison and Ringrose \cite{KadisonRingroseI} among others, we shall denote by $(\mathcal{A})_\alpha$ the ball of radius $\alpha\ge0$ centered at the origin in $\mathcal{A}$. In particular, $(\mathcal{A})_1 = \{a\in \mathcal{A} \colon \|a\|\le 1\}$. Moreover, we will denote the positive semidefinite elements of a $C^*$-algebra as $(\mathcal{A})_+$.


\subsubsection{Weakly*-Continuous State-valued Functions}
Throughout, $\Omega$ will be a compact Hausdorff space, and $\A$ will be a unital $C^*$ algebra. We will primarily be concerned with the set of continuous functions $A:\Omega \mapsto \A$ which we shall write as $C(\Omega, \A)$. As the later sections will deal with comparing functions and their evaluations, we shall endeavor to always write an element of $\ca$ with capital Latin letters and elements of $\A$ with lowercase Latin letters. 

\begin{lem}\label{lem:ctsinA}
    Let $\Omega$ be a compact Hausdorff space, and let $\A$ be a $C^*$-algebra with a unit. Denote by $C(\Omega, \A)$ the set of continuous functions $F:\Omega \to \A$. The following hold. 
    \begin{enumerate}
        \item If a function $F:\Omega \to \A$ is continuous then the mapping $\omega \mapsto \|F(\omega)\|\in \mathbb{R}^+$ is continuous.
        \item Equipped with pointwise operations, $C(\Omega, \A)$ is a $C^*$-algebra under the norm 
            \[
                \|F\| := \sup_{\omega \in \Omega} \|F(\omega)\|\,.
            \] Moreover, the norm is always achieved at some $\omega_0$. 
        \item The algebraic tensor product $C(\Omega) \odot A$ can be isometrically identified with a dense $*$-subalgebra of $C(\Omega, \A)$. In particular, $C(\Omega)\otimes_{\min} A \cong C(\Omega, A)$. 
    \end{enumerate}
\end{lem}
\begin{proof}[Sketch of the proof]
    1. This follows from the reverse triangle inequality. 2. Is a standard exercise. 3. We refer the reader to \cite[Theorems IV.4.14 and IV.7.3]{TakesakiI} for the details. 
\end{proof}

Let $S(\A)$ denote the state space of $\A$, i.e., a linear functional $\psi \in S(\A)$ if and only if $\psi(\one) =1$ and $\psi(a^*a) \ge 0$ for all $a\in A$.  We will be considering maps $\psi:\Omega \to S(\A)$. Let us recall some basic information here. 

\begin{lem}
    Let $\A$ be a unital $C^*$-algebra and let $S(\A)$ denote its state space. Then, $S(\A)$ is a norm-closed convex subset of the unit ball of $\A^*$. Moreover by an application of the Krein-Milman theorem, $S(\A)$ is the weak-* closure of the convex hull of its extreme points. In particular, when $\A$ is separable, $S(\A)$ is a closed, compact metrizable space. 
\end{lem} This information is standard and can be found in \cite{Conway_FA, TakesakiI} among other sources. A nice fact we shall use with some frequency is the following:

\begin{lem}\label{lem:weak*_cts_on_C(A)}
    Let $\Omega$ be a compact Hausdorff space and suppose that $\omega\mapsto \psi_\omega$ is a weakly*-continuous mapping into the state space $S(\A)$ with respect to the subspace topology induced by the norm on $\A^*$. Then, for any $A\in C(\Omega,\A)$, the map $\omega \mapsto \psi_\omega(A(\omega))\in \mathbb C$ is continuous. 
\end{lem}
\begin{proof}[Sketch of the proof] This follows from a straightforward $\epsilon/2$ argument.
   \end{proof}

\subsubsection{Quasilocal Algebras and Quantum Spin Chains}\label{subsubsec:quasilocal}
Now, we recall the construction of the local and quasilocal algebras of observables. For a more detailed discussion of this construction, we refer the reader to Chapter 2 in \cite{BratteliRobinsonI}. Consider the integer lattice $\Z$ equipped its usual metric $d(n,m) = |n-m|$. To each site $n\in \mathbb{Z}$, we associate a unital $C^*$-algebra $\A_n$ which is an isomorphic copy of a fixed unital $C^*$-algebra, $\A$. I.e. $\A_n \cong \A$ for all $n\in \mathbb{Z}$. Given a finite subset $\Lambda \subset \mathbb{Z}$ one forms the \textit{local algebra (of observables)} 
\[
    \A_{\Lambda} := \bigotimes_{n\in \Lambda} \A_n\,,
\] where the tensor product is the minimal tensor product. We will write $\Lambda \Subset \Z$ to indicate that $\Lambda$ is a finite subset. 

Now, given any two finite volumes $\Lambda_0 \subset \Lambda_1$, there is a canonical inclusion $\iota_{\Lambda_0\to \Lambda}:\A_{\Lambda_0} \hookrightarrow \A_{\Lambda_1}$ given by $a\mapsto a\otimes \one_{\Lambda_1 \setminus \Lambda_0}$ where the latter factor indicates a $\one$ at any site in $\Lambda_1\setminus \Lambda_0$. In this way, one identifies $A_{\Lambda_0}$ with its inclusion into $\A_{\Lambda_1}$, and therefore, we may consider $\A_{\Lambda_0}\subset \A_{\Lambda_1}$ whenever $\Lambda_0 \subset \Lambda_1$. The family of local algebras $\{ \A_{\Lambda} \colon \Lambda \subset \mathbb{Z} \text{ is finite}\}$ is therefore partially ordered by inclusion. Thus, the union forms a $*$-algebra called the algebra of \textit{local observables}
    \[
        \A_{\mathbb{Z}}^{\loc}:= \bigcup_{\Lambda \Subset  \mathbb{Z}} \A_{\Lambda}\,,
    \] which is equipped with a canonical $C^*$-semi-norm $\|\cdot \|$ given by the norms on the $\A_{\Lambda}$'s. Completing $\A_{\mathbb{Z}}^{\loc}$ with respect to this seminorm forms the \textit{algebra of quasilocal observables} $\A_{\mathbb{Z}} := \overline{\bigcup_{\Lambda \Subset \mathbb{Z}} \A_{\Lambda}}^{\|\cdot \|}$. One may regard the quasilocal algebra as the inductive limit of a sequence of local algebras as well (see \cite[Chapter 6]{Murphy}).

    The following theorem was noted by Takeda in \cite{Takeda} (see also \cite{Spiegel}). 

    \begin{thm}[Takeda's Theorem]\label{thm:Takeda}
        Let $\mathbb{Z}$ be the integer lattice and consider the set $\mathcal{P}$ of all finite subsets $\Lambda \subset \mathbb{Z}$. Associate to each site a replica of the unital $C^*$-algebra $\mathcal{A}$, and form the local algebras $\A_{\Lambda} = \bigotimes_{x\in \Lambda}\mathcal{A}_x$, and quasilocal algebra $\A_{\Z}$ as above. Suppose $\psi_{\Lambda}\in \mathcal{S(\A_{\Lambda})}$ is a family of states satisfying the compatibility condition that whenever $\Lambda_0 \subset \Lambda \Subset \mathbb{Z}$ one has
        \begin{equation}
            \psi_{\Lambda} \circ \iota_{\Lambda_0 \to \Lambda} = \psi_{\Lambda_0}\,.
        \end{equation} Then, there is a uniquely defined state $\psi\in \mathcal{S}(\A_{\Z})$ so that $\psi|_{\A_{\Lambda}} = \psi_{\Lambda}$ for all $\Lambda \Subset \mathbb{Z}$.
    \end{thm} In particular, Takeda's theorem provides a rigorous justification for the existence of a product state in the thermodynamic limit. 

    We introduce some further notation for the canonical translation symmetry on $\A_{\Z}$. Let $\tau_1:\A_{\Z} \to \A_{\Z}$ denote right translation which acts on pure tensors as follows.  Given an interval $[n,n+m] \subset \Z$  and a collection $a_n,\ldots,a_{n+m} \in \A$, let $a_n\otimes \cdots \otimes a_{n+m}$ denote the element of $\A_{\Z}^{\loc}$ with the operator $a_j$ at position $j$ for each $j=n,\ldots,n+m$.  Then 
    \[
        \tau_1(a_n \otimes \cdots \otimes a_{n+m} ) = b_{n+1}\otimes \cdots \otimes b_{n+m+1}\,\forall n\in \mathbb{Z},\, m\in \mathbb{N},
    \] where $b_j=a_{j-1}$ for each $j=n+1,\ldots,n+m+1$. This definition is clearly consistent with the inclusion mappings in the sense that  $\tau_1(\A_{\Lambda_0}) = \A_{\Lambda_0 +1}$ whenever $\Lambda_0 \Subset \mathbb{Z}$. Note $\tau_k:= (\tau_1)^k$ extends to an automorphism of $\mathcal{A}_{\Z}$ for all $k\in \Z$. 

    In \cite{FannesNachtergaeleWerner}, the authors study the structure of states which are $\tau$-invariant by examining a specific quotient of the quasilocal algebra. Our aim is now to give a similar structure theorem in the case that $\psi$ is a weakly* continuous function of states which co-varies with respect to $\tau$. This leads us to the following section.

\section{Structure of Translation Covariant States}\label{sec:Structure_Theorem}
In this section, we state and prove our main structure theorem, Theorem~\ref{thm:Fundamental}, for the following class of states.

\begin{define}
    Let $\Omega$ be a compact Hausdorff space equipped with a  homeomorphism $\vartheta$. Let $\omega \mapsto \psi_\omega \in \mathcal{S}(\A_{\Z})$ be a weakly*-continuous function. We say that $\psi$ is a \textbf{translation covariant state} (TCVS) if  the following holds
        \begin{equation}
            \psi_{\omega}\circ \tau_{k} = \psi_{\vartheta^{k} \omega}\quad \forall k\in \mathbb{Z};\, \omega\in \Omega\,.
        \end{equation} 
    If additionally $\vartheta$ is an ergodic automorphism of $\Omega$ with respect to a given Radon probability measure on $\Omega$, then we say $\psi$ is an \textbf{ergodic TCVS}.
\end{define}

Speaking colloquially, Theorem~\ref{thm:Fundamental} below generalizes Proposition 2.1 of \cite{FannesNachtergaeleWerner}, showing that (in a sense we shall make rigorous) any translation covariant state factors into iterations of a composition of \textit{random linear maps} in the sense that
    \[
        \psi_\omega(a_1 \otimes \cdots \otimes a_n) = \varrho_\omega \circ E_{a_1, \vartheta \omega} \circ \cdots E_{a_n, \vartheta^n \omega}(e_{\vartheta^{n+1}\omega})\,,
    \] where the $E_{a, \omega}$'s are linear maps in an appropriate spaces $\B_\omega^+$ and $e_{\omega}$ is a random element of this space, and $\varrho_\omega$ is a random linear functional in $(\B_\omega^+)^*$.
\subsection{Setup and Basic Properties}
Let $\mathcal{A}$ be a unital $C^*$-algebra and consider the associated spin chain algebra $\A_{\mathbb{Z}}$ equipped with the group action of translation $\Z \ni k \mapsto \tau_k \in \Aut(\A_{Z})$ as in Section~\ref{sec:C*algs}. Let $\Omega$ be a compact Hausdorff space equipped with a Radon probability measure $\P$ and an ergodic measure preserving homeomorphism $\vartheta$. As we will be considering mappings from $\Omega \to \mathcal{A}$, we recall our convention that elements of $C(\Omega, \mathcal{A})$ are written with capital Latin letters and elements of $\mathcal{A}$ with lower case Latin letters, the exception being the unit, which we write as $\one$. 

 Let us record an interesting consequence of Takeda's theorem: the existence of a co-variant product state. 
\begin{prop}\label{prop:existence_TCPS}
    Let $\Omega$ be a compact Hausdorff space equipped with an (ergodic) homeomorphism $\vartheta$, and suppose that $\omega \mapsto \eta_\omega \in \mathcal{S}(\A)$ is a norm-continuous mapping. Form $\xi^{\Lambda}_\omega := \bigotimes_{x\in \Lambda} \eta_{\vartheta^x \omega}\in \mathcal{S}(\A_{\Lambda})$ for all $\Lambda \Subset \mathbb{Z}$. Then, there is a unique state $\xi_\omega$ extending the $\xi_\omega^{\Lambda}$'s that is weakly*-continuous in $\omega$, translation co-variant, 
    $$ \xi_\omega(\tau_1(a))=\xi_{\vartheta \omega}(a)\,,$$
    and has the property that for any $X,Y\Subset \mathbb{Z}$ with $X\cap Y = \emptyset$, one has
        \begin{equation}\label{eqn:TCPS_factor}
            \xi_{\omega}(a_Xb_Y) = \xi_{\omega}(a_X) \xi_{\omega}(b_Y),
        \end{equation} for all $a_X\in \A_X$ and all $B_Y\in \A_Y$. 
\end{prop}

\begin{define}\label{def:TCPS}
    We call the state $\xi_\omega$ defined as in Proposition~\ref{prop:existence_TCPS} a \textbf{translation covariant product state}. 
\end{define}

\begin{proof}[Proof of Proposition~\ref{prop:existence_TCPS}]
    For fixed $\omega\in \Omega$, we observe that $\xi_\omega$ exists and is a state on $\A_{\Z}$ by Takeda's theorem. From the definition of $\xi_\omega$ it follows that $\xi_\omega(a)=\xi_\omega^\Lambda(a)$ for any observable $a\in \A_\Lambda$. Translation co-variance, the factor property, and weak-$*$ continuity follow by direct computation using sufficiently large finite volumes.\qedhere
\end{proof} 

In view of Lemma~\ref{lem:ctsinA}, if $\mathcal{A}_0$ is some $C^*$-sub-algebra of $\A_{\mathbb{Z}}$, we may regard the continuous functions into $\mathcal{A}_0$ as a sub-algebra of $C(\Omega,\A_{\Z})$, namely by identifying $C(\Omega) \otimes_{\min} \mathcal{A}_0$ with the appropriate sub-algebra in $C(\Omega, \A_{\Z})$. With this in mind, it will be convenient for us to introduce the notion of a \textbf{cut site} as follows: let $x\in \Z+ \frac12$ be a half-integer and let $\A^{+}_{x} = \A_{[x,\infty)\cap \Z}$ denote the sub-algebra of $\A_\Z$ associated to making a cut at $x$ and including all sites greater than $x$. Similarly, let $\A^-_x = \A_{(-\infty, x)\cap \mathbb{Z}}$ be the sub-algebra of $\A_Z$ associated to sites less than $x$. For example, $\A_{1/2}^+ = \A_{\N}$ and $\A^{-}_{1/2} = \A_{-\N \cup \{0\}}$. Below we will introduce a similar equivalence relation to the one defined in \cite{FannesNachtergaeleWerner}. We will focus on the relation with respect to $\A_x^+$, but one can define analogous relations for $\A_x^-$. 

Now, let $\psi$ be a translation co-variant state as above, and consider the continuous functions $C(\Omega, \A^{+}_x)$ for a fixed $x$. Declare two functions $A,B \in C(\Omega, \A^{+}_{x})$ equivalent, written $A\sim_x B$, if and only if for all $\omega \in \Omega$ one has 
    \begin{equation}
        \psi_{\omega}( C^-(\omega) \otimes (A(\omega)-B(\omega)))=0,\,\, \forall C^-\in C(\Omega, \A^{-}_x)\,.
    \end{equation} 
    
    \textit{A priori}, these equivalence classes of functions depend on $x$. We shall write $[A]^+_x$ for the equivalence class of $A\in C(\Omega, \A^+_x)$. Furthermore, for any function $C^-\in C(\Omega, \A^-_{x})$, we note that the set\\ $\mathfrak{G}_{C^{-}}:=\{F\in C(\Omega, \A^+_x) \colon \psi(C^{-}\otimes F) = 0\}$ is a closed subspace of $C(\Omega, \A^+_x)$. This follows from the fact that $C^- \otimes F$ is a continuous mapping into $\A_{\Z}$ and the images of $\psi$ are bounded linear operators. Hence, the intersection 
        \begin{equation}
            I_{x}:= \bigcap_{C\in C(\Omega, \A_{x}^-)} \mathfrak{G}_{C} \,,
        \end{equation} is also a closed subspace. Let us write $\B^+_x:= C(\Omega, \A^+_x) / I_x$ which is a Banach space. 

    Note that we can take a point-wise approach as well. Given $a,b\in \A^+_x$, we define the equivalence relation $a\sim_{x,\omega} b$ if and only if $\psi_\omega(c\otimes (a-b)) = 0$ for all $c\in \A^{-}_x$. By essentially the same argument as above, the set $\mathfrak{g}_{c,\omega} = \{a\in \A^+_x\colon \psi_\omega(c\otimes a) =0\}$ is a closed subspace of $\A^{+}_x$ and so is the common intersection 
    \begin{equation}
        I_{x,\omega} := \bigcap_{c\in \A_{x}^-} \mathfrak{g}_{c,\omega}\,.
    \end{equation} Form the vector space $\B^{+}_{x,\omega} = \A^{+}_x/I_{x,\omega}$ which we note is a Banach space as well. 
    
    \begin{note}\label{note:fiber_notation}
        For an equivalence class in $\B^{+}_{x,\omega}$, as constructed above, we shall write $[a,\omega]_x$. 
    \end{note}

\begin{lem}\label{lem:fibers_nonzero}
    For every $\omega \in \Omega$, the dimension of $\mathcal{B}_{x,\omega}^+$ is at least one. 
\end{lem}
\begin{proof}
    Suppose there is an $\omega$ for which $\B_{x,\omega}^+$ is zero. Then, for every $a\in \A_x^+$, one has $\psi_\omega(\one \otimes a) = 0$. In particular this means $\psi_\omega(\one_{\A_{\Z}}) = 0$, which means $\psi_\omega$ is the zero linear functional, a contradiction. 
\end{proof}

    The following lemma justifies this choice of notation in~\ref{note:fiber_notation} above. 
\begin{lem}\label{lem:quot_eval}
        For every cut site $x$, and for every $\omega\in \Omega$, there exists a contractive linear map $p_{x,\omega}$ making the following diagram commute: 
        \begin{center}
            \begin{tikzcd}
                {C(\Omega, \mathcal{A}^{+}_x)} \arrow[r, "\mathrm{eval}_\omega"] \arrow[d, "\sim_x"'] & \mathcal{A}^+_x \arrow[d, "\sim_{x,\omega}"] \\
                \B_x^+ \arrow[r, "{p_{x,\omega}}", dashed]                       & {\B^+_{x,\omega}}                
        \end{tikzcd}\,;
    \end{center}
    where the vertical arrows denote the quotient maps.
\end{lem}
In this way, we can view $p_{x,\omega}$ as ``evaluation at $\omega$" for the equivalence class $[A]_{x}^+$. 

\begin{proof}[Proof of Lemma~\ref{lem:quot_eval}]
    As a preliminary, observe that $\mathrm{eval}_\omega(I_x) \subset I_{x, \omega}$. 
     Define $p_{x,\omega}([A]_{x}^+) = [A(\omega),\omega]_{x}$, which is well defined by the inclusion $\mathrm{eval}_\omega(I_x)\subset I_{x, \omega}$. Now, to see that $p_{x,\omega}$ is contractive, note the following
        \begin{align*}
            \|p_{x,\omega}[A]_x^+\|_{\B^{+}_{x,\omega}} = \inf_{b\in I_{x,\omega}} \|\mathrm{eval}_\omega(A) + b\|_{\A^+_x}&\le \inf_{b\in \mathrm{eval}_\omega(I_x)} \|\mathrm{eval}_\omega(A) + b\|_{\A^+_x}= \inf_{B\in I_x}\|\mathrm{eval}_\omega (A+B)\|_{\A^+_x}\\
            &\le \inf_{B\in I_x}\|A+B\|_{C(\Omega, \A^{+}_x)}= \|[A]^+_x\|_{\B_x}\,,
        \end{align*} where the first inequality follows from our observation that $\mathrm{eval}_\omega(I_x)\subset I_{x,\omega}.$ 
\end{proof}

\begin{rmk}
    Note that one should expect $I_{x,\omega} \supsetneq
    \mathrm{eval}_\omega(I_x)$, in general. Indeed, given $a\in I_{x,\omega}$, it is not clear how to construct a continuous function $A(\omega)$ so that $A(\omega) = a$, and $\psi_\omega(C^{-}(\omega) \otimes A(\omega))=0$ for all $\omega$ and $C^{-}\in C(\Omega, \A^{-}_x)$. 
\end{rmk}

We will now turn to a more detailed discussion of the subspaces $I_{x,\omega}$. For any $\omega\in \Omega$, one can view the $I_{x,\omega}$ as the intersection over the kernels of a family of seminorms given by 
    \begin{equation}\label{eqn:seminorms}
        \begin{split}
            \rho_{\omega,c}&: \A_{x}^+ \to \mathbb{R}^+\,,\\
            &\quad a\longmapsto |\psi_\omega(c \otimes a)|,
        \end{split} 
    \end{equation} so that $\mathfrak{g}_{c,\omega} = \rho_{\omega, c}^{-1}\{0\}.$ It is not difficult to see that the $\rho_{\omega, c}$ are linearly homogeneous in the sense that for any $\lambda \in \C$, one has 
        \[
           |\lambda| \rho_{\omega, c} = \rho_{\omega, \lambda c}\,,
        \] for any $c\in \A_{x}^-$. Thus, we obtain the equalities
    \begin{equation}
             I_{x,\omega} = \bigcap_{c\in \A_{x}^-} \mathfrak{g}_{c,\omega} = \bigcap_{c\in (\A_x^{-})_1} \mathfrak{g}_{c,\omega}\,.
    \end{equation} 

    One can similarly form a class of seminorms on $\A_x^-$ via $\sigma_{\omega,a} = |\psi_\omega(\, \cdot \, \otimes a)|$. By similar arguments one can form $\B_{x,\omega}^-$ as the quotient of $\A_x^-$ by the intersection $J_{x,\omega}:= \bigcap_{a\in \A_x^+} \sigma_{\omega, a}^{-1}(\{0\})\,.$ An analogue of equation~(\ref{eqn:seminorms}) holds in this context as well. Before stating the next result, let us record an elementary fact for convenience.

\begin{lem}\label{lem:equi_sup}
    Let $\Omega$ be a compact Hausdorff space and let $\mathfrak{X}$ be a Banach space. If $\mathcal{F}\subset C(\Omega, \mathfrak{X})$ is a uniformly bounded, equicontinuous family of functions, then the assignment $\omega \mapsto \sup_{f\in \mathcal{F}} \|f(\omega)\|$ is continuous. 
\end{lem} We leave the proof as an exercise to the reader. With Lemma~\ref{lem:equi_sup} in hand, we aim to show the following.
\begin{lem}\label{lem:n-norm}
    Let $x\in \mathbb{Z}+\frac 12$ be a cut site and fix $\omega\in \Omega$. Let $\mathfrak{n}_\omega:\A_{x}^+ \to \mathbb{R}^+$ via 
    \[
        \mathfrak{n}_\omega (a) = \sup_{c\in (\A_{x}^+)_1} \rho_{\omega, c}(a) = \sup_{c\in (\A_x^+)_1} |\psi_\omega(c\otimes a)| \, .
    \] Then, $\mathfrak{n}$ is a subadditive, linearly homogeneous, positive semi-definite function which is constant on equivalence classes modulo $I_{x,\omega}$. Thus $\mathfrak{n}$  descends to a norm on $\B_{x,\omega}^+$. This norm, $\|\cdot \|_{\mathfrak{n}_\omega}$, is majorized by the quotient norm. Moreover, for all $a\in \A_{x}^+$, the mapping $\omega\mapsto \|[a,\omega]_x\|_{\mathfrak{n}_\omega}$ is continuous in $\omega$. 
\end{lem}
\begin{proof}
    It follows from the reverse triangle inequality that the seminorms defined in~(\ref{eqn:seminorms}) are constant on the equivalence classes in $\B_{x,\omega}^+$. 
    It follows easily that $\mathfrak{n}_\omega(a) = \mathfrak{n}_\omega(a+b)$ for any $b\in I_{x,\omega}$. Suppose now that $\mathfrak{n}_\omega(a)=0$. This means that $|\psi_\omega(c\otimes a)| = 0$ for all $c\in (\A_{x}^-)_1$. Thus, by linear homogeniety of the seminorms $\rho_{\omega,c}$, we obtain $a\sim_{x,\omega} 0$. Lastly, note that for any $b\in I_{x,\omega}$, and any $c\in \A_{x}^-$, 
    \[
        \rho_{\omega, c}(a) = \rho_{\omega, c}(a +b) = |\psi_\omega(c \otimes (a+b))| \le \|c\| \|a+b\|\,.
    \] Therefore, $\|[a]_{x,\omega}\|_{\mathfrak{n}_\omega} \le \|[a]_{x,\omega}^+\|_{\B_{x,\omega}^+}$ by definition of the quotient norm as an infimum.

       We aim to show that for fixed $a$, the family \[\mathscr{F}_{a}:=\{ g_c:\Omega \to \mathbb{R}^+ \text{ via } g_c(\omega) = |\psi_\omega(c\otimes a)| \colon c\in (\A_{x}^-)_1\}\,,\] is equicontinuous. To this end, let $\omega_0\in \Omega$ and $\epsilon >0$. Recall that for any Hermitian $h\in \A_x^-$, one has $-\|h\| \one \le h \le \|h\| \one$, since $\A_x^-$ is a unital $C^*$-algebra. Now, note that for any $a\in (\A_x)^+$, the linear functional $\psi_\omega((\, \cdot \,)\otimes a^*a)$ is positive on $\A_x^-$. Combining these two facts, we obtain
    \[
        -\|h\| \psi_{\omega}(\one \otimes a^*a) \le \psi_\omega(h\otimes a^*a) \le \|h\| \psi_\omega(\one \otimes a^*a)\,,
    \] for all $h\in (\A_x^-)_{\mathrm{s.a.}}$. Therefore, whenever $0\le a\in \A_x^+$ and $h\in (\A_x^-)_{s.a.}$, one has
    \begin{align*}
        |\psi_{\omega'}(h\otimes a) - \psi_{\omega}(h\otimes a) | \le \|h\| |(\psi_{\omega'}-\psi_{\omega})(\one \otimes a)|\,.
    \end{align*}
    
    In general, if $c\in \A_x^-$ is nonzero, we have that 
    \begin{align*}
        |\psi_{\omega'}(c\otimes a) - \psi_{\omega}(c\otimes a) |&\le |\psi_{\omega'}(\Re(c)\otimes a) - \psi_{\omega}(\Re(c)\otimes a) |+|\psi_{\omega'}(\Im(c)\otimes a) - \psi_{\omega}(\Im(c)\otimes a) |\\
        &\le (\|\Re(c)\| + \|\Im(c)\|) | (\psi_{\omega'}-\psi_{\omega})(\one \otimes a)|\\
        &\le 2\,\| c\|\,\cdot  | (\psi_{\omega'}-\psi_{\omega})(\one \otimes a)|\,.
    \end{align*}

    Lastly, recall that any $a\in (\A_x^+)$ can be decomposed into a sum of positive operators: $a = a_1 -a_2 + ia_3 - ia_4$  which once again is a general fact about functional calculus in $C^*$-algebras (See \cite{Murphy} or \cite[Section 2.2]{BratteliRobinsonI}). Therefore, one has the following estimate for any $c\in \A_x^-$ and any $a\in \A_x^+$
    \begin{equation}\label{eqn:f_fam_equi}
        |\psi_{\omega'}(c\otimes a) - \psi_{\omega}(c\otimes a) |\le \sum_{j=1}^4 |\psi_{\omega'}(c\otimes a_j) - \psi_{\omega}(c\otimes a_j) |\le 2\|c\| \sum_{j=1}^4 |(\psi_{\omega'} - \psi_{\omega})(\one\otimes a_j)|\,.
    \end{equation} Now, since $\psi_\omega$ is weakly*-continuous, each $|(\psi_\omega)(\one\otimes a_j)|$ is a continuous function $\Omega \mapsto [0,\infty)$. Therefore, we can choose a neighborhood $V\subset \Omega$ containing $\omega_0$, so that \[|(\psi_{\omega} - \psi_{\omega_0})(\one\otimes a_j)|< \frac{\epsilon}{8}, \quad \text{for all $\omega \in V$ and all $j=1,2,3,4$.}\]
    
    Hence, for all $\omega \in V$, we have 
    \[
        |\psi_{\omega}(c\otimes a) - \psi_{\omega_0}(c\otimes a) |< \epsilon \|c\|\,.
    \] Therefore, since the choices of $U_j$ did not depend on $c$ at all, it follows that $\mathscr{F}_a$  forms a bounded family of functions that is equicontinuous at $\omega_0$, thus equicontinuous everywhere since $\omega_0$ was arbitrary. Hence, \[
        \|[a]_{x,\omega}\|_{\mathfrak{n}_\omega} = \sup_{c\in (\A_x^-)_1} \rho_{\omega, c} (a) = \sup_{g\in \mathscr{F}_a}g(\omega),\] is continuous as a function of $\omega$ by Lemma~\ref{lem:equi_sup}. 
\end{proof}

\begin{rmk}
    Similarly, the $\sigma_{\omega, a}$ seminorms are constant on equivalence classes in $B_{x,\omega}^-$.
   \end{rmk}

\begin{rmk}\label{rmk:nondegenerate}
    Let $\eta_\omega: \B_{x,\omega}^- \times \B_{x,\omega}^+\to \C$ via $\eta_\omega([c,\omega]_x^-, [a,\omega]_x^+) := \psi_\omega(c\otimes a)$. This is a well-defined, nondegenerate, bilinear form by construction. Indeed, if $\eta_\omega([c,\omega]_x^-, [a,\omega]^+_x) = 0$ for all $a\in \B_{x,\omega}^+$, then, $\psi_\omega(c\otimes a) = 0$ for all $a$, thus, $c \sim_{x,\omega} 0 \in \B_{x,\omega}^-$. One also notes that since $\psi_\omega$ is constant on equivalence classes in $\B_{x,\omega}^{-}$ and $\B_{x,\omega}^+$ in the first and second coordinates respectively, the following holds:
        \[
            |\eta_\omega( [c,\omega]_x^-, [a,\omega]_x^+ )| =|\eta_\omega([c+d,\omega]_x^-, [a+b,\omega]_x^+)| \le \|c+d\| \|a+b\|\,,
        \] for all $d\in J_{x,\omega}$ and $b\in I_{x,\omega}$. Hence 
        \[
            |\eta_\omega( [c,\omega]_x^-, [a,\omega]_x^+ )| \le \|[c,\omega]_x^-\|_{\B_{x,\omega}^-} \|[a,\omega]_x^+\|_{\B_{x,\omega}^+}\,. 
        \] We note that $\eta_\omega([\one, \omega]_x^-, [\one,\omega]_x^+) \equiv 1$, so we conclude that $\|\eta_\omega\| = 1$ for all $\omega\in \Omega$. 
\end{rmk}

\begin{prop}\label{prop:n-norm_equals_quot}
    Let $x$ be a cut site and let $\mathfrak{n}_\omega : \A_x^+ \to \mathbb{R}^+$ be as in Lemma~\ref{lem:n-norm}. Whenever $\omega$ is such that $\B_{x,\omega}^+$ is finite dimensional, we have equality of the $\mathfrak{n}_\omega$ and quotient norms: $\|\cdot \|_{\mathfrak{n}_\omega} = \|\cdot \|_{\B_{x,\omega}^+}$. In particular, suppose that there is an open set $\Omega_0\subset \Omega$ so that for all $\omega_0\in \Omega_0$, the quotient space $\B_{x,\omega_0}^+$ is finite dimensional.  Then, for all $\omega\in \Omega_0$,  the mapping $\Omega_0 \ni \omega \mapsto \|[a,\omega]_x^+\|_{\B_{\omega, x}^+}$ is continuous for all $a\in \A_{x}^+$. 
\end{prop}
\begin{proof}
    Consider the bilinear form $\eta_\omega: \B_{x,\omega}^-\times \B_{x,\omega}^+ \to \C$ defined via $\eta_\omega([a,\omega]_{x}^-, [c,\omega]_x^+)= \psi_\omega(a\otimes c)$. Notice $\eta_\omega$ is nondegenerate for all $\omega \in \Omega_0$ by Remark~\ref{rmk:nondegenerate}. Hence, $\B_{x, \omega}^-$ is finite dimensional and canonically isomorphic to the dual $(\B_{x,\omega}^+)^*$ for all $\omega\in \Omega_0$.   
    
    Now, let $[a,\omega]_x^+\in \B_{x,\omega}^+$. By basic duality theory, we know there is some linear functional $\ell\in ((\B_{x,\omega}^+)^*)_1$, the unit  ball, so that $\ell([a,\omega]_x^+) = \|[a,\omega]_x^+\|_{\B_{x,\omega}^+}$. Therefore, since $\B_{x,\omega}^- \cong (\B_{x,\omega}^+)^*$, there is some $[c,\omega]_x^-\in \B_{x,\omega}^-$ with $\|[c,\omega]_x^-\|_{\B_{x,\omega}^-}\le 1$ so that $\ell = \eta_\omega([c,\omega]_{x}^-,\, \cdot\,)$. Thus, 
    \[
        \|[a,\omega]_x^+\|_{\B_{x,\omega}^+} = |\ell([a,\omega]_x^+)| = |\eta_\omega([c,\omega]_x^-, [a,\omega]_x^+)| = |\rho_{\omega,c}(a)|\le \|[a,\omega]_x^+\|_{\mathfrak{n}_\omega}\,,
    \] as claimed. Thus, we conclude the point-wise equality $\|\cdot \|_{\B_{x,\omega}^+} = \|\cdot \|_{\mathfrak{n}_\omega}$ for all $\omega\in \Omega_0$. If the second hypothesis holds, continuity of the mapping $\omega\mapsto \|[a]_{x,\omega}\|_{\B_{x,\omega}^+}$ in $\Omega_0$ is a consequence of the equality established above and the fact that $\omega\mapsto \|[a,\omega]_x^+ \|_{\mathfrak{n}_\omega}$ is continuous in $\Omega_0$ for all $a\in \A_{x}^+$.
\end{proof}
Some remarks are in order. 
\begin{rmk}
    The key step in order to conclude $\|\cdot \|_{\B_{x,\omega}^+} = \|\cdot \|_{\mathfrak{n}_\omega}$ was that $\ell$ is of a specific form, namely, that there exists some $c\in \A_{x}^-$ implementing $\ell$ via $\psi_\omega$. We used the fact that $\B_{x,\omega}^+$ is finite dimensional to conclude this via the canonical isomorphism. It may not be true in general that such an $\ell$ obtained by, e.g., the Hahn-Banach theorem can be implemented by such an element of $\A_{x}^-$.
\end{rmk}

\begin{rmk}
    In particular, when $\B_{x,\omega}^+$ is finite dimensional for all $\omega$, this means that for all $a\in \A_{x}^+$ the assignment $\omega \mapsto \|[a,\omega]_x^+\|_{\B_{x,\omega}^+}$ is continuous in $\omega$ by Lemma~\ref{lem:n-norm}. 
\end{rmk}

\begin{lem}\label{lem:cts_B_x}
    Suppose there is an open set $\Omega_0 \subset \Omega$ so that for every $\omega\in \Omega_0$, the quotient $\B_{x,\omega}^+$ is finite dimensional. In this case, given any $A\in C(\Omega, \A_x^+)$, the numerical mapping given by 
        \[
            \Omega_0 \ni \omega \mapsto \|[A(\omega),\omega]_x^+\|_{B_{x,\omega}}\,,
        \] is continuous. 
\end{lem}
\begin{proof}
    Again our strategy is to prove that a certain family of functions forms a bounded, equicontinuous family. Let $\omega_0 \in \Omega$ and note  for all $c\in \A_x^-$and all $\omega \in \Omega$ one has
    \begin{align*}
        |\rho_{c,\omega}(A(\omega)) - \rho_{c, \omega_0}(A(\omega_0))| &\le |\psi_\omega(c\otimes A(\omega)) - \psi_{\omega_0}(c\otimes A(\omega_0))|\\
        &\le  \|c\| \|A(\omega) - A(\omega_0)\|+|(\psi_\omega - \psi_{\omega_0})(c\otimes A(\omega_0))| \\
        &\le \|c\| \|A(\omega)-A(\omega_0)\| + \|c\| \sum_1^4 |( \psi_\omega - \psi_{\omega_0})\left[ \one \otimes A_j(\omega_0)\right]|\,,
    \end{align*} by the same argument preceding equation~(\ref{eqn:f_fam_equi}), where we have written $A(\omega_0)$ as its cartesian decomposition $A_1(\omega_0) - A_2(\omega_0) +iA_3(\omega_0) - iA_4(\omega_0)$. Therefore, by a similar argument to the one proceeding equation~(\ref{eqn:f_fam_equi}), the family $\mathcal{F}_A:= \{f_c(\omega) = |\psi_\omega(c\otimes A(\omega))| \colon c\in (\A_x^-)_1\}$ is equicontinuous, and therefore so is their supremum by Lemma~\ref{lem:equi_sup}. By Proposition~\ref{prop:n-norm_equals_quot}, we have 
    \[
        \sup_{c\in (\A_x^-)_1} \rho_{\omega, c}(A(\omega)) = \|[A(\omega),\omega]_x^+\|_{\mathfrak{n}(\omega)} = \|[A(\omega),\omega]_x^+\|_{\B_{x,\omega}^+}\,,
    \] for all $\omega\in \Omega_0$, and this assignment is continuous in $\omega$ there, by equicontinuity of $\mathcal{F}_A$. 
\end{proof}

\begin{thm}\label{thm:quotients_are_bbundled}
    Let $x$ be a cut site and suppose that there is an open subset $\Omega _0 \subset \Omega$ for which every $\B_{x,\omega}^+$ is finite-dimensional whenever $\omega\in \Omega_0$. Let $\mathscr{B}_{x} = \bigsqcup_{\omega \in \Omega_0} \B_{x,\omega}^+$ equipped with the surjection $\pi^x(\B_{x,\omega}^+)= \omega$. Then, taking $\Gamma(\mathscr{B}_x)= \mathcal{B}_x^+$ as the space of sections, $(\mathscr{B}_x, \pi^x)$ may be furnished with a unique topology carrying the structure of a Banach bundle over $\Omega_0$.  
\end{thm}
\begin{proof}
    Lemma~\ref{lem:cts_B_x} demonstrates that $\B_x^+$ fulfills the first requirement of Theorem~\ref{thm:Banach_Bundle}. For the next requirement, observe that $\{[A(\omega),\omega]_x^+ \colon A\in C_{\Omega}(\A_x^+)\}$ is indeed dense in $\B_{x,\omega}$ since the constant functions are continuous maps into $\A_x^+$. 
\end{proof}

\begin{note}
    In the case that  $\B_{x,\omega}^+$ is finite dimensional for all $\omega\in \Omega$, the bundle resulting from Theorem~\ref{thm:quotients_are_bbundled} will be written $\mathscr{B}_x = (\{\B_{x,\omega}\}_\omega, \pi^x, \Omega)$.
\end{note}

We now turn to a discussion of the action of the translation automorphism $\tau_1$ and of the Koopman operator $K_\vartheta$ on $\Gamma(\mathscr{B}_x)$ and $\B_{x,\omega}^+$. As we shall see later, the canonical map associated to $\tau_1$ is linked with the canonical map associated to the Koopman operator $K_\vartheta$ given by pre-composition with the measure-preserving ergodic map $\vartheta$.
 
First, note that for any cut site $x$, the automorphism $\tau_1$ restricts to an isomorphism of $\A_x^+$ onto $\A_{x+1}^+$. One can canonically extend the action of $\tau_1$ to the continuous functions via post-composition of the image. The resulting map is easily shown to be an isometry and in particular, repeating the procedure for $\tau_{-1}$ shows these maps to be bounded-ly invertible. Similarly, by defining $K_{\vartheta}:C(\Omega, \A_x^+) \to C(\Omega, \A_x^+)$ via $K_\vartheta:= A\circ \vartheta$, one obtains a bounded-ly invertible isometry \cite[Section 4.5]{Eisner_et_al}. Let us record this as a lemma.
 \begin{lem}
     Let $\tau_1$ denote the automorphism of translation on $\A_{\mathbb{Z}}$, and let $\vartheta:\Omega \to \Omega$ denote a continuous measure preserving ergodic map. Let $x$ be a cut site. The following hold:
     \begin{enumerate}
         \item Define $T_x:C(\Omega, \A^{+}_{x}) \to C(\Omega, \A^{+}_{x+1})$ via $T_x(A)(\omega):= \tau_1(A(\omega))$. Then, $T_x$ defines a linear isometry with bounded inverse.
         \item Define $K_{x}:C(\Omega, \A_{x}^+) \to C(\Omega, \A_{x}^+)$ via $K_x A := A\circ \vartheta$. The $K_x$ is a boundedly invertible linear isometry of $C(\Omega, \A_x^+)$.
       \item The following commutation relation holds 
        \[
        K_{x+1}T_x = T_x K_{x}\,.
       \]
     \end{enumerate} 
 \end{lem}
\begin{proof}
    Parts 1. and 2. are standard. To see the commutation relation in part 3., let $\omega \in \Omega$ and observe that \[(T_x K_x A)(\omega)= \tau_1(A(\vartheta \omega)) = K_{x+1}( \tau_1(A(\omega)) = (K_{x+1}T_x A)(\omega)\,.\qedhere\]
\end{proof}
\begin{rmk}
    Since $T_x$ and $K_x$ are both instances of the restriction of a linear map on $C(\Omega, \A_{\Z})$ to a specific subalgebra, we will slightly abuse notation and drop the subscripts when no confusion can arise. 
\end{rmk}

 Now, we shall study the action of $T$ and $K$ under the quotient. 
 
\begin{lem}\label{lem:bundlemap-covariance}
Let $x$ be a cut site. The following hold: 
\begin{enumerate}
    \item For each $\omega$, one has a canonical isometric isomorphism between $\B^+_{x,\vartheta(\omega)}$ and $\B^{+}_{x+1,\omega}$ given by the map $\tilde T_{x,\vartheta\omega}:[a,\vartheta \omega]_x^+\mapsto [\tau_1(a),\omega]_{x+1}^+$.
    \item  Let $T$ and $K$ be as above. There exists an isometric isomorphism $V_x:\B_x^+ \to \B_{x+1}^+$ so that the following diagrams commute 
    \begin{center}
        \begin{tikzcd}
            {C(\Omega,\A_x^+)} \arrow[r, "\sim"]\arrow[d, "T_x"] & {\B_x^+} \arrow[d, "V_x", dashed] \arrow[r, "p_{x,\vartheta \omega}"]& \B_{x, \vartheta \omega}^+\arrow[d, "\tilde T_{x,\theta \omega}"] \\
            {C(\Omega, \A_{x+1}^+)}\arrow[r, "\sim"]& {\B_{x+1}^+} \arrow[r, "p_{x+1,\omega}"]& \B_{x+1, \omega}^+
        \end{tikzcd}\,.
    \end{center}
\end{enumerate}

\end{lem} 
\begin{proof}
    1.  Define the linear map $\tilde T_{x,\vartheta\omega} [a,\vartheta \omega]_x^+=[\tau_1(a),\omega]_{x+1}^+$. To check that this is well defined, suppose $[a,\vartheta \omega]_x^+ = 0$. Then, for all $c\in \A_{x}^-$, we have that 
        \begin{equation}\label{eqn:theta_tau_transfer}
            0 = \psi_{\vartheta \omega}(c\otimes a)= \psi_\omega( \tau_1(c\otimes a) )= \psi_\omega( \tau_1(c) \otimes \tau_1(a))\,.
        \end{equation} Since $\tau$ is an automorphism, in particular we have the above holding for all $d\in \A_{x+1}^-$. Hence $[\tau_1(a)]_{x+1, \omega}=0$ as required. 

    It is not hard to argue that equation~(\ref{eqn:theta_tau_transfer}) implies $I_{x+1,\omega} = \tau_1(I_{x, \vartheta \omega})$. Therefore, 
    \begin{align*}
        \|\tilde T_{x, \vartheta \omega}[a,\vartheta \omega]_x^+ \|_{\B_{x+1, \omega}^+}&= \inf_{b\in I_{x+1,\omega}} \|\tau_1(a) + b\|_{\A_{x+1}}\\
        &= \inf_{\tau_{-1}(b) \in I_{x, \vartheta \omega}}\|a+ \tau_{-1}(b)\|_{\A_x}= \|[a,\vartheta \omega]_x^+\|\,.
    \end{align*} Thus,  $\tilde T_{x,\vartheta \omega}$ is an isometry. One can run this same argument with the canonical map associated to $\tau_{-1}$ to see that $\tilde T_{x, \vartheta \omega}$ has a bounded inverse.

    2. Let $V_x[A]_x = [KTA]_{x+1}$. To see this is well defined, let $[A]_x^+ = 0$. For any $C^{-}_{x+1}\in C(\Omega, \A_{x+1}^{-})$ and for all $\omega\in \Omega$, we have 
    \begin{align*}
        \psi_\omega( C^{-}_{x+1}(\omega) \otimes KTA(\omega)) &= \psi_\omega( C^{-}_{x+1}(\omega) \otimes \tau_1(A(\vartheta \omega))) = \psi_\omega\left( \tau_1\left(\tau_{-1}(C_{x+1}^{-}(\omega)) \otimes A(\vartheta \omega) \right) \right)\\
        &= \psi_{\vartheta \omega}\left( \left(\tau_{-1}(C_{x+1}^{-}(\omega)) \otimes A(\vartheta \omega) \right) \right)= \psi_{\omega'}(D(\omega') \otimes A(\omega')) = 0,
    \end{align*} where we have set $\omega' = \vartheta \omega$, and $D(\omega) = T^{-1}K^{-1}C^{-}_{x+1}(\omega)\in C(\Omega, \A^{-}_{x})$. Notice that the final equality in the above holds for all $\omega ' \in \Omega$, and since $\vartheta$ is a homeomorphism, we conclude that $[KTA]_{x+1}^+=0$ as required. Equivalently, $I_x \supset KTI_{x+1}.$ Using essentially the same argument as in Lemma~\ref{lem:quot_eval}, we conclude that $V_x$ is bounded. Running this same argument with $K^{-1}T^{-1}$, we see that in fact, $V_x$ has a bounded inverse as well.

Now, let $[A]_{x}^+ \in \B_x^+$. We have \[
\tilde T_{x, \vartheta \omega}\,\circ\, p_{x, \vartheta \omega}[A]_x^{+} = [\tau_1(A(\vartheta \omega)]_{x+1, \omega}^+ = p_{x+1, \omega} [KTA]_{x+1}^+ = p_{x+1, \omega} V_x [A]_x^+,\] as required. 
\end{proof}

\begin{cor}\label{cor:small_corrs}
    If $\B_{x,\omega}^+$ is finite dimensional for all $\omega\in \Omega$, then every $\B_{x+1,\omega}^+$ is also finite dimensional. 
\end{cor}

\begin{rmk}\label{rmk:normed_quotient_bundle}
    A natural question is whether or not our construction of a Banach bundle out of the quotient spaces $\B_{x,\omega}$ will hold if we relax the condition that $\B_{x,\omega}$ be finite dimensional. This is not clear since Lemma~\ref{lem:n-norm} requires us to produce a linear functional $\ell$ of a specific form, namely, $\ell = \eta_{\omega}([c,\omega], \cdot)$ in order to show that the $\mathfrak{n}_\omega$ norm equals the quotient norm. In infinite dimensions, the Hahn-Banach theorem could be used here, but it will not guarantee that $\ell$ is of the correct form, and it is not clear to us at this time how to verify continuity in $\omega$ of the norms $\B_{x,\omega}$ which is required for the Fell-Doran theorem. 

    However, there is a remarkably similar construction in the functional analysis literature in \cite{ResendeSantos}. A careful review of this paper shows that if we relax the finite-dimensionality requirement, then $\{\B_{x,\omega} \colon \omega \in \Omega\}$ can be topologized to form a continuous normed quotient vector bundle. This is more generality than we need, so we will tend to other results. 
\end{rmk}

\subsection{The Fundamental Theorem}

Let us recall the notation and set-up of the previous section. Fix a cut site $x\in \mathbb{Z}+ \frac12$, and consider the family of Banach spaces $\{\B_{x,\omega}^+\colon \omega \in \Omega\}$ defined fiber-wise by the equivalence relation $a\sim_\omega b$ if and only if $\psi(c\otimes (a-b)) = 0$ for all $c\in \A_x^-$. Consider also the family of selections $\B_x^+:=\{ \omega \mapsto [A(\omega), \omega]_x\}$ where $A\in C(\Omega,\A_x^+)$, which forms a linear space. Theorem~\ref{thm:Banach_Bundle} and Corollary~\ref{cor:small_corrs} motivate the following definition: 

\begin{define}\label{def:small_correlations}
    Let $\Omega$ be a compact Hausdorff topological space equipped with an homeomorphism $\vartheta$. Let $\psi_\omega$ a weakly* continuous mapping into the state space of the quasi-local algebra $\A_{\mathbb{Z}}$ that is translation co-variant with respect to $\vartheta$. If for every $\omega\in \Omega$, each quotient space $\B_{x,\omega}^+$  at the cut site $x\in \mathbb{Z}+\frac{1}{2}$, is finite dimensional, we say that $\psi_\omega$ \textbf{has small correlations}. 
\end{define} 

An example of such a TCVS is as follows.

\begin{exmp}\label{exmp:covariant_product_state}    
Let $\xi_\omega$ be the weakly*-continuous covariant product state from Proposition~\ref{prop:existence_TCPS}. Then, $\xi_\omega$ has small correlations with fiber dimension 1. 
\end{exmp}
\begin{proof}
    To see this, $a,b\in \A_x^+$. We note that if $a\sim_{x,\omega}b$, then, $\psi_\omega(a-b) = 0$ and so $I_{x,\omega} \subset \ker \psi_\omega|_{\A_{x}^+}$. Conversely, if $a\in \ker \psi_{\omega}|_{\A_{x}^+}$, then since $\psi_\omega$ factors over the bond at $x$, we get $\psi_\omega(c\otimes a) = \psi_\omega(c) \psi_{ \omega}(a)=0$ for all $c\in \A_{x}^-$ for all $\omega$. So $I_{x,\omega}$ is the kernel of a linear functional and is therefore of co-dimension one by general functional analysis, hence $\B_{x,\omega}^+ \cong \C$ for all $\omega$. 
\end{proof}

Several remarks are in order. 

\begin{rmk}
    The analysis in Example~\ref{exmp:covariant_product_state} works for any weakly* continuous function for which the expectation $\psi_\omega(a\otimes b) = \psi_\omega(a) \psi_{\omega}(b)$ for all $a\in \A_{x}^-$, all $b\in \A_{x}^+$ and all $\omega \in \Omega$ (i.e. a state with no correlations across the bipartition of the chain at $x$)
\end{rmk}
\begin{rmk}\label{rmk:dimension_may_vary}
    In general, we do not expect that the Bundle fibers are all equi-dimensional as in Example~\ref{exmp:covariant_product_state}. However, we show in Lemma~\ref{lem:cut_independent}  below, that the Banach bundle describing the correlations is independent of the cut-site up to isomorphism.
\end{rmk}

\begin{rmk}\label{rmk:canonical_bundle}
    If $\psi_\omega$ has small correlations across $x$, then the Banach spaces $\{B_{x,\omega}^+\}_{\omega\in \Omega}$ equipped with the linear space of sections given by $\{\omega \mapsto [A(\omega), \omega)]_x\colon A\in C_{\Omega}(A_x^+)\}$ defined in Lemma~\ref{lem:quot_eval} form a Banach bundle over $\Omega$. We denote this bundle by $(\mathcal{C}, \pi, \Omega)$ (reference to a specific cut site is not necessary due to Remark~\ref{rmk:dimension_may_vary}), and we refer to it as the \textbf{small correlation bundle}. We shall see in Theorem~\ref{thm:minimal} below that the use of the definite article is justified since $\mathcal{C}$ is minimal with respect to certain natural conditions. 
\end{rmk}

We now show a TCVS analog of \cite[Proposition 2.1]{FannesNachtergaeleWerner} in Theorem~\ref{thm:Fundamental} below. We split the proof into a few easier but technical lemmas first. 
\begin{lem}\label{lem:pre-fundamental}
    Let $\mathcal{A}$ be a unital $C^*$-algebra and let $\A_{\mathbb{Z}}$ be the associated quasilocal algebra. Let $\Omega$ be a compact Hausdorff space equipped with an homeomorphism $\vartheta$. Let $\psi_\omega$ be a weakly* continuous, translation covariant state on $\A_{\mathbb{Z}}$ with small correlations. Let also $\mathscr{B}_x$ be the Banach bundle at the cut site $x$. Fix $a\in \mathcal{A}$ and $\omega\in \Omega$. Define the mapping 
    \begin{equation}\label{eqn:canonical_transfer}
        \begin{split}
            E^x_{a,\omega}&:\B_{x,\omega}^+ \to \B_{x,\vartheta^{-1}\omega}^+ \text{ via }\\
            &[b,\omega]_x \mapsto [a\otimes \tau_1(b), \vartheta^{-1}\omega]_x\,.
        \end{split}
    \end{equation}
    The following properties hold. 
    \begin{enumerate}[label = (\alph*)]
        \item For each $a\in A$ and $\omega\in \Omega$, the $E^x_{a,\omega}$ mapping is well-defined bounded linear map on the fiber with 
        \begin{equation}\label{eqn:transfer_norm_est}
            \|E^x_{a,\omega}[b,\omega]\| \le \|a\|_A \|[b,\omega]_x\|\,.
        \end{equation} 

        \item For all $a,b\in A$, the mapping 
            \begin{equation}
                \omega \mapsto E^x_{a,\omega}[b,\omega]_x \in B_{x,\vartheta^{-1}\omega}\,,
            \end{equation} is continuous in the (unique) Banach bundle topology. Thus, $E_a \in \hom^\vartheta(\mathscr{B}_x, \mathscr{B}_x)$.
        \item There exists a continuous section $e^x\in \Gamma(\mathscr{B}_x)$ and a homomorphism $\varrho^x\in \hom(\mathscr{B}_{x}, \Omega \times \C)$ so that for any pure tensor $a\otimes b\otimes \cdots \otimes c\in \A_{[x, x+n]}$, one has
        \begin{equation}\label{eqn:transfer_factorization}
            \psi_\omega(a\otimes b\otimes \cdots \otimes c) = \varrho^x_\omega \circ E^x_{a,  \omega}\circ E^{x}_{b,\vartheta \omega}\circ \cdots \circ E^x_{c, \vartheta^{n}\omega}(e^x(\vartheta^n\omega))\,. 
        \end{equation} 
    \end{enumerate}
\end{lem}
\begin{proof}
$(a)$: Linearity in the argument is obvious, so let us check this map is well-defined. Let $c\in \A_{x}^-$ and note that 
    \[
        \psi_{\vartheta^{-1}\omega}(c \otimes a \otimes \tau_1(b)) = \psi_{\vartheta^{-1}\omega}(\tau_1(\, \tau_{-1}(c\otimes a) \otimes b\,)) = \psi_{\omega}( \tau_{-1}(c\otimes a) \otimes b)\,.
    \] Thus, if $[b,\omega]_x = 0$, its image $E^x_{a,\omega}[b,\omega]_x\in \B_{x, \vartheta^{-1}\omega}^+$ is also zero. To show the norm estimate, let $d\in I_{x,\omega}$, and note that 
    \[
        \|E^{x}_{a,\omega}[b,\omega]_x\|_{\B_{x,\vartheta^{-1}\omega}} = \|E^{x}_{a,\omega}[b+d,\omega]_x\|_{\B_{x,\vartheta^{-1}\omega}} = \inf_{f\in I_{x,\vartheta^{-1}\omega}}\| a\otimes \tau_1(b+d) + f\|_{\A_{x}^+}\le \|a\| \|b+d\|_{\A_{x}^+}\,,
    \] since $\tau_1$ is an automorphism. Taking the infimum in $d$, over both sides yields~(\ref{eqn:transfer_norm_est}). 

    $(b)$: Let $(\omega_{\lambda})_{\lambda\in \Lambda}$ be a net in $\Omega$ which converges to $\omega'$. We need to show that $E^x_{a, \omega_{\lambda}}[b, \omega_{\lambda}]_x \to E^x_{a, \omega'}[b, \omega']_x$ in the bundle $\mathscr{B}_x$. To begin, notice that by the Stone-Weierstra{\ss}   theorem for bundles combined with Theorem~\ref{thm:Banach_Bundle}, the linear space of sections given by 
    \[
        \B_x^+ := \{ \omega\mapsto [A(\omega), \omega]_x \colon A\in C(\Omega,\A_{\Z})\}\,,
    \] is dense in $\Gamma(\mathscr{B}_x)$. Therefore, by Lemma~\ref{lem:bundle_convergence} $(iii)$, it is sufficient to show that the base-points converge and that for every $A\in C(\Omega,\A_{\Z})$, the net of norms $\|E^x_{a, \omega_{\lambda}}[b, \omega_{\lambda}]_x - [A(\omega_\lambda), \omega_\lambda]_x\| \to \|E^x_{a, \omega'}[b, \omega']_x - [A(\omega'), \omega']_x\|$. We note $\pi(E^x_{a,\omega_{\lambda}}[b, \omega_{\lambda}]_x ) = \vartheta^{-1}(\omega_\lambda) \to \vartheta^{-1}(\omega') = \pi(E^x_{a,\omega'}[b, \omega']_x)$ since $\vartheta$ is an homeomorphism. 

    Now, set $\tilde A(\omega) = a\otimes \tau_1(b) - A\circ \vartheta^{-1}(\omega_{\lambda}) \in C(\Omega,\A_{\Z}).$ By Lemma~\ref{lem:n-norm}, and the assumption that $\psi_\omega$ has small correlations, we have 
    \[
        \Omega \ni \omega \mapsto \|[\tilde A(\omega), \omega]_x\|_{B_{x,\omega}^+}\,,
    \] is continuous. One computes
    \begin{align*}
        \|E^x_{a, \omega_{\lambda}}[b, \omega_{\lambda}]_x - [A(\omega_\lambda), \omega_\lambda]_x\|&= \|[a\otimes \tau_1(b) - A\circ \vartheta^{-1}(\omega_{\lambda}), \vartheta^{-1}(\omega_\lambda)]_x\|\\
        &= \|[\tilde A(\omega_{\lambda}), \vartheta^{-1}\omega_{\lambda}]_x\|_{B_{x,\vartheta^{-1}\omega_{\lambda}}^+}\\
        &\to \|[\tilde A(\omega'), \vartheta^{-1}\omega']_x\| = \|E^x_{a, \omega'}[b, \omega']_x - [A(\vartheta^{-1}\omega'), \vartheta^{-1}\omega']_x\|\,,
    \end{align*} concluding the proof. 

$(c)$: Take $e^x(\omega) = [\one_{\A_{x}^+}, \omega]_x$ and $\varrho_\omega[b,\omega]_x = \psi_\omega(\one_{\A_x^-} \otimes b)$. The section $e^x$ is continuous by definition of the topology on $\mathscr{B}_x$. Notice also that $\varrho^x_\omega[A(\omega), \omega] = \psi_\omega(\one\otimes A(\omega))$ is continuous by Lemma~\ref{lem:weak*_cts_on_C(A)}. To show~(\ref{eqn:transfer_factorization}), let us first start with two uncorrelated observables $a\otimes b\in \A_{[x, x+1]}$. One has 
\begin{align*}
    \psi_\omega(\one_{\A_x^{-}} \otimes a \otimes b) &= \varrho^x_{\omega}[a\otimes b, \omega]_x= \varrho^x_\omega[a\otimes \tau_{1}( \tau_{-1}(b)), \vartheta^{-1} \vartheta \omega]_x\\
    &= \varrho^x_\omega \circ E_{a,\omega}[\tau_{-1}(b), \vartheta \omega]_x= \varrho^x_\omega \circ E_{a,\omega}[\tau_1(\tau_{-1}(\tau_{-1}(b))), \vartheta^{-1} \vartheta(\vartheta \omega)]_x\\
    &= \varrho^x_\omega E_{a,\omega}^x \circ E_{b, \vartheta \omega}^x[\one_{A_{x}^+}, \vartheta^2\omega]_x = \varrho^x_{\omega} \circ E^x_{a,\omega} \circ E^x_{b, \vartheta \omega}(e^x(\vartheta^2\omega))\,.
\end{align*} Note that since $b$ is supported on a single site $\tau_{-2}(b\otimes \one_{\A_{x+2}^+}) = b \otimes \one_{\A_{x}^+}$, from which the last equality follows. The general formula now follows by induction.
\end{proof}

\begin{lem}[Freedom from a specific cut site]\label{lem:cut_independent}
    Suppose that $\psi_\omega$ is a TCVS with small correlations. Then, for any cute site $x\in \mathbb{Z}+\frac12$ there is an isometric Banach bundle isomorphism $\tilde T$ between $\mathscr{B}_x$ and $\mathscr{B}_{x+1}$ which is covariant in $\vartheta$. Moreover the transfer operators defined in~(\ref{eqn:canonical_transfer}) obey
    \begin{equation}\label{eqn:transfer_automorphism_covariance}
        E_{a,\omega}^x = \tilde T^{-1} \circ E_{a,\vartheta^{-1}\omega}^{x+1} \circ \tilde T\,.
    \end{equation}
\end{lem}
\begin{proof}
    Recall that possessing small correlations is enough to guarantee that the family $\{B_{x,\omega}^+\}_{\omega\in \Omega}$ forms a Banach bundle $(\mathscr{B}_x, \pi^x, \Omega)$ with a dense $C(\Omega)$-sub-module of sections given by the vector space of functions $\{\omega \mapsto [A(\omega), \omega]_x \colon A\in C(\Omega,\A_x^+)\}\subset \Gamma(\mathscr{B}_x)$. By Corollary~\ref{cor:small_corrs}, we in fact have a family of Banach bundles indexed by the half-integers.

    Let $\tilde T[b,\omega]_x = [\tau_1(b), \vartheta^{-1}\omega]_{x+1}$ as in Lemma~\ref{lem:bundlemap-covariance}. Notice that $\pi^{x+1}\circ \tilde T = \vartheta^{-1} \circ \pi^x$. Thus, $\tilde T$ preserves convergence of the base-points since $\vartheta$ is an homeomorphism. We have already seen that $\tilde T$ is fiber-wise a linear isometry, so all that needs to be done is to show that it is an homeomorphism in the bundle topology. This is immediate by observing that 
    \[
        \|[b,\omega]_{x} - [A(\omega), \omega]_{x}\| = \|\tilde T [ b-A(\omega), \omega]_x\| = \|[b,\vartheta^{-1}\omega]_{x+1} - [(K_{\vartheta}A)(\vartheta^{-1}\omega), \vartheta^{-1} \omega]_{x+1}\|\,.
    \]Since the Koopman operator $K_{\vartheta}$ is a $*$-isomorphism, it follows that $\{\omega \mapsto [K_{\vartheta}A(\omega), \omega]_x \colon A\in C(\Omega,\A_x^+)\}$ is a dense $C(\Omega)$ sub-module of $\Gamma(\mathscr{B}_{x+1})$. Taking these observations together is enough to guarantee that whenever $(b_\lambda) \subset \mathscr{B}_x$ is a convergent net, then so is $(\tilde Tb_\lambda)\subset \mathscr{B}_{x+1}$ by Lemma~\ref{lem:bundle_convergence}. Equation~(\ref{eqn:transfer_automorphism_covariance}) follows from straightforward computation. 
\end{proof}

\begin{note}
    In view of this freedom of choice, we will henceforth drop the dependence on a choice of cut site and instead write all  formulae which follow with respect to the cut $x = -1/2$, so that $\A_x^+ = \A_{\mathbb{N}\cup \{0\}}=:\A^+$ and $\A_x^- = \A_{-\mathbb{N}}:= \A^-$. We will also reserve the letters $\varrho = \varrho^{-1/2}$, $E:=E^{-1/2}$, and $e^{-1/2}:=e$ for the family of objects defined as in Lemma~\ref{lem:bundlemap-covariance} on the canonical bundle $(\mathcal{C}, \pi, \Omega)$. 
\end{note} 

 In \cite{FannesNachtergaeleWerner}, the transfer operators are defined in a slightly different way by considering a map $E$ defined on a tensor product of the on-site $C^*$-algebra $\mathcal{A}$ and the space of correlations $B$. This allows them to define iterates on successive tensor powers of $\mathcal{A}$. Since we prefer point-wise notation for our transfer operators, it is not as obvious that such an extension exists but the repairs are only minimal as the following lemma shows. 

\begin{lem}[Iterates of Transfer Operators]\label{lem:iterate_hom}
    Let $\mathcal{A}$ be a unital $C^*$-algebra and $\A_\Z$ the associated quasilocal algebra. Let $\mathcal{C}$ denote the canonical bundle as in Remark~\ref{rmk:canonical_bundle}. Given a family of transfer operators $T_1:=\{E_a:a\in \mathcal{A}\}\subset \hom^{\vartheta^{-1}}(\mathcal{C}, \mathcal{C})$, and a positive number $p\in \mathbb{N}$, and any finite sum of pure tensors $s:=\sum_{j}a_1^{(j)} \otimes \cdots \otimes a_p^{(j)} \in \mathcal{A}^{\otimes p}$, 
        \begin{equation}
            \|\sum_j E_{a_1^{(j)}}\circ \cdots \circ E_{a_p^{(j)}}\|\le \|\sum_{j}a_1^{(j)} \otimes \cdots \otimes a_p^{(j)}\|_{\min}\,.
        \end{equation} Thus, there exists a family of transfer operators $\{E^{(p)}_{b} : b\in \mathcal{A}^{\otimes p}\} \subset \hom^{\vartheta^{-p}}(\mathcal{C}, \mathcal{C})$ extending length $p$ compositions $E_{a_1}\circ \cdots \circ E_{a_p}$ of elements of $T_1$. 
\end{lem}
\begin{proof}
    We have already observed that the transfer operators are linear in $a$. It is clear that if the estimate holds, then for any $d\in \mathcal{A}^{\otimes p}$, one can find a sequence of iterates of the form $s$ as above converging to $d$ in $\|\cdot \|_{\min}$, and the resulting operator agrees with $[b,\omega] \mapsto [d\otimes b, \vartheta^{-p}\omega]$. So the proof reduces to showing that the estimate holds. We will show this in the case that $p=2$, and from there the proof follows easily by induction. 

    Let $v:= \sum_j a_j \otimes b_j \in \mathcal{A}\otimes_{\min} \mathcal{A}$, and consider the operator $E^{(2)}_{v, \omega} := \sum_{j} E_{a_j, \vartheta \omega} \circ E_{b_j, \omega}$. We see that for any $[x,\omega] \in \mathcal{C}$, one has
        \begin{align*}
            \|E^{(2)}_{v, \omega}[x, \omega]\| &= \|\sum_{j} E_{a_j, \vartheta \omega} \circ E_{b_j, \omega}[x,\omega] \|= \|\left[\sum_j a_j \otimes b_j \otimes x, \vartheta^{-2}\omega \right] \|\\
            &\le \|\sum_j a_j \otimes b_j\| \|[x, \vartheta^{-2}\omega]\|= \|v\|_{\min} \|[x, \vartheta^{-2}\omega]\|.
        \end{align*} To obtain the last equality, recall that in the expression $a\otimes \tau_1(b)$ we have identified $a$ with its image under the faithful representation $\iota_{0}:\mathcal{A} \to \A_\Z$ via $a\mapsto \one_{(-\infty, -1]} \otimes a \otimes \one_{[1, \infty)},$ and similarly with $b$ and $\iota_1$. Recalling that the value of the minimal tensor norm does not depend on the choice of faithful representation (\cite[Proposition 3.3.11]{BrownOzawa}), we obtain 
        \[
            \| \sum_j a_j \otimes b_j \|_{\A_\Z} = \| \sum_j a_j \otimes b_j\|_{\min} = \|v\|_{\min},
        \] completing the proof. 
\end{proof}

\begin{thm}[Minimality]\label{thm:minimal}
    Suppose that $\psi_\omega$ is a TCVS with small correlations.  Let $(\mathcal{C}, \pi, \Omega)$ be the canonical correlation bundle defined in Remark~\ref{rmk:canonical_bundle} with its native transfer apparatus $(\varrho, E,e)$. Let $(\mathcal{E}, \pi', \Omega)$ be any other Banach bundle satisfying the following:
            \begin{enumerate}[label = (\roman*)]
                \item Each fiber is nontrivial and has finite dimension, 
                \item There is:
                    \begin{enumerate}[label = (ii.\arabic*)]
                        \item a distinguished continuous section $f\in \Gamma(\mathcal{E})$,
                        \item a homomorphism $\varphi \in \hom(\mathcal{E}, \Omega \times \C)$, with $\varphi_\omega(f(\omega)) \equiv 1$ and $\sup_{\omega \in \Omega}\|f_\omega\|\le 1$,
                        \item and, there is a family of \textbf{transfer operators}
        \[
            \{F_{a,\omega}:\mathcal{E}_\omega \to \mathcal{E}_{\vartheta^{-1}\omega}\colon a\in \mathcal{A}, \omega \in \Omega\}\subset \hom^\vartheta(\mathcal{E}, \mathcal{E})
        \] which are bounded and linear in $a\in \mathcal{A}$ and $e\in \mathcal{E}_\omega$, and satisfy $F_{\one, \omega}(f(\omega)) = f(\vartheta^{-1}\omega)$ and $\varphi_{\vartheta^{-1}\omega} \circ F_{\one, \omega} = \varphi_\omega$. 
                    \end{enumerate} so that whenever $n<m$ we have 
        \begin{equation}\label{eqn:transfer_expectation}
            \psi_\omega(a_n \otimes a_{n+1} \otimes \cdots \otimes a_m) = \varphi_{\vartheta^{-(n+1)}\omega}\circ F_{a_n,\vartheta^n \omega}\circ \cdots \circ F_{a_m, \vartheta^m \omega}( f(\vartheta^m \omega))\,.
        \end{equation}\end{enumerate}

    Then the following two conditions are equivalent
    \begin{enumerate}[label = (\alph*)]
        \item There exists a $\vartheta^{-1}$-co-variant isometric isomorphism $\Phi: \mathcal{C} \to \mathcal{E}$. 
        \item For all $\omega$, the span of the transfer operators $\spn\{F_{a,\omega} \circ F_{b,\vartheta\omega} \circ \cdots \circ F_{c, \vartheta^n}\circ f(\vartheta^{n}\omega )\}\subset \mathcal{E}_{\vartheta^{-1}\omega}$ and, $\spn\{\varphi_{\vartheta^{-(n+1)}\omega}\circ F_{a,\vartheta^{-n} \omega} \circ F_{b,\vartheta^{-(n-2)}\omega} \circ \cdots\circ  F_{c, \vartheta^{-1}\omega}\}\subset \mathcal{E}^*_{\vartheta^{-1}\omega}$ are dense for all $\omega \in \Omega$. 
    \end{enumerate}
Moreover, in the case that $(a)$ or $(b)$ holds, the transfer operators satisfy: 
    \begin{equation}
        F_{a,\omega} = \Phi \circ E_{a, \omega} \circ \Phi^{-1}\,.
    \end{equation}

\end{thm}
\begin{proof}
    $(a)\Rightarrow (b)$ follows from the definition of a co-variant isometric isomorphism.  To see $(b)\Rightarrow (a)$, define $\Phi([a, \omega]) = F_{a,\omega}(f(\omega))$. In particular, we have that $\Phi(e(\omega)) = \Phi([\one, \omega]) = f(\omega)$. To check that $\Phi$ is well-defined, suppose $[a,\omega] = 0$. In which case, for all $X^-\in \mathcal{A}^-$,
        \begin{align*}
            0 = \psi_\omega(X^- \otimes a)\, .
        \end{align*} In particular, this holds for all $x_{-n} \otimes \cdots \otimes x_{-1}$. But by the transfer conditions, this means
        \[
            0  = \varphi_{\vartheta^{-(n+1)}\omega}\circ F_{x_{-n}, \vartheta^{-n}\omega} \circ F_{x_{-1}, \vartheta^{-1}} \circ F_{a,\omega}(f(\omega))= \varphi_{\vartheta^{-(n+1)}\omega}\circ F_{x_{-n}, \vartheta^{-n}\omega} \circ F_{x_{-1}, \vartheta^{-1}} \circ \Phi([a,\omega])\,. 
        \] By linearity, the above then holds for a dense subspace of $\mathcal{E}^*_{\vartheta^{-1} \omega}$, so it must be that the  $\Phi([a,\omega]) = 0_{\vartheta^{-1}\omega}$, which denotes the zero vector in $\mathcal{E}_{\vartheta^{-1}\omega}$. 
        
        Now, let $a_0 \otimes \cdots \otimes a_n \in \mathcal{A}^+$, and observe that for any $d\in I_{0,\omega}\cap \mathcal{A}^{\loc}_\Z$ (i.e., $d\sim_{0,\omega} 0$)
        \[
            \|\Phi([a_0\otimes \cdots \otimes a_n, \omega])\| = \|\Phi([a_0\otimes \cdots \otimes a_n +d, \omega])\|= \|F_{a + d, \omega} (f(\omega)) \| \le \|F_{a + d, \omega}\| \le \|a+d\|_{\mathcal{A}_\Z}\,.
        \] Whence, 
        \[
            \|\Phi([a_0\otimes \cdots \otimes a_n, \omega])\|\le \|[a_0 \otimes \cdots \otimes a_n, \omega]\|_{\mathcal{C}_\omega}\,,
        \] since $\mathcal{A}^{\loc}_{\Z}$ is dense and $I_{0,\omega}$ is a closed subspace of $\mathcal{A}_\Z$. Therefore, $\Phi$ extends to a linear map on all fibers in $\mathcal{C}$. By construction, one sees that $\Phi$ has a fiber-wise dense range in $\mathcal{E}$. 
    
        Moreover, $\Phi$ is an injection. To see this, suppose that $\Phi([a_0 \otimes \cdots \otimes a_n, \omega])= 0_{\vartheta^{-1}\omega}$. By the transfer property, this means that 
        \[
            0 = \varphi_\omega \circ F_{a_0, \omega}\circ \cdots \circ F_{a_n, \vartheta^n \omega} \circ f(\vartheta^n \omega) = \psi_\omega(a_0\otimes \cdots \otimes a_n)\,.
        \] By Lemma~\ref{lem:n-norm}, this means, 
        \[
            \|[a_0 \otimes \cdots \otimes a_n, \omega ]\| = \sup_{c\in (\mathcal{A}^-)_1} |\psi_\omega (c\otimes a_0 \otimes \cdots \otimes a_n)| \le 8|\psi_\omega (\one \otimes a_0 \otimes \cdots \otimes a_n)| = 0\,,
        \] where the inequality is gotten by appealing to the Cartesian decomposition of operators in a $C^*$-algebra as in the argument leading up to inequality~(\ref{eqn:f_fam_equi}) in Lemma~\ref{lem:n-norm}. In fact, by working a little harder we can see that $\Phi$ is an isometry: appealing again to Lemma~\ref{lem:n-norm}, we have 
    \begin{align*}
        \|[b,\omega]\| &= \sup_{c\in (\A^-)_1}|\psi_\omega(c\otimes b)| = \sup_{d\in D\subset (\A^-)_1} |\varphi_{\vartheta^{\min\{\supp(d)\}-1}\omega} \circ F^{(|\supp(d)|)}_{d, \omega} \circ F_{b,\omega}(f(\vartheta \omega))| \\&= \sup_{g\in \mathcal{E}_{\vartheta^{-1}\omega}^*} |g(F_{b,\omega}(f(\omega)))| = \|F_{b,\omega}(f(\omega))\|\,, 
    \end{align*} where by $D$ we denote intersection of $\mathcal{A}^{\loc}_{\Z}$ and $(A^-)_1$, which is dense in $(\mathcal{A}^-)_1$; and by $E^{(|\supp(d)|)}_{d, \omega}$ we mean the iterated transfer operator as in Lemma~\ref{lem:iterate_hom}.

        We need to check that $\Phi$ is an homeomorphism. This follows since $\Phi$ takes a dense submodule of $\Gamma(\mathcal{C})$ to the dense submodule induced by the transfer operators in $\Gamma(\mathcal{E})$. Indeed, let $\Gamma_0$ be the family of sections gotten by $\spn_{C(\Omega)}\{E_{a, \omega} \circ \cdots \circ E_{b, \vartheta^n \omega}\circ e\}$. Then, $\Phi(\Gamma_0) = \spn_{C(\Omega)}\{ \omega \mapsto F_{a,\omega}\circ \cdots \circ F_{b, \vartheta^n \omega}\circ f(\omega)\}$ which is fiber-wise dense in $\mathcal{E}$ by assumption. or equivalently, it is a dense $C(\Omega)$-submodule of $\Gamma(\mathcal{E})$. Thus, the topologies on $\mathcal{E}$ induced by the dense submodule $\Phi(\Gamma_0)$ sections are the same as that induced by $\Gamma(\mathcal{E})$, so we conclude $\Phi$ is continuous. The same argument works to show that $\Phi^{-1}$ is continuous. 
\end{proof}

\begin{rmk}
    In words, the minimality theorem says that with respect to properties (i)-(iv), having a dense subspaces generated by the transfer operators acting on distinguished sections $\varphi_\omega$ and $e(\omega)$, is enough to ensure that the whole apparatus is in fact isometrically isomorphic to the small correlations bundle. 
\end{rmk}

In Lemmas~\ref{lem:pre-fundamental} through~\ref{lem:iterate_hom} and Theorem~\ref{thm:minimal}, we have extensively explored the properties of the small correlations bundle $\mathcal{C}$. We have seen enough properties that we can now state our first main result.  
\begin{thm}[Theorem~\ref{thmx:Fundamental}]\label{thm:Fundamental}
    Let $\mathcal{A}$ be a unital $C^*$-algebra and $\A_{\Z}$ be the associated quasi-local algebra with canonical group action by translation $\Z \acts{\tau} \A_\Z$. Let $\Omega$ be a compact Hausdorff space equipped with an homeomorphism $\vartheta$. Suppose $\omega \mapsto \psi_\omega$ is a weakly*-continuous mapping into the states $\mathcal{S}(\A_{\Z})$ which is translation covariant. 
    The following are equivalent. 
    \begin{enumerate}[label = (\alph*)]
        \item For all $\omega \in \Omega$, the subspace $\spn \{\psi_\omega(X\otimes (\, \cdot \,))\colon X\in \A^-\}\subset (\A^+)^*$ is finite dimensional. 
        \item There exists a \textbf{transfer apparatus} consisting of a Banach bundle $(\mathcal{E}, \pi, \Omega)$ so that each fiber is nontrivial and finite dimensional, and a triple $(\phi, f, \{F_{a,\omega}\}_{a\in \A^+})$ of covariant homomorphisms implementing $\psi_\omega$ as in~(\ref{eqn:transfer_expectation}).
    \end{enumerate}
\end{thm}
\begin{proof}
    To see the $(a)\Rightarrow (b)$ direction, note that the spanning condition implies that the Banach spaces $B^-_{-1/2, \omega}:= \B^-_\omega$ are all finite dimensional. By our observations in Remark~\ref{rmk:nondegenerate}, there is a perfect bilinear pairing $\eta_\omega :B_\omega^- \times B^+_\omega \to \C$, and thus $\B_\omega^+$ is finite dimensional for all $\omega$ as well. Thus, $\psi_\omega$ has small correlations across every cut by Lemma~\ref{lem:cut_independent}. Further, all the Banach bundles $\mathscr{B}_x^+$ are isomorphic to the canonical bundle $\mathcal{C} = \mathscr{B}_{-1/2}^+$, which we take to be $\mathcal{E}$. The rest of the theorem statement follows from Lemma~\ref{lem:pre-fundamental}.

    The direction $(b)\Rightarrow (a)$ follows essentially the same argument as Proposition 2.1 in \cite{FannesNachtergaeleWerner} with some minor modifications. Define a pair of maps $Q^{+}_\omega : \mathcal{A}^{+} \to \mathcal{E}_\omega$ and $Q^{-}_\omega : \mathcal{A}^- \to \mathcal{E}_\omega^*$ using the homomorphisms $\varphi$, $F_a$ (with $a\in \mathcal{A}$), and the section $f$ as follows:
    \[
        \begin{split}
            Q^+_\omega(a_0 \otimes \cdots \otimes a_n) &= F_{a_0, \omega} \circ F_{a_1, \vartheta \omega} \circ \cdots \circ F_{a_n, \vartheta^n \omega}(f(\vartheta^n\omega))\in \mathcal{E}_{\vartheta^{-1}\omega}\,,\\
            Q^-_{\omega}(a_{-n} \otimes \cdots \otimes a_{-1}) &= \varphi_{\vartheta^{-(n+1)} \omega} \circ F_{a_{-n}, \vartheta^{-n} \omega } \circ \cdots \circ F_{a_{-1}, \vartheta^{-1}\omega}\in \mathcal{E}^*_{\vartheta^{-1}\omega}. 
        \end{split}
    \] 
    
    One checks that for all $\omega\in \Omega$, both $Q^{\pm}$ are linear mappings from $\A^+\to \mathcal{E}_\omega$ and $\A^-\to \mathcal{E}^*_\omega$ respectively. We need to check that $Q_\omega^{\pm}$ extend to bounded operators. Since the span of pure tensors is dense in the minimal $C^*$-norm, it is sufficient to show that 
    \[
        Q_\omega^{+}\left( \sum_{j} a_1^{(j)}\otimes \cdots a_n^{(j)}\right) \le \left\| \sum_{j} a_1^{(j)}\otimes \cdots a_n^{(j)}\right\|_{\min}\,.
    \] To do so, let $V_\omega = \spn\{ \varphi_{\vartheta^{-(n+1)}\omega} \circ F_{a,\vartheta^{-n}\omega} \circ \cdots \circ F_{b, \vartheta^{-1} \omega}\}$, and $W_\omega = \spn\{F_{a_0, \omega} \circ F_{a_1, \vartheta \omega} \circ \cdots \circ F_{a_n, \vartheta^n \omega}(f(\vartheta^n\omega))\}$.  Since every $\mathcal{E}_\omega$ is finite dimensional, so is each dual space $\mathcal{E}_{\omega}^*$, and hence so are $V_\omega$ and $W_\omega$. Now by the Minimality Theorem~\ref{thm:minimal}, we know that $V_\omega$ and $W_\omega$ can be mapped isometrically into the small correlations bundle. By Lemma~\ref{lem:n-norm}, we conclude that $V_\omega$ and $W_\omega$ are norming for one another. 

    Furthermore, one checks that 
    \begin{align*}
        Q^-_\omega\left(\sum_{j'} b_{-m}^{(j')}\otimes \cdots b_0^{(j')}\right)\circ Q_\omega^{+}\left( \sum_{j} a_1^{(j)}\otimes \cdots a_n^{(j)}\right) &= \psi_\omega\left\{\left(\sum_{j'} b_{-m}^{(j')}\otimes \cdots b_0^{(j')}\right)\otimes \left( a_1^{(j)}\otimes \cdots a_n^{(j)}\right)\right\}\,.
    \end{align*} Therefore, by decomposing a $v\in (V_\omega)_1$ into a sum of pure tensors, we may rewrite \[|v(Q^+(\sum a_1 ^{(j)}\otimes \cdots \otimes a_n^{(j)})| = \psi_\omega( v'\otimes \sum a_1 ^{(j)}\otimes \cdots \otimes a_n^{(j)}),\] for some sum of pure tensors $v'\in \A^-$. Thus, we can bound $\|Q^+(a_1 ^{(j)}\otimes \cdots \otimes a_n^{(j)})\|\le 8|\psi_\omega(v' \otimes \one)| \|\sum_j a_1 ^{(j)}\otimes \cdots \otimes a_n^{(j)}\|$ by appealing to Cartesian decompositions if necessary. Therefore, $Q^+$ extends to a bounded linear operator since $|\psi_\omega(v' \otimes \one)|\le \|[v', \omega]^-\|_{\mathcal{C}_\omega^-} = \|v\|$, by Theorem~\ref{thm:minimal} once again. A similar argument works for $Q^-$. Therefore,  
    \[
        \psi_\omega(X^- \otimes X^+) = Q^-_\omega(X^-) \circ Q^+_\omega(X^+), \quad \forall X^-\in \A^-, X^+ \in \A^+\,.
    \]

    Now, since the range of $Q_\omega^-$ is $\spn\{ \varphi_{\vartheta^{-(n+1)}\omega} \circ F_{a,\vartheta^{-n}\omega} \circ \cdots \circ F_{b, \vartheta^{-1} \omega}\}\subset \mathcal{E}_{\vartheta^{-1}\omega}^*$, the latter being finite dimensional, we have that \[\mathcal{E}^*_\omega \supset \spn\{ F^-_\omega(X^-) \colon X^- \in \A^-\}  \cong \spn\{ \psi_\omega(X^- \otimes (\, \cdot \,))\colon X^- \in \A^-\}\subset (\A^+)^*\,,\] is finite dimensional for all $\omega$.  
\end{proof}

\begin{define}\label{def:apparatus}
    In view of Theorems~\ref{thm:Fundamental}, the existence of a triple $(\varphi, F, f)$ where $\varphi \in \hom(\mathcal{E}, \Omega \times \C)$, and $F = \{E_{a}\in \hom^{\vartheta^{-1}}(\mathcal{E}, \mathcal{E})\colon a\in \mathcal{A}\}$, and a section $f\in \Gamma(\mathcal{E})$ satisfying Theorem~\ref{thm:Fundamental}.(b)(i)-(iv) is inherent to the expansion of a TCVS $\psi_\omega$ with small correlations. To make the following more readable, we shall call such a triple a \textbf{transfer apparatus} on the bundle $(\mathcal{E}, \pi, \Omega)$. We shall reserve the letters $(\varrho, E, e)$ for the transfer apparatus defined on the canonical bundle $(\mathcal{C}, \pi, \Omega)$ as in Lemma~\ref{lem:bundlemap-covariance}.  
\end{define}

An immediate corollary is the following characterization of the small correlations property.
\begin{cor}
    Let $\psi_\omega$ be a TCVS on $\mathcal{A}_\Z$. Then $\psi_\omega$ has small correlations if and only if either of the equivalent conditions in Theorem~\ref{thm:Fundamental} are satisfied. 
\end{cor}

\subsection{TCVS with Small Correlations are Weakly* Dense}
In \cite[Proposition 2.6]{FannesNachtergaeleWerner}, it is shown that the class of $C^*$-finitely correlated states is weakly* dense inside the translation invariant states. We aim to show an analogue of this for translation covariant states with small correlations. To do so, we need to recall some light details about direct sums and tensor products of Banach bundles. 

Let $X$ and $Y$ be two Banach bundles over $\Omega$. Form the family $\{Z_\omega := X_\omega \oplus Y_\omega\}_{\omega \in \Omega}$ by taking the topological direct sum of the fibers, equipped with the canonical surjection $\pi$ onto the base-points. Taking $\Gamma = \Gamma(X) \oplus \Gamma(Y)$ defines a Banach module of sections into $Z:= \bigsqcup_\omega Z_\omega,$ and therefore there exists a unique topology turning $(Z,\pi, \Omega)$ into a Banach bundle with $\Gamma(X) \oplus \Gamma(Y)$ a dense sub-module of continuous sections (Theorem~\ref{thm:Banach_Bundle}). We say that $Z$ is the \textit{direct sum} of $X$ and $Y$ and write $Z = X\oplus Y$. It is not difficult to check that the analogous definition for homomorphisms is valid.

\begin{note}\label{note:set_of_tcvs}
    Let $\mathcal{A}$ be a unital $C^*$-algebra and let $\A_{\Z}$ be the associated quasilocal algebra. Let $\Omega$ be a compact Hausdorff space equipped with an homeomorphism $\vartheta$. By $\mathfrak{C}$, we mean the set of translation co-variant states on $\A_\Z$. The subset of those translation co-variant states with small correlations is denoted by $\mathfrak{K}\subset \mathfrak{C}$. 
\end{note} 

\begin{lem}[$\mathfrak{C}$ is Convex]\label{lem:convex}
    With notation as in~\ref{note:set_of_tcvs}, both $\mathfrak{C}$ and $\mathfrak{K}$ are convex within the vector space of weakly*-continuous mappings into $\mathcal{S}(\A_\Z)$. 
\end{lem}
\begin{proof}
    The proof is essentially the same as in \cite[Proposition 2.6(2)]{FannesNachtergaeleWerner}. The nontrivial part is to show that $\mathfrak{K}$ is convex. To this end, let $\psi, \phi \in \mathfrak{K}$ and $\mu \in (0,1)$. Clearly, $g:= \mu \psi + (1-\mu) \phi \in \mathfrak{C}$. Theorem~\ref{thm:Fundamental} tells us there exists Banach bundles $(\mathcal{E}, \pi_{\mathcal{E}}, \Omega)$ and $(\mathcal{F}, \pi_{\mathcal{F}}, \Omega)$, sections $e\in \Gamma(\mathcal{E})$ and $f\in \Gamma(\mathcal{F})$, homomorphisms $\varrho\in \hom(\mathcal{E}, \Omega \times \C)$ and $\varphi\in \hom(\mathcal{F}, \Omega \times \C)$, and transfer operators $\{E_a\}_{a\in \mathcal{A}}\in \hom^\vartheta(\mathcal{E},\mathcal{E})$, and $\{F_a\}_{a\in \mathcal{A}}\subset \hom^\vartheta(\mathcal{F}, \mathcal{F})$, for $\psi$ and $\phi$ respectively. Form $\mathcal{G}:= \mathcal{E}\oplus \mathcal{F}$ as the direct sum of the Banach bundles. One readily checks that with $\gamma:= \mu \varrho \oplus (1-\mu)\varphi \in \hom(\mathcal{G}, \Omega \times \C)$, $h:= e \oplus f$, and transfer operators $\{G_a := E_a \oplus F_a\}_{a\in \mathcal{A}},$ that the bundle $(\mathcal{G}, \pi_{\mathcal{G}}, \Omega)$ satisfies~(\ref{eqn:transfer_expectation}) for the state $g$. 
\end{proof}

\begin{lem}[$\mathfrak{K}$ is robust under grouping]
    Let $p\in \mathbb{N}$ and let $\psi\in \mathfrak{K}$. There exists a weakly* continuous mapping $\psi^{(p)}:\Omega \to \mathcal{S}( (\mathcal{A}^{\otimes p})_\mathbb{Z})$ that is $\vartheta^p$-covariant with respect to translation on $(\mathcal{A}^{\otimes p})_\Z$, and $\psi^{(p)}$ has small correlations. 
\end{lem}
\begin{proof}
    Just as in \cite[Proposition 2.6(3)]{FannesNachtergaeleWerner}, $\psi^{(p)}$ is just $\psi$ viewed as taking arguments from $(\mathcal{A}^{\otimes p})_{\Z}$, in which case, the iterated transfer maps from Lemma~\ref{lem:iterate_hom} determine the expectation values. 
\end{proof}

\begin{lem}[Aromatic Average]\label{lem:aromatic_average}
    Suppose that $\phi$ is a TCVS on $(\mathcal{A}^{\otimes p})_{\Z}$ with small correlations that co-varies with respect to $\vartheta^p$. That is, $\phi \circ \tau_{kp} = \phi_{\vartheta^{kp}}$ for all $k\in \Z$. If $\phi$ has small correlations, then the state 
        \begin{equation}
            \bar \phi_\omega := \frac{1}{p} \sum_{n=0}^p \phi_{\vartheta^{-n}\omega}\circ \tau_n\,,
        \end{equation} is translation covariant with respect to $\vartheta$ on $\A_{\Z}$ and has small correlations. 
\end{lem}

\begin{proof}
    Checking translation covariance of $\bar \phi$ is not difficult. Now, we need to construct a bundle satisfying part (b) of Theorem~\ref{thm:Fundamental}. To begin, let $(\mathcal{C}^\phi, \pi, \Omega)$ be the canonical bundle associated to $\phi$ on $(\mathcal{A}^{\otimes p})_\Z$ equipped with transfer apparatus $(\varrho^\phi, \{E^{\phi}_x\}_{x\in \mathcal{A}^p}, e^\phi)$. 

    Let $0<n\le p-1$, consider the vector subspaces of $\mathcal{C}^\phi_\omega$ given by
    \[
        \mathcal{C}^{(n)}_\omega := \spn \{ [\one_{\mathcal{A}}^{\otimes n}\otimes \tau_n(\xi), \omega] \colon \xi \in \mathcal{A}_{[0, \infty)}\}. 
    \] We take $\mathcal{C}^{(0)}_\omega = \mathcal{C}^\phi_\omega$. For each $\omega\in \Omega$ we now define the finite-dimensional vector space 
    \begin{equation}
        \mathcal{R}_\omega := \spn \left\{ \bigoplus_{n=0}^{p-1} [\one_{\mathcal{A}}^{\otimes n} \otimes \tau_n(\xi), \vartheta^{-n}\omega] \colon \xi \in \mathcal{A}_{[0, \infty)}\right\} \subset \bigoplus_{n=0}^{p-1} \mathcal{C}^{(n)}_{\vartheta^{-n} \omega}\,,
    \end{equation} equipped with the graph norm 
    \begin{equation}
        \left\| \bigoplus_{n=0}^{p-1} [\one_{\mathcal{A}}^{\otimes n} \otimes \tau_n(\xi), \vartheta^{-n}\omega]\right\|_{\mathcal{R}_\omega} := \sum_{n=0}^{p-1} \|[\one_{\mathcal{A}}^{\otimes n} \otimes \tau_n(\xi), \vartheta^{-n}\omega]\|_{\mathcal{C}^{\phi}_{\vartheta^{-n}\omega}}\,.
    \end{equation} Note that each element of $\mathcal{R}_\omega$ is a $p$-tuple of equivalence classes of the translates of a single operator $\xi \in \mathcal{A}_{[0, \infty)}$ up to the cut $p$.  

    Now, let $\mathcal{R}:= \bigsqcup_{\omega \in \Omega}\mathcal{R}_\omega$ be equipped with the surjection $\pi(\oplus[\one^{\otimes n}\otimes \tau_n(\xi), \vartheta^{-n}\omega]) = \omega$. Let also $\Gamma$ be the linear space of functions of the form 
        \begin{equation}\label{eqn:aromatic_sections}
            \alpha: \omega \mapsto \bigoplus_{n=0}^{p-1}[\one^{\otimes n}\otimes \tau_n( A(\omega) ), \vartheta^{-n}\omega]\,,
        \end{equation} where $A\in C(\Omega, \mathcal{A}_{[0, \infty)})$. One readily checks that $\Gamma$ is fiber-wise dense in $\mathcal{R}$. Furthermore, one checks that for all $\alpha \in \Gamma$, the assignment $\omega \mapsto \|\alpha(\omega)\|_{\mathcal{R}_\omega}$ is continuous since each summand in the graph norm is the norm of a translate of the continuous section $(\omega \mapsto [A(\omega), \omega])\in \Gamma(\mathcal{C}^\phi)$. Once again, we note that the tuple above is defined using the same function $A$ in each coordinate. 

    Therefore, by the Fell-Doran Theorem (Theorem~\ref{thm:Banach_Bundle} above), there is a unique topology making $\mathcal{R}$ into a Banach bundle with $\Gamma$ a dense $C(\Omega)$-submodule of continuous sections. For brevity, we will refer to $\mathcal{R}$ as the \textbf{aromatic bundle}. 

    We need to equip $\mathcal{R}$ with a transfer apparatus that implements the average $\bar \phi$. To do so, let $a\in \mathcal{A}$ and define the maps
    \begin{equation}
        \begin{split}
            F_{a,\omega} &\left(\oplus_{n=0}^{p-1} [\one_{\mathcal{A}}^{\otimes n} \otimes \tau_n(\xi), \vartheta^{-n} \omega] \right):= \bigoplus_{n=0}^{p-1}[\one_{\mathcal{A}}^{\otimes n}\otimes a \otimes \tau_{n+1}(\xi), \vartheta^{-(n+1)}\omega]\\
            &=[a\otimes \tau_1(\xi), \vartheta^{-1}\omega]\oplus [\one \otimes a \otimes \tau_2(\xi), \vartheta^{-2}\omega] \oplus \cdots \oplus [\one_{\mathcal{A}}^{\otimes (p-1)}\otimes a \otimes \tau_p(\xi), \vartheta^{-p}\omega]\,.
        \end{split}
    \end{equation} 
    \begin{claim}\label{claim:Fa_well_defined}
        The maps $F_{a,\omega}$ are well-defined $\vartheta^{-1}$-covariant homomorphisms of $\mathcal{R}$. 
    \end{claim}

    Postponing the proof of the claim for now, we aim to show the $F_a$'s form a set of transfer operators in a transfer apparatus (see Definition~\ref{def:apparatus}) for $\bar \phi$. To complete the transfer apparatus, let $f(\omega) = \bigoplus [\one, \vartheta^{-n}\omega]$, and let $\mu_\omega = \frac{1}{p} \bigoplus_{n=0}^{p-1} \varrho^\phi_{\vartheta^{-n} \omega}$. One checks that $f$ is a continuous section and that $\mu_\omega\in \hom(\mathcal{R}, \Omega \times \C)$. Furthermore, $\mu_\omega(f(\omega)) \equiv 1$. Now, let $a\in \mathcal{A}$. We will show that $F_{a,\omega}$ implements $\bar\phi$ with $f$ and $\mu$, and the general case will follow by iteration. One has
    \begin{align*}
        \mu_{\vartheta^{-1}\omega}(F_{a, \omega}(f(\omega))) &= \frac{1}{p}\sum_{n=0}^{p-1} \varrho^\phi_{\vartheta^{-(n+1)}\omega} [\one_{\mathcal{A}}^{\otimes n}\otimes a \otimes \one, \vartheta^{-(n+1)}\omega]= \frac{1}{p} \sum_{n=0}^{p-1} \phi_{\vartheta^{-(n+1)}\omega}(\one^{\otimes n}\otimes a)\\
        & = \frac{1}{p} \sum_{n=0}^{p-2} \phi_{\vartheta^{-(n+1)}}(\one^{\otimes n}\otimes a) + \frac{1}{p}\phi_{\omega}(a) = \frac{1}{p} \sum_{n=0}^{p-2} \phi_{\vartheta^{-(n+1)}}(\tau_n(a)) + \frac{1}{p}\phi_{\omega}(a)\\
        &= \bar \phi_\omega (a)\,,
    \end{align*} where the second-to last equation follows since $\phi$ is $\vartheta^p$-covariant, and the result of translating by $-p$ puts $a$ at the zero site. Therefore, $(\mathcal{R}, \pi, \Omega)$ is a Banach bundle whose fibers are finite dimensional with $(\mu, F, f)$ serving as a transfer apparatus for $\bar \phi$. Thus by Theorem~\ref{thm:Fundamental}, we conclude that $\bar \phi$ has small correlations. 
\end{proof}

Below, we record a proof of Claim~\ref{claim:Fa_well_defined}. 
\begin{proof}[Proof of Claim~\ref{claim:Fa_well_defined}]
    Recall that in any $C^*$-algebra, an element is generated by a Cartesian sum of self-adjoint elements. In fact, due to the triangle inequality, to show that the $F_a$ are well-defined, it is sufficient to show that if $\xi \in \mathcal{A}_{[0, \infty)}$ is self-adjoint and $r= \oplus [\one^{\otimes n} \otimes \xi, \vartheta^{-n}\omega]$ is zero, then $F_a(r) = 0$. 
    
    Now, let $a\in \mathcal{A}$. We may write $a = a_0 - a_1 +ia_3 -ia_4$ where each $a_j\ge 0$ and $\sum \|a_j\| \le 2 \|a\|$. Observe that for any $c\in (\mathcal{A}_{(-\infty, 0)})_{s.a.}$, one has the estimate 
    \begin{equation}\label{pf:eqn:long_inequality}
        -\|a_j\| \phi_{\vartheta^{-n}\omega}(c \otimes \one^{\otimes n} \otimes \tau_{n+1}(\xi)) \le \phi_{\vartheta^{-n}\omega}(c\otimes \one^{\otimes (n-1)}a_j\otimes \tau_{n+1}(\xi)) \le \|a_j\| \phi_{\vartheta^{-n}\omega}(c \otimes \one^{\otimes n}\otimes \tau_{n+1}(\xi))\,.
    \end{equation}

    Now, we may write any $d\in \mathcal{A}_{(-\infty, 0)}$ as $d = \Re(d) + i\Im(d)=: d_1 + i d_2$, and $d_1,d_2 \in (\mathcal{A}_{(-\infty, 0)})_{s.a.}$. Therefore, by applying the inequality in~(\ref{pf:eqn:long_inequality}) for both $d_1$ and $d_2$, we obtain for any $0\le n \le p-1$ the inequality
    \begin{align*}
        |\phi_{\vartheta^{-n}\omega}(d \otimes \one^{\otimes (n-1)}&\otimes a \otimes \tau_{n+1}(\xi))|\le |\phi_{\vartheta^{-n}\omega}(d_1\otimes \one^{\otimes (n-1)}\otimes a \otimes \tau_{n+1}(\xi))| +|\phi_{\vartheta^{-n}\omega}(d_2 \otimes \one^{\otimes (n-1)}\otimes a \otimes \tau_{n+1}(\xi))|\\
        &\le 2\|a\|\left[ |\phi_{\vartheta^{-n}\omega}( \frac{d + d^*}{2} \otimes \one^{\otimes n} \otimes \tau_{n+1}(\xi) ) | + |\phi_{\vartheta^{-n}\omega}( \frac{d - d^*}{2} \otimes \one^{\otimes n} \otimes \tau_{n+1}(\xi) ) |\right]\\
        &\le 2\|a\| \left[|\phi_{\vartheta^{-n}\omega}( d \otimes \one^{\otimes n} \otimes \tau_{n+1}(\xi) ) |+|\phi_{\vartheta^{-n}\omega}( d^* \otimes \one^{\otimes (n+1)} \otimes \tau_{n+1}(\xi) ) | \right]\\
        &= 4 \|a\|  |\phi_{\vartheta^{-n}\omega}( d \otimes \one^{\otimes n} \otimes \tau_{n+1}(\xi) ) |\,,
    \end{align*} where the last equality follows since $\xi$ is a self-adjoint. Now, since the fibers $\mathcal{C}^{\phi}_\omega$ are finite dimensional, Lemma~\ref{lem:n-norm} tells us that 

    \begin{equation}\label{pf:eqn:summands}
        \|[\one^{\otimes (n-1)}\otimes a \otimes \tau_{n+1}(\xi), \vartheta^{-(n+1)}\omega]\| \le 4\|a\| \|[\one^{\otimes n} \otimes \tau_{n+1}(\xi), \vartheta^{-(n+1)}\omega]\|\,. 
    \end{equation} Lastly, we note that for $n=p-1$

    \begin{equation}\label{pf:eqn:p-1}
        \|[\one^{\otimes p}\otimes \tau_{p}(\xi), \vartheta^{-p}\omega]\| \le \|[\xi, \omega]\|\,,
    \end{equation} since $\phi$ is $\vartheta^p$-covariant: indeed, we note that
    \[
        \|[\one^{\otimes p}\otimes \tau_{p}(\xi), \vartheta^{-p}\omega]\| = \inf_{x\in I_{\vartheta^{-p} \omega}}\|\one^{\otimes p}\otimes \tau_{p}(\xi) + x\| = \inf_{\tau_{-p}(x)} \|\xi + \tau_{-p}(x)\| \le \inf_{y \in I_{\omega}} \|\xi + y\| = \|[\xi, \omega]\|\,.
    \]
    
    Finally, it follows by the triangle inequality, inequality~(\ref{pf:eqn:summands}), and inequality~(\ref{pf:eqn:p-1})  that for any element $r\in \mathcal{R}$ as defined in the beginning of the proof
    \begin{equation}\label{pf:eqn:estimate_for_s.a._elts_in_R}
        \|F_{a, \pi(r)}(r)\| \le 4 \|a\| \|r\|.
    \end{equation} Whence $F_a$ is well-defined for all $a\in \mathcal{A}$. 
    
    We now work a little harder to extend inequality~(\ref{pf:eqn:estimate_for_s.a._elts_in_R}) to all of $\mathcal{R}$. Notice that if $\xi \in \mathcal{A}_{[0,\infty)}$, then 
    \[
        \|[\xi, \omega]\|_{\mathcal{C}^\phi_\omega} = \|[\xi^*, \omega]\|_{\mathcal{C}^\phi_\omega}\,.
    \] In fact, if $d\in \mathcal{I}_\omega$, then for any $x\in \mathcal{A}_{(-\infty, 0)}$ one has $0 = \phi_\omega(x\otimes d) = \overline{ \phi_\omega(x^*\otimes d^*)}$, whence  $\mathcal{I}_\omega = \mathcal{I}_\omega^*$ by letting $x$ range. Thus, 
    
    \begin{equation}\label{pf:eqn:bundle_norm_self_adjoint}
        \|[\xi, \omega]\|_{\mathcal{C}^\phi_\omega} = \inf_{d\in \mathcal{I}_\omega} \|\xi + d\| = \inf_{d\in \mathcal{I}_\omega} \|\xi^* + d^*\| = \|[\xi^*, \omega]\|_{\mathcal{C}^\phi_\omega}.
    \end{equation} In fact, $\|[\xi, \omega]\| = 0$ implies both $\|[\Re(\xi), \omega]\| = 0$ and $\|[\Im(\xi), \omega]\|=0$ by the triangle inequality. 

    Now, let $s = \oplus_0^{p-1} [\one^{\otimes n} \otimes \eta, \vartheta^{-n}\omega]$ for $\eta \in \mathcal{A}_{[0,\infty)}$, and let $s_1 = \oplus_0^{p-1} [\one^{\otimes n} \otimes \Re(\eta), \vartheta^{-n}\omega]$ and \\$s_2 = \oplus_0^{p-1} [\one^{\otimes n} \otimes \Im\eta, \vartheta^{-n}\omega]$ so that $s = s_1 + i s_2$ in $\mathcal{R}$. We note $\|s_j\| \le \frac{1}{2} ( \|s\| + \|s^*\|) =\|s\|$ by~(\ref{pf:eqn:bundle_norm_self_adjoint}). Therefore, 
    \begin{equation}\label{eqn:bound_for_aromatic_transfer}
        \|F_{a,\omega}(s)\| \le \|F_{a,\omega}(s_1)\| + \|F_a(s_2)\| \le 4\|a\|\left(\|s_1\| + \|s_2\| \right) \le 8\|a\| \|s\|_\infty\,.
    \end{equation} Thus, we have also shown that $\|F_a\|_\infty\le 8 \|a\|$. 

    That $F_a$ is $\vartheta^{-1}$-covariant is now obvious, so the last thing to check is that for any $\alpha \in \Gamma(\mathcal{R})$, the image $F_a\cdot \alpha$ is a continuous function into $\mathcal{R}$. This follows immediately from the observation that $a\otimes \tau_1(A(\omega)) \in C(\Omega, \mathcal{A}_{[0, \infty)})$ for all $A\in C(\Omega, \mathcal{A}_{[0, \infty)})$, so in fact on the dense $C(\Omega)$ submodule $\Gamma$ given by sections as in~(\ref{eqn:aromatic_sections}),  $F_a \cdot \Gamma \subset \Gamma$, and we are done. 
\end{proof}

We now conclude this section with the following theorem. 
\begin{thm}[Theorem~\ref{thmx:w*dense}]\label{thm:w*_dense}
    Let $\mathcal{A}$ be a unital $C^*$-algebra and let $\A_{\Z}$ be the associated quasilocal algebra. Let $\Omega$ be a compact Hausdorff space equipped with an homeomorphism $\vartheta$. By $\mathfrak{C}$, we mean the set of translation co-variant states on $\A_\Z$ and the subset of those translation co-variant states with small correlations is denoted by $\mathfrak{K}\subset \mathfrak{C}$. One has that $\mathfrak{K}$ is weakly-* dense in $\mathfrak{C}$ uniformly in $\omega$. 
\end{thm}
\begin{proof}
    We roughly follow the strategy laid out in \cite[Proposition 2.6(d)]{FannesNachtergaeleWerner}. Let $\psi_\omega\in \mathfrak{C}$ be a translation co-variant state. For $p\in \mathbb{N}$, let $\phi_\omega = \psi_\omega |_{\mathcal{A}_{[0, p-1]}}$ denote the restriction to the local algebra on the interval $[0, p-1]$. Note that $\phi_\omega$ is weakly-* continuous as a function of $\omega$. Let $\phi'_\omega = \bigotimes_{n=-\infty}^\infty \phi_{\vartheta^{np}\omega}$ denote the $\vartheta^p$-covariant product state on $(\mathcal{A}^{\otimes p})_{\Z}$ as constructed in Proposition~\ref{prop:existence_TCPS}. 

    As we saw in Example~\ref{exmp:covariant_product_state}, it follows that $\phi'_\omega$ has small correlations on $(\mathcal{A}^{\otimes p})_{\mathbb{Z}}$. Let now $\bar \phi_\omega = \frac{1}{p} \sum_{n=0}^{p-1} \phi_{\vartheta^{-n}\omega}'\circ \tau_n$ be the Aromatic average. By Lemma~\ref{lem:aromatic_average}, we know that $\bar \phi$ is a $\vartheta$-TCVS on $\mathcal{A}_\Z$ with small correlations. 

    Now, let $a\in \mathcal{A}_{[0, \infty)}^{\loc}$. Find $p$ large enough so that $\supp(a) \subset [j, j+m] \subset [0, p-1]$, for some integers $j, m\in \mathbb{N}\cup \{0\}$. Notice that $\psi_\omega(a)$ and $\phi_{\vartheta^{-n}\omega}\circ \tau_n(a)$ agree unless $j+m+n\ge p$. Therefore, 
    \begin{align}\label{pf:eqn:w*_dense}
        |\psi_\omega(a) - \bar \phi_\omega(a)| &\le \frac{1}{p} \sum_{p-1-(j+m)}^{p-1} | \psi_\omega(a) - \phi_{\vartheta^n\omega}'\circ \tau_n(a)|\le 2\|a\| \frac{j+m}{p} \to 0\,,
    \end{align} as $p\to \infty$. 

    To see that this inequality holds on all $\mathcal{A}_{\Z}^{\loc}$, note that if $0\in \supp(a)$, then the translate $\tilde a = \tau_m(a) \subset [0, p-1]$ for $p$ large enough where we set $m = |\min\{\supp(a)\}|+1$. By translation co-variance, we have $|\psi_\omega(a) - \bar \phi_\omega(a)|= |\psi_{\vartheta^m\omega}(\tilde a) - \phi_{\vartheta^m \omega}(\tilde a)|$, so the inequality~(\ref{pf:eqn:w*_dense}) applies to the `$\sim$'-ed quantity. But $\|\tilde a\| = \|a\|,$ and we are done since the estimate is uniform in $\omega$. 
\end{proof}
\subsection{Fiber Structure}
In this section, we present a refinement of Theorem~\ref{thm:Fundamental} that allows one to construct a TCVS from general parts. This result is along the same lines as Proposition 2.4 in \cite{FannesNachtergaeleWerner}, but is slightly more general in that we do not need to assume the single-site algebra is finite dimensional. 

Let us recall some terminology before stating our main theorem. We direct the interested reader to \cite[Chapter 13]{Paulsen} (see also \cite[Chapter 5]{EffrosRuan}) for further information. A complex vector space $V$ is a $*$-vector space if it comes equipped with an anti-linear involution $(\,\cdot\,)^*:V\to V$. Those elements with $v^*=v$ are called \textit{self-adjoint}, and we write $(V)_{s.a.}$ for the vector subspace of all self-adjoint elements. We write $\M_n(V)$ for the linear space of $n\times n$ arrays with entries in $V$ equipped with the involution $(v_{i,j})^* = (v_{j,i}^*)$. A $*$-vector space is said to be \textit{matrix ordered} if both of the following are satisfied:
\begin{enumerate}
    \item For every $n\ge 0$, there exists a collection of convex cones $C_n \subset (\M_n(V))_{s.a}$ with $C_n \cap -C_n = \{0\}$ which induce an ordering on $\M_n(V)$. 
    \item For every rectangular $A\in \M_{n\times m}(\C)$, one has $A^*C_n A \subset C_m$ where 
    \[
        A^*(v_{j,k})_{j,k =1}^n A = \left( \sum_{\ell=1}^m \sum_{k=1}^n a_{s,k}a_{r,\ell}^* v_{k,\ell}\right)_{s, \ell =1}^m\,.
    \]
\end{enumerate} Furthermore, if $V$ is a matrix-ordered $*$-vector space, an element $e\in C_1$ is said to be an \textit{order unit} if for all $x\in (V)_{s.a.}$, there is real $r>0$ with $re+x\in C_1$. The element $e$ is \textit{Archimedean} if $re+x \in C_1$ for all $r>0$ implies $x\in C_1$. The last piece of terminology we recall is that $e$ is said to be an \textit{Archimedean matrix order unit} if $e\otimes \one_{n\times n}$ is an Archimedean order unit in $\M_n(V)$ for all $n\ge0$.  

\begin{thm}[Construction Theorem]\label{thm:construction}
    Let $\mathcal{A}$ be a unital $C^*$-algebra and let $(\mathcal{A})_\Z$ be the associated spin chain. Let $\Omega$ be a compact Hausdorff space equipped with an homeomorphism $\vartheta$. The following are equivalent. 
    \begin{enumerate}[label = (\alph*)]
        \item A state $\psi_\omega$ on $\mathcal{A}_\Z$ is a translation co-variant state with small correlations. 
        \item There is a Banach bundle $(\mathcal{E}, \pi, \Omega)$ with a transfer apparatus $(\varrho, E, e)$ so that each of the following holds: 
            \begin{enumerate}[label = (\roman*)]
                \item Each fiber $\mathcal{E}_\omega$ is a finite-dimensional matrix-ordered $*$-vector space with $e(\omega)$ forming an Archimedean matrix order unit. 
                \item Each transfer operator of the form $E_{a^*a, \omega}$ is completely positive with respect to the matrix ordering, and the norm estimate $\|E_{a^*a,\omega}\|\le \|a^*a\|_{\mathcal{A}}$. \item The linear functional $\varrho_\omega$ is completely positive. 
                \item Both $E_{\one, \omega}(e(\omega)) = e(\vartheta^{-1}\omega)$ and $\varrho_\omega(e(\omega))\equiv 1$. 
                \item One has 
                \begin{equation}\label{eqn:back_fibering}
                    \varrho_\omega \circ E_{\one, \omega}\circ E_{\one, \vartheta\omega}\circ \cdots \circ E_{\one, \vartheta^n \omega}= \varrho_{\vartheta^n \omega}.
                \end{equation}
                \item Lastly, defining $\psi^{[-n,n]}_\omega = \varrho_{\vartheta^{-n}\omega} \circ E^{(2n+1)}\circ e(\vartheta^n \omega)$ on $\mathcal{A}^{\otimes (2n+1)}$, via the iterates, the weak-* limit as $n\to \infty$ exists and is a TCVS with small correlations. 
            \end{enumerate}
    \end{enumerate}
\end{thm}
 Before we begin, let us make a remark.
\begin{rmk}
    The main content of Theorem~\ref{thm:construction} is that one can assemble all TCVS with small correlations from bundles of Archimedean matrix ordered vector spaces and a family of completely positive transfer operators thereon. Furthermore, such constructions are weak-* dense in $\mathfrak{C}$ since the argument in Theorem~\ref{thm:w*_dense} did not rely upon the fibers having any additional structure. Put another way: these constructions are weak-* dense in $\mathfrak{C}$ since Theorem~\ref{thm:construction} says \textit{every} TCVS with small correlations can be recovered in this way. 
\end{rmk}

\begin{proof}
    The forward direction $(a)\Rightarrow (b)$ follows by verifying each of $(b)(i)-(vi)$ for the canonical bundle. By the work of \cite[Proposition 3.4]{KavrukPaulsenTodorov}, if we can show that there exists a family of states $\alpha \in \mathcal{S}(\mathcal{A}_{\mathbb{N}})$ so that $I_\omega = \bigcap_{\alpha} \ker \alpha$, then $\mathcal{C}_\omega$ is an operator system with cones given by
    \[
        V_\omega^{(n)} := \{([x_{i,j},\omega]) \colon \forall \epsilon>0,\, \exists d_{i,j}\in I_\omega \colon \epsilon \one_{\mathcal{A}_{\mathbb N}}\otimes \one_{n\times n} + (x_{i,j}+d_{i,j}) \in \M_n(\mathcal{A}_{\mathbb N})_+\}\,.
    \] Let $\mathcal{T}_\omega$ denote the set of $c\in (\mathcal{A}_{(-\infty, 0)})_+$ so that $\psi_\omega(c\otimes \one_{[0, \infty)}) = 1$. Each $c\in \mathcal{T}_\omega$ defines a state by setting $f_c(\,\cdot\,) = \psi_\omega(c\otimes (\,\cdot\,))$. The inclusion $\bigcap_{f\in \mathcal{T}_\omega}\ker f \supset I_\omega$ is trivial. To see the reverse inclusion, first, we note that if $h\in (\mathcal{A}_{(-\infty, 0)})_+$ with $\psi_\omega(h) = 0$, and if $x\in \bigcap_{f\in \mathcal{T}_\omega}\ker f$, then the inequality 
    \[
        0\le |\psi_\omega(h\otimes x)| \le |\psi_\omega(h)\| \|x\| =0\,,
    \] holds since $f_h$ is a positive linear functional on a $C^*$-algebra with unit, and therefore achieves its operator norm at the unit. In general, for any $g\in \mathcal{A}_{(-\infty, 0)}$, we may write $g = g_1 - g_2 +i g_3 - i g_4$ where $g_j\ge 0$. Whence 
    \[
        |\psi_\omega(g \otimes x)| \le \sum_{j=1}^4 |\psi_{\omega}(g_j \otimes x)| = \sum \psi_\omega(g_j \otimes \one) |\psi_\omega(\tilde g_j \otimes x)| = 0\,,
    \] where $\tilde g_j = \psi_\omega(g_j)^{-1} g_j$ where defined. 

    Now, we aim to show $b(ii)$ which is the only nontrivial part of the theorem statement. This follows since \cite[Proposition 3.4]{KavrukPaulsenTodorov} says that the quotient map $p_\omega$ (as defined in Lemma~\ref{lem:quot_eval}) is completely positive for all $\omega$. Notice the mapping $\mu_{a^*a}: (x\mapsto a^*a\otimes \tau_1(x))$ is completely positive on $\mathcal{A}_{[0, \infty)}$. We trivially have $E_{a^*a,\omega} \circ p_\omega= p_{\vartheta^{-1} \omega} \circ \mu_{a^*a}$ the latter being a composition of completely positive maps, hence $E_{a^*a,\omega} \circ p_\omega:\mathcal{A}_{[0, \infty)} \to C_{\vartheta^{-1}\omega}$ is completely positive. But the image of $\A_{[0, \infty)}$ under $p_\omega$ is just $C_\omega$.

    The nontrivial part of $(b)\Rightarrow (a)$ is verifying the weak-* limit property. To this end, we employ Takeda's theorem (Theorem~\ref{thm:Takeda} above). Let $\Lambda \subset \Z$ be finite and write $\Lambda = \{x_0, \dots, x_{|\Lambda|}\}$ with $x_0 < x_1 < \cdots < x_{|\Lambda|}$. Our strategy is to fill in the `gaps' in $\Lambda$ with compositions by $E_{\one, \vartheta^p \omega}$ with $p$ being a site satisfying $x_0< p <x_{|\Lambda|}$ with $p\in \Lambda^c$. To this end, let $M = \max\{|x_0|, |x_{|\Lambda|}|\}$ and set 
    \[
        \psi^{\Lambda}_{\omega} = \varrho_{\vartheta^{-M}\omega} \circ E^{(2M+1)} \circ e(\vartheta^M \omega)= \psi^{[-M,M]}_\omega |_{\mathcal{A}_{\Lambda}}.
    \] Therefore, it is sufficient to check that inclusions of intervals are consistent. To see this, suppose that $m>n$. One has \begin{align*}
        \psi_\omega^{[-m, m]}|_{\mathcal{A}_{[-n,n]}} &= \varrho_{\vartheta^{-m}\omega}\circ E_{\one, \vartheta^{-m}\omega} \circ \cdots \circ E_{\one, \vartheta^{-n-1}\omega}\circ E^{(2n+1)} \circ E_{\one, \vartheta^{n+1}}\circ \cdots \circ e(\vartheta^m\omega)\\
        &= \varrho_{\vartheta^{-n}\omega} \circ E^{(2n+1)} \circ e(\vartheta^n \omega) = \psi^{[-n, n]}_\omega\,,
    \end{align*} where we have used $(b)(v)$ and $(b)(iv)$ to obtain the second equality. Thus, there exists a well-defined state $\psi_\omega$ on $(\mathcal{A})_{\mathbb{Z}}$ that restricts to the finite-volume states we have defined. Moreover, by construction, each $\psi^{\Lambda}_\omega$ is weakly*-continuous in $\omega$ on $\mathcal{A}_{\Lambda}$ from which it follows $\psi_\omega$ is weakly-* continuous in $\omega$ on the dense subalgebra $\mathcal{A}^{\loc}_{\Z}$. Verifying translation covariance of $\psi_\omega$ amounts to checking on $\mathcal{A}^{\loc}_{\mathbb{Z}},$ which is a straight-forward calculation. Lastly, $\psi_\omega$ is a bounded linear functional because of the norm estimate in $(b)(ii)$, and $\psi_\omega$ is positive since each $\psi^{\Lambda}$ is a composition of completely positive maps. But $\psi_\omega:\mathcal{A}_{\Z} \to \C$, the latter being an Abelian $C^*$-algebra, whence $\psi_\omega$ is completely positive. To conclude, note $\psi_\omega$ is a state because of the normalization condition in $(b)(iv)$.
\end{proof}

\begin{prop}\label{prop:C*Bundle}
    Let $\mathcal{A}$ be a unital $C^*$ algebra and let $(\mathcal{A})_\mathbb{Z}$ be the associated quasi-local algebra. Let $\Omega$ be a compact Hausdorff space equipped with an homeomorphism $\vartheta$. Suppose that $\psi_\omega$ is a translation co-variant state with small correlations, and denote by $(\mathcal{C}, \pi, \Omega)$ the canonical bundle. There exists a $C^*$-algebraic bundle over $\Omega$ extending $(\mathcal{C}, \Omega, \pi)$. Moreover, there exists a transfer apparatus for $\psi_\omega$ on the $C^*$-bundle. 
\end{prop}

\begin{proof}
    Let $(\mathcal{H}_\omega, \pi_\omega, \Psi_{\omega})$ be the GNS triple associated to $\psi_\omega$ and $\mathcal{A}_{\Z}$. Let $\mathbb{B}_\omega = \pi_\omega(\mathcal{A}_{\Z})$ denote the image of $\mathcal{A}_\Z$ under $\pi_\omega$ which is a $C^*$-algebra itself for all $\omega\in \Omega$. 

    Now, suppose that $a\in (\mathcal{A}_\Z)_{s.a.}$, which implies $\pi_\omega(a) \in (\mathbb{B}_\omega)_{s.a} \subset B(\mathcal{H}_{\omega})_{s.a.}$ for all $\omega\in \Omega$. Then, 
        \[
            \|\pi_\omega(a)\| = \sup_{\xi, \eta \in (\mathcal H)_1} |\<\xi|\pi_\omega(a) \xi\>| = \sup_{\substack{b\in \mathcal{A}_{\Z}\colon \\ \psi_\omega(b^*b) \le 1}} |\< \pi_{\omega}(b) \Psi_\omega | \pi_\omega(a) \pi_\omega(b) \Psi_{\omega}\>| = \sup_{\substack{b\in \mathcal{A}_\Z} \colon \\ \psi_{\omega}(b^*b) \le 1} |\psi_\omega(b^* a b)|\,.
        \] By arguing similarly to the proof of Lemma~\ref{lem:n-norm}, one finds that this supremum is continuous in $\omega$ whenever $a\in (\mathcal{A}_\Z)_{s.a.}$. Since this set linearly generates the quasilocal algebra, we may apply the Fell-Doran theorem to $\mathbb B_{\omega}$ and see it becomes a Banach bundle when equipped with the continuity structure $\Gamma = \{ \omega \mapsto \pi_\omega(A(\omega)) \colon A\in \caz\}$. 

        To show this is a $C^*$-bundle, we need to check that the multiplication and involution on each fiber is continuous in the bundle topology. Checking continuity of the involution is straightforward.

        Turning toward the multiplication, recall from Remark VII.8.2 in \cite{FellDoran}, we need to check that for each pair of sections $f,g$ the map $\omega \mapsto f(\omega) \otimes g(\omega)$ is continuous. But $f\cdot g \in \Gamma$ by construction, therefore by the Fell-Doran Theorem~\ref{thm:Banach_Bundle}, the topology we have given the $\bigcup_{\omega}\mathbb{B}_{\omega}$ makes the products continuous. 
        
        For all $k\in \Z$, define the maps 
        \[
            U_k: \mathcal{H}_{\omega}\to \mathcal{H}_{\vartheta^k \omega}\,,
        \] via 
        \[
            U_k \pi_\omega(a)\Psi_\omega = \pi_{\vartheta^k \omega}(\tau_k(a)) \Psi_{\vartheta^k \omega}\,.
        \] It is not difficult to show these densely defined linear maps preserve the inner product, and therefore extend to unitaries on $\mathcal{H}_\omega$. Then, the translation may be implemented by the unitary maps 
            \[
                U_1 : \mathcal{H}_\omega \to \mathcal{H}_{\vartheta\omega}\,,
            \] which are independent of $\omega$, and therefore the triple $E_{a,\omega}(\pi_\omega(b)) := U_1^* \pi_{\omega}(\tau_{-1}(a)\cdot b) U_1 $, $\varrho_\omega = \< \Psi_\omega |\, \cdot \,|\Psi_\omega\>$, and $e(\omega) = \one_{B(\mathcal{H}_\omega)}$ defines a transfer apparatus that implements $\psi_\omega$. 
            
            It is sufficient to show that the implementation happens with just two operators by induction. We have
            \begin{align*}
                \psi_{\omega}(a_0 \otimes a_1) &= \langle \Psi_\omega| \pi_\omega(a_0 \otimes a_1)|\Psi_\omega\>= \<\Psi_\omega| \pi_\omega(a_0) \pi_\omega(\tau_1(a_1)) |\Psi_\omega\>\\
                &=\<\Psi_\omega| \pi_\omega(a_0) U_1^* \pi_{\vartheta \omega}(a_1) U_1|\Psi_\omega\>= \<\Psi_\omega| \pi_\omega(a_0) E_{a_1,\vartheta\omega}(\one_{B(\mathcal{H}_{\vartheta\omega})})|\Psi_\omega\>\\
                &= \<\Psi_{\vartheta^{-1} \omega}| U_1^* \pi_\omega(a_0) E_{a_1,\vartheta \omega}(\one)U_1|\Psi_{\vartheta^{-1} \omega}\>= \varrho_{\vartheta^{-1} \omega}\circ E_{a_0,\omega} \circ E_{a_1, \vartheta\omega}(e(\omega))\,,
            \end{align*} as required. By minimality of the small correlations bundle, the extension property is verified. 
\end{proof}

\begin{rmk}
    We are not aware of any way to guarantee that the minimal small correlations bundle embeds into a $C^*$-algebraic bundle whose fibers are finite-dimensional. In fact, Fannes Nachtergaele and Werner proceed by requiring finite-dimensionality in the definition of a $C^*$-Finitely Correlated State (see \cite[Definition 2.4]{FannesNachtergaeleWerner}). The general question of how to embed an operator system in a $C^*$-algebra beyond our scope, but we point the reader to some primary sources. Apparently, the two canonical constructions are the \emph{universal} \cite{KirchbergWassermann} and \emph{minimal} \cite{Hamana} $C^*$-\emph{envelopes}. In each case, there are simple examples of finite-dimensional operator systems where the universal/minimal $C^*$-algebras are necessarily infinite dimensional (see \cite[Example 5.3]{KirchbergWassermann} or \cite[Example 5.2]{Hamana}, respectively). In a related note, it is an outstanding open question as to whether or not every FCS is a $C^*$FCS (see \cite{vanLuijk_Schmidt_Rank, vanLuijk_Review}). 
\end{rmk}

We have shown that any TCVS admits a bundle with a transfer apparatus where the fibers are $C^*$-algebras. Following \cite{FannesNachtergaeleWerner}, let us define the \emph{Covariant Matrix Product States (CMPS)} $\psi_\omega$ to be those for which there is a $C^*$-algebraic bundle which implements $\psi_\omega$ through a transfer apparatus where each fiber is finite-dimensional.

\begin{prop}
    The set of CMPS is a weakly-* dense convex subset of $\mathfrak{C}$ as in Lemma~\ref{lem:convex}.
\end{prop}
\begin{proof}
    The proof is a corollary of the proof of Theorem~\ref{thm:w*_dense} gotten by noticing that the aromatic bundle of a $C^*$-algebraic bundle is still a $C^*$-algebraic bundle. The rest of the proof of Theorem~\ref{thm:w*_dense} does not depend on the multiplication on the fibers. 
\end{proof}

\section{An Independent, Identically Distributed AKLT Model}\label{sec:application}
In this section, we construct and analyze an explicit model of a Translation Co-variant State with small correlations. Certainly, one can recover from our theory the deterministic translation-invariant MPS by taking all the transfer operators to be constant in the disorder. Here, we reify a TCVS with `strong' on-site disorder. We accomplish this by beginning with the one-parameter family of states introduced in \cite{FannesNachtergaeleWerner}, and pick a random value of the parameter at each site. In general, the dimension of the small correlation bundle need not be constant, but our model admits a transfer apparatus on a trivial Banach bundle. We build on the existing theory to show that our IID-AKLT model admits a nearest-neighbor parent Hamiltonian that is measurable in the parameter and translation co-variant. Finally we use this to show the surprising fact that with probability one, the bulk gap of our IID-AKLT model vanishes, yet with probability one, spatial correlations decay exponentially almost surely. We close with a discussion of time reversal symmetry and calculate that our IID-AKLT model has a nontrivial Tasaki index almost surely.

\subsection{Preliminary Observations}\label{subsubsec:example_set_up}
An interesting family of examples comes from the one-parameter family of $C^*$-Finitely Correlated States studied in \cite[Examples 1-5, 7-8]{FannesNachtergaeleWerner}. Many of the proofs given in section 5 of \cite{FannesNachtergaeleWerner} for the MPS expansions extend to the disordered setting in a straightforward way \cite{EkbladMorenoNadalesRoonSchenker}, so we will only record the proofs when necessary. Let us recall the setting, and then examine its properties from our point of view. 

Consider the spin-1 chain where $\mathcal{A} = \M_3$ consists of $3\times 3$ matrices at every site. Equip the local Hilbert spaces $\C^3$ with the basis of spin states $\ket{-1}, \ket{0}, \ket{1}$. Similarly, equip $\C^2$ with the spin-$\tfrac{1}{2}$ basis $\ket{-\tfrac{1}{2}}, \ket{+\tfrac{1}{2}}$. Following \cite[Example 1]{FannesNachtergaeleWerner}, for every $\omega \in [0, \pi)$, define the mapping $V_\omega: \C^2 \to \C^3 \otimes \C^2$ via
    \begin{equation}
        \begin{split}
            V_{\omega}\ket{\tfrac{1}{2}} &= \cos(\omega) \ket{1, - \tfrac12} - \sin(\omega) \ket{0,\tfrac12}\,,\\
            V_\omega \ket{-\tfrac{1}{2}} &= \sin(\omega) \ket{0, -\tfrac12} - \cos(\omega) \ket{ -1, \tfrac 12}\,.
        \end{split}
    \end{equation} 

We may realize $V_\omega$ as a $6\times 2$ matrix as (with respect to the convention that $\ket{+\frac{1}{2}} = [1\, 0]^T, \ket{-\frac{1}{2}} = [0\, 1]^T, \ket{+} = [1\, 0\,0]^T$, etc.)
\begin{equation}\label{eqn:V_matrix}
    V_\omega^* = \begin{bmatrix}
         0 & \cos(\omega) & -\sin(\omega) & 0 & 0 & 0 \\
        0 & 0 & 0 & \sin(\omega) & -\cos(\omega) & 0
    \end{bmatrix}\,. 
\end{equation} We will sometimes write $s = \sin(\omega)$ and $c= \cos(\omega)$ for convenience. We may define the transfer operators 
    \begin{equation}\label{eqn:AKLT_transfer}
    E_{a, \omega} = V^*_\omega ( a \otimes (\, \cdot\,))V_\omega\,,
    \end{equation} which have continuously varying entries in $\omega$ and are completely positive whenever $a\in (\M_3)_+$. One checks \[
            E_{\one, z}(\one) = \one = E_{\one,z}^{\dagger}(\one)\,,
        \] where the dagger denotes the Hilbert-Schmidt adjoint of $E_{\one, z}$. Thus $E_{\one, \omega}$ is unit and trace preserving for all $\delta\in [0, \pi)$.
    
\begin{note}\label{note:pre_disorder_AKLT_state}
By slight abuse of notation, we will let $\omega = (\omega_n)_{n=-\infty}^\infty$ denote a bi-infinite sequence with values in $[0, \pi)$. Later, we will define a probability measure making the space of such sequences into a probability space which justifies the notation. Now, define the state $\nu_{\omega}$ on $\mathcal{A}_{\Z} = (\M_3)_\Z$ by the equations: 
        \begin{equation}\label{eqn:delta_state}
            \nu_{\omega}(a_n \otimes \cdots \otimes a_m) = \frac{1}{2}\tr\left[E_{a_n, \omega_n}\circ \cdots \circ E_{a_m, \omega_m}(\one)\right] \,,
        \end{equation} where the $E_{a,\bullet}$ are defined in equation~(\ref{eqn:AKLT_transfer}) for $a\in \M_2$. 
\end{note}

Since $E_{\one, \bullet}$ is unital and trace-preserving, the formula~(\ref{eqn:delta_state}) is consistent with $*$-homomorphisms associated to volumetric inclusions. By Takeda's theorem~(Theorem~\ref{thm:Takeda}), we have a well-defined bulk state on all $\mathcal{A}_\Z$, whose coefficients in the computational basis are given by a matrix product state expansion with site-dependent angles. The map $E_{\one, \bullet}$ plays a distinguished role in the following and so we shall abbreviate it as $\phi_\bullet$. 

Recall that with respect to the Hilbert-Schmidt ortho-normal basis $\{\one, \sigma^z, \sigma^+, \sigma^-\} \subset \M_2$ formed from the Pauli Matrices, that the matrix $\llbracket \phi_\omega \rrbracket$ defining $\phi_\omega$ is  with respect to the basis of $\M_2$ as above: 
    \begin{equation}\label{eqn:diagonal}
        \llbracket \phi_\omega \rrbracket = \diag(1, \sin^2(\omega) - \cos^2(\omega), -\sin^2(\omega),-\sin^2(\omega)).
    \end{equation}
\begin{lem}\label{lem:translation_operator_strictly_positive}
    Let $\phi_\omega = E_{\one, \omega}\in B(\M_2)$ be the transfer operator defined above. Let $\omega\in ((0, \frac{\pi}{2})\cup(\frac{\pi}{2}, \pi))^\Z$. Put $\Phi_{m,n} = \phi_{\omega_m} \circ \cdots \circ \phi_{\omega_n}$ and let $\mathcal{T}:\M_2 \to \M_2$ be the replacement channel $\mathcal{T}(b) = \tfrac12 \tr(b) \one$. For every $\mu \in (0,1)$ and $x\in [m,n]$, there is  a constant $c_{\mu, x}>0$ so that  
        \begin{equation}
            \|\Phi_{m,n} - \mathcal{T}\| = \max\left\{\prod_{j=m}^n |\cos(2\omega_j)|, \prod_{j=m}^n \sin^2(\omega_j)\right\} \le c_{\mu,x}\, \cdot \mu^{n-m}\,.
        \end{equation}
\end{lem}
\begin{proof}
The first equality is essentially immediate from the fact that $\llbracket\Phi_{n,m}\rrbracket$ is diagonal by equation~(\ref{eqn:diagonal}) above. To obtain the upper estimate, one first checks that $\Phi_{m,n}$ is \emph{eventually strictly positive} in the sense of \cite[Assumption 1]{MovassaghSchenker}.  In fact, we aim to show that $\phi_\omega$ is a strictly positive map itself, in the sense that $A\in (\M_2)_+$ implies $\phi_\omega(A)>0$.  

Fix an $\omega \in (0, \frac{\pi}{2})\cup (\frac{\pi}{2}, \pi)$. Recall \cite{Bhatia_MA, Bhatia_PDM}, a nonzero $2\times 2$ matrix is positive semidefinite if and only if $\tr A = a+c >0$ and $\det A \ge 0$, and moreover if $A$ is positive definite if and only if $\det A >0$. Let $A = \begin{bmatrix}
    a& b\\b^* &c
\end{bmatrix}$ be a positive semidefinite matrix.  
One can compute explicitly 
\[
    \det \phi_\omega(A) = (a^2 + c^2)\sin^2 \omega\,  \cos^2\omega + a \cos^4\omega + (ac-|b|^2) \sin^4 \omega >0\,,
\] for any value of $\omega$ in range provided one of $a,c$ is nonzero. Since $\phi_\omega$ is trace preserving, we know that $\phi_\omega$ is a strictly positive map for any $\omega$ in range. Moreover, since $\phi_\omega$ is trace-preserving, it is faithful, thus $\ker \phi_\omega \cap (\M_2)_+ = \{0\}$, and similarly since $\phi_\omega$ is unital the Schr\"odinger picture adjoint satisfies $\ker \phi_\omega^\dagger \cap (M_2)_+ = \{0\}$. The upper estimate follows from the analysis leading up to \cite[Theorem 2]{MovassaghSchenker}.
\end{proof}

\begin{note}\label{note:Matrices_for_Gamma}
    We note that with respect to the canonical bases, one may define the matrices $X^-(z), X^0(z), X^+(z)$ as follows
        \[
            [X^{(j)}(z)]_{k,\ell = 1}^2 := \<j, k| V_z|\ell\>; \quad j\in \{+, 0, -\} \quad \forall z\in [0, \pi)\,.
        \] One obtains the expressions:
        \begin{equation}\label{eqn:AKLT_Kraus}
    X^{+}(z) := \begin{bmatrix}
        0 & 0\\ \cos(z) & 0\\
    \end{bmatrix}\,\,\,\,\,\,\,\,\,\,  X^{0}(z):= \begin{bmatrix}
        -\sin(z) & 0 \\
        0 & \sin(z)\\
    \end{bmatrix}\,\,\,\,\,\,\,\,\,\,   X^{-}(z) := \begin{bmatrix}
        0 & -\cos(z)\\
        0 & 0\\
    \end{bmatrix} \,.
\end{equation} We remark that \cite{FannesNachtergaeleWerner} define matrices $v^{(j)}$ similarly, and the relation between our notations is that $X^{(j)}(z) = [v^{(j)}(z)]^*$. When convenient, we will often write $X^{(j)}_z$ instead of using function notation. 
\end{note}

\begin{note}\label{note:Gamma_map}
Consider an interval $[m,n] \subset \Z$. We define a family of mappings from $\M_2$ into the local Hilbert spaces $\mathcal{H}^{[m,n]} \cong (\C^3)^{\otimes (n-m)}$ via 
 \begin{equation}\label{eqn:gamma_mn}
    \Gamma^{[m,n]}_{\omega}(b) = \sum_{i_m, \dots, i_n \in \{-, 0, +\}}\, \tr[b X^{i_n}_{\omega_n} X^{i_{n-1}}_{\omega_{n-1}}\cdot \cdots \cdot X^{i_{m}}_{\omega_m}]\, \ket{i_m, \dots, i_n}\,,
 \end{equation} where $\ket{i_m, \dots, i_n}$ is the notation for the computational basis in $\mathcal{H}^{[m,n]}$. 
\end{note}
We remark that the convention of some authors is to find matrices $Y^{(j)}$ that represent the state where the multiplication inside the trace is taken in the order $Y^{i_m} Y^{i_{m-1}} \cdots Y^{i_n}$ (see e.g. \cite{Perez-Garcia_et_al, Tasaki_book}). Since we appeal to formulas in \cite{FannesNachtergaeleWerner}, we shall conform to their conventions for the convenience of the reader.

Notice that for all $\omega\in [0, \pi)^{\Z}$ the composition of transfer operators given by 
 \begin{equation}\label{eqn:z_iterate_transfer}
     E_{a_n \otimes \cdots \otimes a_m, \omega}^{[n,m]}:= E_{a_n, \omega_n}\circ \cdots \circ E_{a_m, \omega_m}:\M_2 \to \M_2\,,
 \end{equation} for $a_n, \dots, a_m \in \M_3$ extends to a linear operator on $\mathcal{A}_{[n,m]}$ (see Lemma~\ref{lem:iterate_hom}, the proof is essentially the same). We have the following lemmas. 
\begin{lem}\label{lem:gamma_formulas}
    Let $z \mapsto (V_z : \C^k \otimes \C^k \to \C^k)$ be the continuous isometry-valued function defined in Equation~(\ref{eqn:V_matrix}). Let $m<n$ be integers and $\omega\in [0, \pi)^{\Z}$.  Let $\{\ket{j}\}_{j\in \{-,0,+\}}$ denote the canoncal basis of $\C^3$ and $\{\ket{k}\}_{k\in \{-\frac12, +\frac12\}}$ denote the canonical basis of $\C^2$. With notation as in~\ref{note:Matrices_for_Gamma} and~\ref{note:Gamma_map}, and equation~(\ref{eqn:z_iterate_transfer}) above, the following hold. 
    \begin{enumerate}[label = \alph*.)]
        \item For all $m<n\in \Z$ and every $a\in \mathcal{A}_{[m,n]} \cong (\M_3)^{\otimes (n-m)}$ the transfer operator $E^{[m,n]}_{a,\omega}: \M_2 \to \M_2$ may be expressed as 
        \begin{equation}
            E^{[m,n]}_{a,\omega}(b) = \sum_{\substack{i_m, \cdots, i_n\\ j_m, \dots, j_n}} \<i_m, \dots, i_n|\, a\, |j_m, \dots, j_n\> (X^{i_m}_{\omega_m})^* \cdots (X^{i_n}_{\omega_n})^*\,\cdot  b\,\cdot X^{j_n}_{\omega_n} \cdots X^{j_m}_{\omega_m}\,.
        \end{equation} 
        \item One has that for all $a\in \mathcal{A}_{[m, n]}$ and all $x, y\in \M_2$, the following expectation value on $\mathcal{H}^{[m,n]}$ has the expansion
        \begin{equation}
            \<\Gamma^{[m,n]}_{\omega}(x)|\, a\,|\Gamma_{\omega}^{[m,n]}(y)\>_{\mathcal{H}^{[m,n]}} = \sum_{j_1,\,j_2\in \{\pm \frac 12\}} \<j_1|\, E^{[m, n]}_{a,\omega}\left(\,x^* \ket{j_1}\bra{j_2}y\, \right)\, | j_2\>_{\C^2}\,.
        \end{equation} In particular, putting $\rho = \tfrac{1}{2}\one_{2\times 2}$, and letting $\langle x, y \rangle_\rho = \tr[\rho x^*y]$, we have the estimate 
    \begin{equation}\label{eqn:rho_product_estimate}
           |\<\Gamma^{[m,n]}_{\omega}(x), \Gamma^{[m,n]}_{\omega}(y)\>_{\mathcal{H}} - \<x,y\>_\rho|\le 2\|\Phi_{m,n}-\mathcal{T}\|\,\|x\|_\rho \|y\|_\rho.
        \end{equation}
        \item For all $\omega\in ((0,\frac{\pi}{2})\cup(\frac{\pi}{2},\pi))^{\Z}$ and all $m\le n$, and all $x,y\in \M_2$, we have 
        \begin{equation}
            |\<\Gamma^{[m,n]}_\omega(x)| a | \Gamma^{[m,n]}_{\omega}(y)\> - \nu(a) \<x|y\>_{\rho}| \le   2\|\Phi_{m,n} - \mathcal{T}\|\, \|a\| \|x\|_{\rho}\|y\|_{\rho} \,.  \end{equation}
            More generally, if $\supp(a) = [m',n']\subset [m,n]$, then, for all $x,y\in \M_2$ with $\|x\|_{\rho}, \|y\|_{\rho} \le 1$, we have
        \begin{equation}\label{eqn:approximate_nu}
            |\<\Gamma^{[m,n]}_\omega(x)|  a\otimes \one_{[m,n]\setminus [m',n']} | \Gamma^{[m,n]}_{\omega}(y)\> - \nu(a) \<x|y\>_{\rho}| \le  2( \|\Phi_{m,m'-1} - \mathcal{T}\| + \|\Phi_{n'+1, n} - \mathcal{T}\|) \, \|a\| \,.
        \end{equation}
    \end{enumerate}
\end{lem} 

The proofs are straightforward extensions of Lemma 5.1 of \cite{FannesNachtergaeleWerner}, but we direct the reader to \cite{EkbladMorenoNadalesRoonSchenker} for full details. We now aim to show that the $\Gamma^{[m,n]}$-maps are injective for every $\omega\in ((0, \frac{\pi}{2})\cup(\frac{\pi}{2}, \pi))^\Z$ whenever $n-m\ge 2$. In the language of \cite{FannesNachtergaeleWerner}, we will show that $\nu_{\omega}$ has an interaction length of $2$ for all $\omega$ which avoid $0$ and $\frac{\pi}{2}$.

Recall \cite{FannesNachtergaeleWerner, Perez-Garcia_et_al} that the subspace \begin{equation}\label{eqn:ground-state-space}
        \Gamma^{[j, j+1]}_{\omega}(\M_2) := \mathcal{G}_{j,j+1} = \left\{ \sum_{i_j, i_{j+1}} \tr(bX_{\omega_{j+1}}^{i_{j+1}}X^{i_{j}}_{\omega_{j}}) |i_j, i_{j+1}\rangle \colon b\in \M_2 \right\} \subset (\C^3)_{j} \otimes (\C^3)_{j+1}\,.
    \end{equation} The subspace $\mathcal{G}_{j, j+1}$ determines the support of a positive operator $P_{j, j+1}$ so that the two-point restriction of $\nu_{\omega}$ satisfies the condition $\nu_{\omega}(P_{j, j+1})=0$. We aim to show that for any $\omega \in ((0, \frac{\pi}{2})\cup(\frac{\pi}{2}, \pi))^\Z$ that $\Gamma^{[j, j+1]}_{\omega}$ is injective. That is, the rank of $\mathcal{G}_{j, j+1}$ is equal to four. 

    Rearranging equation~(\ref{eqn:ground-state-space})  a bit, we see that \[
        \mathcal{G}_{j,j+1} = \underset{0\le k, \ell \le 2}{\spn}\left\{ \sum_{i_j, i_{j+1}} \<k|X^{i_{j+1}}_{\omega_{j+1}}X^{i_j}_{\omega_j}|\ell\>\, \cdot |i_j, i_{j+1}\rangle\right\}\,.
    \]

    Therefore, a basis for $\mathcal{G}_{j, j+1}$ can be constructed from the corresponding entries of the $3\times 3$ array of matrix products $(X^{i_{j+1}}_{\omega_{j+1}}X^{i_j}_{\omega_j})_{i_j,i_{j+1}\in \{-, 0, +\}}$. This results in the following basis vectors which vary as functions in $\omega$ and $j$:
    \begin{equation}\label{eqn:gj_basis}
        \begin{split}
            |\phi_{11}(j)\> &= s_js_{j+1} |00\> - c_jc_{j+1}|+-\>\,,\\
            |\phi_{12}(j)\> &= c_j s_{j+1} \ket{-0} - s_j c_{j+1} \ket{0-}\,,\\
            |\phi_{21}(j)\> &= c_j s_{j+1}|+0\> - s_j c_{j+1} |0+\>\,,\\
            |\phi_{22}(j)\> &= s_j s_{j+1} |00\> - c_j c_{j+1} |-+\>\,.\\
        \end{split}
    \end{equation} where we have written $c_j = \cos( \omega_j)$ and $s_j = \sin(\omega_j)$ for convenience. We have therefore shown the following. 
    \begin{lem}
        Let $\omega$ be any bi-infinite sequence taking values in in $(0, \frac{\pi}{2})\cup(\frac{\pi}{2}, \pi)$ and let $\mathcal{G}_{j, j+1}$ be the ground state space associated to $\Gamma^{[j, j+1]}_{\omega}$. Then the rank of $\mathcal{G}_{j, j+1}$ is identically four for $\omega$ in range. 
    \end{lem}

As noted in \cite{FannesNachtergaeleWerner}, injectivity of the $\Gamma^{[n,m]}$ maps at a short distance does not guarantee they remain injective as $|n-m|$ increases. In order to check this, we must verify that the frustration free intersection property holds. Recall that the three-site ground-state space is given by
    \begin{equation}\label{eqn:G3_subspace}
        \mathcal{G}_{j, j+1, j+2} := \spn\left\{ \sum_{i_j, i_{j+1}, i_{j+2}} \tr\left[b X^{i_{j+2}}_{z_{j+2}}\, X^{i_{j+1}}_{z_{j+1}}\, X^{i_j}_{z_j}\right] |i_j, i_{j+1}, i_{j+2} \rangle \right\}\,.
    \end{equation}
\begin{lem}\label{lem:delta_MPS_injectivity}
    Let $\omega = (\omega_{m})_{m\in \Z}$ be a bi-infinite sequence taking values in $(0, \frac{\pi}{2})\cup(\frac{\pi}{2}, \pi)$ and let $\nu_{\omega}$ be the corresponding state constructed as in~(\ref{eqn:delta_state}). Then, for any site $j\in \Z$,
        \[
            \mathcal{G}_{j, j+1}\otimes (\C^3)_{j+2}\cap (\C^3)_j\otimes \mathcal{G}_{j+1, j+2} = \mathcal{G}_{j, j+1, j+2}.
        \] Thus, the interaction length of $\nu_{\omega}$ is $\ell =2$. Furthermore, for any $\omega$ in range and $n-m>2$, one has 
        \begin{equation}\label{eqn:intersect_ground_state}
            \mathcal{G}_{[m,n]} = \bigcap_{j=m}^{n-2} \mathcal{H}_{[m,j]}\otimes \mathcal{G}_{j, j+1}\otimes \mathcal{H}_{[j+2, n]}\,,
        \end{equation} where we adopt the convention that when $j=m$, the symbol $\mathcal{H}_{[m,m]}$ is interpreted to mean $(\C^3)^{\otimes 0}=\C$.
\end{lem}
\begin{proof}
    After possibly re-indexing, it is sufficient to show that the intersection property holds for three consecutive sites $x, x+1$ and $x+2$ which we will abbreviate as $1,2$ and $3$ for convenience, and let us abbreviate $c_1 = \cos(\omega_j)$, $c_2 = \cos(\omega_{j+1})$, $c_3 = \cos(\omega_{j+2})$, and similarly for the sine values. 
    
Repeating the process leading to the basis~(\ref{eqn:gj_basis}) above, except for $\mathcal{G}_{1,2,3}$, we observe two facts: first, $\mathcal{G}_{1,2,3} \subset \mathcal{G}_{1,2}\otimes (\C^3)_{3}\cap (\C^3)_1\otimes \mathcal{G}_{2,3}$; and second:  $\dim(\mathcal{G}_{1,2,3})= 4$. Indeed since, $\mathcal{G}_{1,2,3}$ is the range of a linear map on $\M_2$, its dimension is bounded by four, but we also can compute exactly four linearly independent vectors from the equality~(\ref{eqn:G3_subspace}). Thus, to show equality of $\mathcal{G}_{1,2,3}$ and $\mathcal{G}_{1,2}\otimes \C^3 \cap \C^3 \otimes \mathcal{G}_{2,3}$ it is sufficient to show $\dim(\mathcal{G}_{1,2}\otimes \C^3 \cap \C^3 \otimes \mathcal{G}_{2,3}) \le 4$.  To do so, we find a convenient decomposition of each space that will allow us to directly compute a basis for the intersection. 
    
    First, notice $S^z_{3}$ commutes with the $z$-component of the spin on $\{1,2\}$ $S^z_{1, 2} = S_1^z+S_2^z$, and similarly with $S^z_1$ and $S_{2, 3}^z$. Thus, we may decompose $\mathcal{G}_{1,2} = \bigoplus_{k=-2}^2 \mathcal{G}_{1, 2}^{(k)}$ where each subspace $\mathcal{G}_{1, 2}^{(k)}$ is contained in the eigensubspace of $S^z_{1, 2}$ corresponding to eigenvalue $k$. However, we note that from~(\ref{eqn:gj_basis}) that $\mathcal{G}_{1, 2}^{(\pm 2)} = \{0\}$. Thus, 
    \[
        \mathcal{G}_{1, 2} \otimes \C^3 = \bigoplus_{k=-1}^1 \bigoplus_{m=-1}^1 \mathcal{G}_{1,2}^{(k)} \otimes \C|m\> = \bigoplus_{k+m = -2}^2 \mathcal{G}_{1,2}^{(k)} \otimes \C|m\>\,.
    \] We may perform a similar decomposition on $\C^3 \otimes \mathcal{G}_{2, 3}$. 

    We have therefore written $\mathcal{G}_{1, 2} \otimes \C^3$ and $\C^3 \otimes \mathcal{G}_{2, 3}$ in the eigenbasis of $S_1^z + S_{2}^z$ and $S^z_2 +S^z_3$ respectively. Both operators commute with the total spin-z $S^z_{1, 2, 3}= S_1^z + S_2^z + S_3^z$ and with each other. Therefore, we can write the intersection as a direct sum over the total spin-z quantum number $\sigma$. All this to say that the intersection admits a direct sum decomposition 
    \[
        \mathcal{G}_{1,2}\otimes \C^3 \cap \C^3 \otimes \mathcal{G}_{2, 3} = \bigoplus_{\sigma=-2}^2 \left(\bigoplus_{k+m = \sigma} \mathcal{G}_{1,2}^{(k)} \otimes |m\>\, \cap \, \bigoplus _{k'+m' = \sigma} |m'\> \otimes \mathcal{G}_{2,3}^{(k')}\right). 
    \] where $k,k',m,m' \in \{-, 0, +\}$.

    Now, to estimate the dimension of $\mathcal{G}_{1,2}\otimes \C^3 \cap \C^3 \otimes \mathcal{G}_{2, 3}$ it is sufficient to estimate the dimension of each direct summand corresponding to one of the total spin-z values $\sigma = \pm 2, \pm 1, 0$. To do so, we will explicitly compute a basis in both the left-hand sum $L(\sigma):= \bigoplus_{k,m\colon k+m=\sigma }\mathcal{G}_{1, 2}^{(k)} \otimes |m\>$ and the right-hand summand $R(\sigma):= \bigoplus_{k',m'\colon k'+m'=\sigma}|m'\>\otimes \mathcal{G}_{2, 3}^{(k')}$. We use these basis expansions to determine the largest possible dimension of $L(\sigma)\cap R(\sigma)$ (namely, by symbolically computing the null space of an appropriate matrix). We will relegate the computations to the Appendix~\ref{apdx:mathematica} below and instead state the results. 
    
    In the case of spin $\sigma =\pm2$,  both the left and right-hand subspaces are one-dimensional and spanned vectors which are linearly independent, hence their intersection is trivial (see Appendix~\ref{apdx:spin2} for the calculations). In the case that $\sigma=1$, we find that the subspaces involved in the intersection are 
    \[
        \begin{split}
            L(1) &= (\mathcal{G}_{1,2}^{(1)}\otimes |0\>)\, \oplus\, (\mathcal{G}_{1,2}^{(0)}\otimes |1\>)\,,\\
            R(1) &= (|0\>\otimes \mathcal{G}_{2,3}^{(1)})\, \oplus\, (|1\>\otimes \mathcal{G}_{2,3}^{(0)})\,.
        \end{split}
    \]
    The resulting basis elements of $L(+1)$ with total spin-z equal to 1 can be written as
    \begin{equation}
        \begin{split}
        v_1&:= |\phi_{21}(1), 0\> = -s_j c_{j+1} |0+0\> + c_j s_{j+1} |+00\>\,,\\
        v_2&: = |\phi_{11}(1), +\> = s_j s_{j+1} |00+\> - c_jc_{j+1}|+-+\>\,,\\
        v_3 &:= |\phi_{22}(1), +\> = s_j s_{j+1} |00+\> - c_j c_{j+1} |-++\>\,.
        \end{split}
    \end{equation} where we note that $v_2, v_3$ are contained in the $m=1$ component, and $v_1$ belongs to the $m=0$ component. Similarly, the right-hand subspace $R(+1)$ is spanned by 
    \begin{equation}
        \begin{split}
            w_1&:= |0, \phi_{21}(2)\> = -s_{j+1} c_{j+2} |00+\> + c_{j+1} s_{j+2} |0+0\>\,,\\
            w_2&: = |+,\phi_{11}(2)\> = s_{j+1} s_{j+2} |+00\> - c_{j+1}c_{j+2}|++-\>\,,\\
            w_3 &:= |+,\phi_{22}(2)\> = s_{j+1} s_{j+2} |+00\> - c_{j+1} c_{j+2} |+-+\>\,.
        \end{split}
    \end{equation} where similarly $w_2$ and $w_3$ belong to the $m'=1$ component and $w_1$ belongs to the $m'=0$ component. Now, if there is a vector in $L(+1) \cap R(+1)$, it is a linear combination of $\{v_1, v_2, v_3\}$ and $\{w_1, w_2, w_3\}$. Whence it is a vector in the kernel of the matrix $Q$, whose columns are formed from the aformentioned basis vectors as follows \[
        \llbracket Q \rrbracket = [v_1,\, v_2,\, v_3,\, -w_1,\, -w_2,\, -w_3]\,,
    \]and whose rows are formed by expanding in the according to the total spin-z basis vectors with $\sigma =1$: \[\ket{-++}, \ket{0+0}, \ket{00+}, \ket{+00}, \ket{+-+}, \ket{++-}.\] The exact entries of $Q$ are given in Table~\ref{tab:s_+1-matrix} in Appendix~\ref{apdx:spin1} below. Calculating the null space of $Q$, (see Appendix~\ref{apdx:spin1}), we find the intersection is one dimensional and spanned by the vector
        \begin{equation}
            n_1 = -2s_1c_2s_3\ket{0+0} + 2c_1s_2s_3 \ket{+00} +s_1s_2c_3\ket{00+} - c_1c_2c_3\ket{+-+}\,,
        \end{equation} The only way for the dimension to collapse to zero is if each coefficient equals zero, which is not possible for our choices of $\omega_1,\omega_2,\omega_3$. 

        Repeating this process in the case of total spin $-1$, we see that the intersection is spanned by
        \begin{equation}
            n_{-1} = c_1c_2c_3\ket{-+-} -s_1s_2c_3 \ket{00-} -c_1s_2s_3\ket{-00}+s_1c_2s_3\ket{0-0}\,,
        \end{equation} for which we  see the only way to make all four coefficients vanish is if $\omega_1 = \omega_2 = \omega_3 = \frac{\pi}{2}$. Therefore $\dim(L(-1) \cap R(-1)) = 1$ except when there is a triple of $\pi/2$'s

    Finally, the case of spin 0 results in a two-dimensional subspace spanned by the vectors
        \begin{equation}
            \begin{split}
                n_0^{1}&=s_1^2c_2c_3 \ket{0+-} - c_1s_1s_2c_3\ket{+0-} -s_1^2 s_2s_3\ket{000}+s_1c_1 c_2 s_3 \ket{+-0}\,, \\
                n_0^{2}&= c_1s_1s_2s_3\ket{000} -c_1^2 c_2 s_3 \ket{-+0} +c_1^2s_2c_3\ket{-0+}-c_1s_1c_2c_3\ket{0-+}\,.
          \end{split}
        \end{equation}
        
Therefore, in the case that $\omega_1,\omega_2,\omega_3 \not \in \{0,\frac{\pi}{2}\}$, we have shown $\mathcal{G}_{1,2,3} = \mathcal{G}_{1,2}\otimes \C^3 \cap \C^3 \otimes \mathcal{G}_{2,3}$
\end{proof}

\begin{rmk}
    In the event that there is a degenerate sequence of values, e.g. $\omega_j=0, \omega_{j+1}=\frac{\pi}{2}, \omega_{j+2}\in(0,\frac{\pi}{2})$, then the dimension of the two-point ground state spaces and the three point ground state spaces can all be different. In the example we just mentioned, one calculates that $\dim(\mathcal{G}_{j, j+1}) = 0$ and $\dim(\mathcal{G}_{j+1,j+2}) = 4$, while $\dim(\mathcal{G}_{j, j+1, j+2})=3$. Therefore, in the degenerate case the interaction length must be greater than 2, if it exists at all. 
\end{rmk}

\begin{cor}\label{cor:Parent_Hamiltonian}
    Let $\omega$ be a bi-infinite sequence with entries in $(0, \pi/2) \cup (\pi/2, \pi)$, and let $\nu_{\omega}$ be the corresponding state on $\mathcal{A}_\Z$. Then, a parent Hamiltonian, h, of $\nu_{\omega}$ is given by
    \begin{equation}\label{eqn:Parent_Hamiltonian}
        h_{j,j+1} = \mathrm{proj}(\mathcal{G}_{j, j+1}^\perp) \,.
    \end{equation} One can explicitly calculate the vectors whose outer products define $h_{j, j+1}$ and we do so in Appendix~\ref{apdx:mathematica}.
\end{cor}
\begin{proof}
    The existence of the parent Hamiltonian is well established in \cite{Perez-Garcia_et_al}. One may explicitly compute the basis of $\mathcal{G}_{j, j+1}$ and the basis of $\mathcal{G}_{j+1, j+2}$ in terms of the $\{|\phi_{k\ell}(j)\rangle\}_{k,\ell = 1,2}$.
\end{proof}

Now that we have established that the interaction length for the $\nu_{\omega}$ is two for suitable $\omega$, we aim to show that  $\Gamma^{[m,n]}_{\omega}$ maps in remains injective when $|m-n|$ increases in analogy to Lemma 5.3 of \cite{FannesNachtergaeleWerner}. 

Consider the following quantity \begin{equation}\label{def:condition_number_gamma}
    \kappa_{[m,n]}(\omega) := \inf_{x\in \M_2} \frac{\|\Gamma^{[m,n]}_{\omega}(x)\|^2}{\|x\|_{\rho}^2}\,,
\end{equation} where $\|x \|_\rho = \frac{1}{2} \|x\|_{\HS}$, which is the smallest singular value of $\Gamma^{[m,n]}_\omega$. 

\begin{lem}
    For a fixed $\omega$ in range, the following inequality holds
        \begin{equation}
            \kappa_{[m,n]}(\omega) \ge (1-\|\Phi_{m,n} -\mathcal{T}\|)\,,
        \end{equation}
    and furthermore, 
        \begin{equation}
            \kappa_{[m,n]}(\omega) \le \min\{ \kappa_{[m-1,n]}(\omega),\, \kappa_{[m,n+1]}(\omega)\}\,,
        \end{equation} That is, $\kappa_{[m,n]}$ is non-decreasing with respect to the interval length $|m-n|$. 
\end{lem}
\begin{proof}
    Once again the proof of \cite[Lemma 5.3]{FannesNachtergaeleWerner} is robust in this setting, so we will only emphasize the one place where our proof differs. Consider the interval $[m-1, n]$. By Lemma~\ref{lem:gamma_formulas}, we have 
    \begin{align*}
         \|\Gamma^{[m-1, n]}_{\omega}(x)\|^2 &= \sum_{i_{m-1}, \dots, i_{n}} |\tr(x X^{i_{n}}_{\omega_n}\, \cdots X^{i_{m-1}}_{\omega_{m-1}}))|^2=\sum_{i_{m-1}} \sum_{i_m, \dots, i_n} |\tr( [X^{i_{m-1}}_{\omega_{m-1}}x]X^{i_{n}}_{\omega_n} \cdots X^{i_{m}}_{\omega_{m}}) )|^2\\
         &= \sum_{i_{m-1}}\|\Gamma^{[m,n]}_{\omega}(X^{i_{m-1}}_{\omega_{m-1}}\, x)\|^2 \ge \sum_{i_{m-1}} \kappa_{m,n}(\omega) \tr(\tfrac{1}{2}\one [X^{i_{m-1}}_{\omega_{m-1}}x]^*[X^{i_{m-1}}_{\omega_{m-1}}x])\\
        &= \kappa_{m,n}(\omega) \sum_{i_{m-1}} \tr[\tfrac{1}{2}\one x^* X_{i_{m-1}}^* X_{i_{m-1}}x]= \kappa_{m,n}(\omega) \|x\|_\rho^2,
    \end{align*} where we have used the fact that the Kraus matrices defining $E_{\one, \omega}$ are trace preserving and unital both. 
\end{proof}

Let us now conclude with a proof that the $\nu_{\omega}$ are pure states for suitable choices of $\omega$. Recall that a state $\varphi$ on a $C^*$-algebra $C$ is pure if and only if whenever $\rho$ is a positive linear functional for which $\rho(c^*c) \le \varphi(c^*c)$ for all $c\in C$, there necessarily exists a $\lambda \in [0, 1]$ so that $\lambda \varphi = \rho$ (see \cite[Chapter 5]{Murphy}). 

\begin{thm}\label{thm:vec_state_pure}
    Let $\omega$ be any bi-infinite sequence taking values in $(0, \frac{\pi}{2})\cup(\frac{\pi}{2}, \pi)$ and let $\nu_{\omega}$ be the state as generated in equation~(\ref{eqn:delta_state}). Then $\nu_{\omega}$ is a pure state on $\mathcal{A}_{\Z}$. 
\end{thm}
\begin{proof}
    We largely follow the proof strategy of \cite[Theorem 5.7]{FannesNachtergaeleWerner}. Suppose that $\rho$ is a state majorized by $\lambda \nu_{\omega}$ for some $\lambda\in (0,1)$. Then for every interval $[m,n]$, one has $\rho|_{\mathcal{A}_{[m,n]}} \le \lambda \nu_{\omega}|_{\mathcal{A}_{[m,n]}}$. In particular, this means that $\rho|_{\mathcal{A}_{[m,n]}}(H^{[m,n]}) = \sum_{j=m}^{n-1} \rho|_{\mathcal{A}_{[m,n]}}(h_{j, j+1}) = 0$. Therefore $\rho|_{\mathcal{A}_{[m,n]}}$ is supported on the intersection in the right-hand side of equation~(\ref{eqn:intersect_ground_state}). But by choice of $\omega$, we know equation~(\ref{eqn:intersect_ground_state}) holds, and therefore $\rho|_{\mathcal{A}_{[m,n]}}$ is supported by $\mathcal{G}_{[m,n]}$. 

    Thus, the density matrix $\hat 
    \rho$ representing $\rho|_{\mathcal{A}_{[m,n]}}$ may be decomposed as 
    \[
        \hat \rho = \sum_k \ket{\Gamma^{[m,n]}_{\omega}(x_k)}\bra{\Gamma^{[m,n]}_{\omega}(x_k)}\,,
    \] with $\sum\|\Gamma^{[m,n]}_{\omega}(x_k)\|^2 = \tr\hat\rho= 1$\,. Now, we see that for any $a$ with finite support $[m',n']\subset [m,n]$, and $|m-n|$ large enough so that $m'\neq m$ and $n'\neq n$, we have  \begin{align*}
        \left|\rho|_{\mathcal{A}_{[m,n]}}(a) - \nu_{\omega}(a) \right|&=  \left|\sum_{k}\bra{\Gamma^{[m,n]}_{\omega}(x_k)}a \ket{\Gamma^{[m,n]}_{\omega}(x_k)} - \nu_{\omega}(a)\<\Gamma^{[m,n]}_{\omega}(x_k)|\Gamma^{[m,n]}_{\omega}(x_k)\>\right| \\
        &\le  \left|\sum_k\bra{\Gamma^{[m,n]}_{\omega}(x_k)}a \ket{\Gamma^{[m,n]}_{\omega}(x_k)} - \tfrac{1}{2}\nu_{\omega}(a)\<x_k |x_k\>_{\HS}\right| + 2\|\Phi_{m,n} - \mathcal{T}\| \|a\| \sum_k \|x_k\|^2_{\rho}\\
        &\le (\|\Phi_{m, m'-1}-\mathcal{T}\| + \|\Phi_{n'+1, n} -\mathcal{T}\|+\|\Phi_{m,n} - \mathcal{T}\|) 2 \|a\| \sum_k \|x_s\|_{\rho}^2 \\
        &\le 2\|a\| \, \frac{\|\Phi_{m, m'-1}-\mathcal{T}\| + \|\Phi_{n'+1, n} -\mathcal{T}\|+\|\Phi_{m,n}-\mathcal{T}\|}{\kappa_{[m,n]}(\omega)} \,,
    \end{align*} where we have used the triangle inequality combined with inequality~(\ref{eqn:rho_product_estimate}) from Lemma~\ref{lem:gamma_formulas} part b to obtain the first inequality and part c from the same lemma for the second. On a set with probability one, the right-hand side of this can be made arbitrarily small by sending $m\to -\infty$ and $n\to \infty$. Therefore, it must be that $\rho|_{\mathcal{A}_{[m',n']}}$ agrees with $\nu_{\omega}|_{\mathcal{A}_{[m',n']}}$. But this forces $\lambda = 1$, and so it must be that $\nu_{\omega}$ is pure.  
\end{proof}

\subsection{IID AKLT}\label{subsec:IID_AKLT}
In this section, we explore the consequences of having IID angles in the parameter space $[0,\pi)$ on the parent Hamiltonian $\{h_{j, j+1}(\omega) \colon j\in \mathbb{Z}, \omega \in [0, \pi)\}$ that we obtained in Corollary~\ref{cor:Parent_Hamiltonian}. To set this up, let $\{\theta_j: j\in \Z\}$ be a family of IID random variables taking values in $[0,\pi)$ with a non-atomic distribution.  By the Lemma~\ref{lem:IID_cts}, we can replace the $\theta_j$'s with an IID family of continuous random variables $\Theta_j:\Omega \to [0, \pi)$ which are related by an ergodic homeomorphism $\vartheta:\Omega \to \Omega$ (i.e. the right shift operator), and $\Omega$ is a compact Hausdorff space. For convenience, we shall write $\Theta_j(\omega) = \omega_j = \vartheta^j \omega_0 = \Theta_0(\vartheta^j(\omega))$. 

By considering $\omega:=(\omega_j)_{j\in \Z}$ as some sequence taking values in $((0,\frac{\pi}{2})\cup(\frac{\pi}{2}, \pi))^\Z$, define the state $\psi_\omega:= \nu_{(\omega_j)_{j\in \Z}}$ as in equation~(\ref{eqn:delta_state}). Then, by Corollary~\ref{cor:Parent_Hamiltonian}, there is a nearest neighbor parent Hamiltonian $h_{j,j+1}(\omega)$ for which $\psi_\omega$ is the unique bulk ground state by Theorem~\ref{thm:vec_state_pure}. 

\begin{lem}\label{lem:parent_covariance}
    In the above situation, the parent Hamiltonian is translation co-variant in the following sense: 
        \begin{equation}
            h_{j,j+1}(\omega) = h_{j+1, j+2}(\vartheta \omega)\,.
        \end{equation}
\end{lem}
\begin{proof}
    This follows from the coordinate expression of $h_{j, j+1}(\omega)$ in equation~(\ref{eqn:Parent_Hamiltonian}).
\end{proof}

We recall some notation from \cite{NachtergaeleSims}. Let $\mathscr{P}_0(\Z)$ denote the finite subsets of $\Z$ and $\Phi:\mathscr{P}_0(\Z) \to \mathcal{A}_\Z^{\loc}$ be a mapping with  $\Phi(Z)=\Phi(Z)^*\in \mathcal{A}_Z$ for all $Z\in \mathscr{P}_0(\Z)$. Recall that for our model, the single site algebra is $\mathcal{A} = \M_3$. For some $\lambda>0$, the $\lambda$-norm of $\Phi$ is the quantity 
    \begin{equation}
        \|\Phi\|_{\lambda}:= \sup_{x\in \Z} \sum_{X\Subset \Z} 9^{|X|}|X|\|\Phi(X)\| e^{\lambda \diam(X)}\,.
    \end{equation} We remark that the numerical quantity comes from the fact that the dimension of the single-site Hilbert space for the AKLT model is $3$. In \cite{NachtergaeleSims}, this is replaced with a more general bound. 

Note that since our parent Hamiltonian $h$ is nearest neighbor, for all $x\in \Z$, the sum in the $\lambda$-norm collapses to a sum over two sets, namely $\{x,x+1\}$ and $\{x-1,x\}$. Then since each $h_{j,j+1}(\omega)$ is a projector, $\|h_{j, j+1}(\omega)\| \equiv 1$. Therefore, the $\lambda$-norm is finite and in fact we have the following.
\begin{lem}
    Let $\{h_{j, j+1}(\omega) \colon j\in \Z, \omega \in \Omega\}$ be the IID parent Hamiltonian given by equation~(\ref{eqn:Parent_Hamiltonian}). Then, 
    \begin{equation}
        \sup_{\omega \in \Omega}\|h\|_{\lambda} = 324e^{2\lambda}\,,\quad \forall\lambda >0\,.
    \end{equation}
\end{lem}

Now, we consider the quantum spin system with local Hamiltonians given by 
    \begin{equation}
        H^{\Lambda}(\omega) = \sum_{Z\subset \Lambda} h_{Z}(\omega) = \sum_{\substack{j\in \Z\colon\\ \{j, j+1\}\subset \Lambda}} h_{j,j+1}(\omega)\in \mathcal{A}_{\Lambda} \text{ where } \Lambda\Subset\Z \,.
    \end{equation}
Let $\alpha^{\Lambda}_{t;\omega}:= e^{itH^{\Lambda}(\omega)}(\,\cdot\,)e^{-itH^\Lambda(\omega)}$ be the dynamics generated by $H^{\Lambda}(\omega)$. We recall that $\alpha^{\Lambda}_{t;\omega}$ is a one-parameter group of automorphisms of $\mathcal{A}_\Lambda$ with respect to $t$. We now come to an interesting observation. 
\begin{lem}\label{lem:dynamics_covariance}
    Let $H^{\Lambda}(\omega)$ be the family of local Hamiltonians given by the IID parent Hamiltonians. Let $\alpha^{\Lambda}_{t;\omega}$ denote the local dynamics generated by $H^{\Lambda}(\omega)$. Then, for all $k\in \Z$, we have
    \begin{equation}\label{eqn:dynamics_covariance}
        \alpha^{\Lambda}_{t;\omega} = \tau_k \circ \alpha^{\Lambda - k}_{t,\vartheta^k\omega} \circ \tau_{-k}\,.
    \end{equation}
\end{lem}
This essentially follows from the relation 
\[
    h_{j,j+1}(\omega) = h_{j+k, j+k+1}(\vartheta^k \omega).
\] For details, see our preprint \cite{RoonSchenker_BGA}. 

Furthermore, since $h_{j,j+1}(\omega)$ is nearest neighbor with a constant operator norm for all $\omega$, the thermodynamic limit $\alpha^{Z}_{t;\omega}\in \Aut(\mathcal{A}_{Z})$ exists strongly and forms a strongly continuous group of automorphisms by \cite{NachtergaeleOgataSims}. In fact, given any increasing sequence of finite volumes $\Lambda_1 \subset \Lambda_2 \subset \cdots \subset \Lambda_n\subset \cdots $ so that $\bigcup_n \Lambda_n = \Z$, then the following limit exists for all $a\in \mathcal{A}_{\Z}^{\loc}$
\begin{equation}
    \alpha^{\Z}_{t;\omega}(a) = \lim_{n\to \infty} \alpha^{\Lambda_n}_{t;\omega}(a)
\end{equation} but by the co-variance property in equation~\ref{eqn:dynamics_covariance} we obtain the following corollary: 
\begin{lem}\label{lem:thermo_covariance}
    The thermodynamic limit of the local dynamics $\alpha^{\Lambda}_{t;\omega}$ is a strongly continuous semigroup of automorphisms in $t$, and the following co-variance condition holds for all $k\in \Z$ all $t\in \mathbb R$ and all $\omega \in \Omega$:
        \begin{equation}\label{eqn:thermo_covariance}
            \alpha^{\Z}_{t;\omega} = \tau_k\circ \alpha^{\Z}_{t;\vartheta^k \omega} \circ \tau_{-k}\,.
        \end{equation}
    Moreover, the state $\nu_\omega$ obtained by applying equation~(\ref{eqn:delta_state}) to $\omega$ is a pure state satisfying \[
        \nu_\omega \circ \alpha^{\Z}_{t;\omega} = \nu_\omega,\, \forall \omega\in \Omega\,.
    \]
\end{lem}
\begin{proof}
    All we need to show is the `moreover' part of the statement. Purity follows immediately from Theorem~\ref{thm:vec_state_pure}. The fact that $\nu_\omega$ is invariant for the dynamics follows by the fact that $\nu_\omega(h_{j,j+1}(\omega)) = 0$ for all $j$, by construction. 
\end{proof}

\begin{cor}\label{cor:spec_const}
    Let $\alpha^{\Z}_{t;\omega}$ be the disordered thermodynamic limit of the local dynamics $\{\alpha^{\Lambda}_{t;\omega}\}$ defined with respect to the IID nearest neighbor Hamiltonian from~(\ref{eqn:Parent_Hamiltonian}). Let $\nu_{\omega}$ be the invariant state gotten by applying equation~(\ref{eqn:delta_state}) to $\omega$ and let also $(\mathcal{H}_\omega, \pi_\omega, \Psi_{\omega})$ be the GNS representation of $\mathcal{A}_{\Z}$ with respect to $\nu_{\omega}$. Let also $S_{t;\omega}$ denote the one-parameter group induced on $\mathcal{H}_\omega$ by $\alpha^{\Z}_{t;\omega}$. There exist unitaries $\{U_k:\mathcal{H}_{\vartheta^k \omega} \to \mathcal{H}_{ \omega}\}_{k \in \Z}$ so that the GNS Hamiltonian $H^{GNS}_\omega$ generating $S_{t;\omega}$ is $\vartheta$-covariant. That is, 
    \begin{equation}\label{eqn:covariant_GNS_gen}
        H_\omega^{GNS} = U_k^* H_{\vartheta^k \omega}^{GNS}U_k\,;\, \forall \omega\in \Omega, k\in \Z\,.
    \end{equation} In particular, $\spec(H^{GNS}_{\omega})$ spectrum is constant for almost every $\omega$. 
\end{cor}
This is an immediate consequence of the existence of the thermodynamic limit and the fact that $\nu$ is a covariant ground state (see \cite{RoonSchenker_BGA} for more details). 

To recap, up to this point we have shown that with respect to our nearest neighbor Hamiltonian $h_{j, j+1}(\omega)$, where the parameter values are IID random variables that change spatially in $j$, then $h_{j, j+1}(\omega)$ admits a translation co-variant thermodynamic limit state which is a ground state of a co-variant automorphism group of dynamics on the quasilocal algebra of observables. In this case, the spectrum of the GNS Hamiltonian is deterministic. Recall that the GNS Hamiltonian has a \emph{spectral gap} if there is $\delta>0$ so that the interval $(0,\delta)\subset \res(H_\omega^{GNS})$ the resolvent set. The \emph{bulk gap} is then the deterministic constant 
    \[
        \gamma:= \sup\{\delta >0 \colon (0,\delta)\subset \res(H_\omega^{GNS})\}\,.
    \] since $\spec(H_{\omega}^{GNS})$ is deterministic. We have arrived at the following theorem
\begin{thm}[Bulk Gap Alternative]
    The bulk gap of the IID AKLT model is deterministic.
\end{thm}

In fact, if $0$ and $\tfrac{\pi}{4}$ are in the support of $\P$, the bulk gap closes with probability one as we will show now. We write $S_{x}^z$ the spin-z matrix for a spin-1 particle at site $x\in \mathbb{Z}$. 

\begin{lem}\label{lem:tee_up_contradiction}
    If $H^{GNS}_{\omega}$ has a nonzero spectral gap $\gamma >0$ (which is necessarily deterministic), then for any $\ell \ge 1$, $\lambda>0$ and any $x\in \Z$, we have
    \[
        |\nu_\omega(S_x^z S_{x+\ell}^z))| \le 3e^{-\mu_{\lambda} \ell}\,.
    \] where $\mu = \frac{\gamma \lambda}{4\|h\|_{\lambda} + \gamma}$
\end{lem}
\begin{proof}
    We aim to verify the conditions of \cite[Theorem 2]{NachtergaeleSims}. That is, we need to show that the element $\pi_\omega(S^z_x)\Psi_\omega$ in the GNS representation of $\A_{\Z}$ with respect to $\nu_\omega$, is orthogonal to the kernel of the GNS Hamiltonian, $H^{GNS}_\omega$. By Theorem~\ref{thm:vec_state_pure}, we note that for any $\omega$ in range, $\nu_\omega$ is pure. This implies, the GNS representation is irreducible hence $\pi_\omega(\A_{Z})' = \C\one$ (\cite[Theorem 5.1.6]{Murphy}). 
    
    Notice that for any $\delta\in (0, \frac{\pi}{2})\cup(\frac{\pi}{2}, \pi)$, the operator  $E_{S_x^z, \delta}(\one_{2\times 2}) =  \cos^2(\delta)\sigma^z$ where $\sigma^z$ is the Pauli-z matrix. Examine the inner product
    \begin{align*}
        \< \, \Psi_{\omega}| \pi_{\omega}(S_x^z) \, \Psi_{\omega}\> &= \nu_{\omega}(S_x^z)= \frac{1}{2}\tr[ E_{S_x^z, \omega_x}(\one_{2\times 2})]= \frac{1}{2}\cos^2(\omega_x)\tr[\sigma^z]=0\,.
    \end{align*} This implies that $\pi_{\omega}(S^z_x)\,\Psi_{\omega}$ is orthogonal to the cyclic vector $\Psi_{\omega}$. Since $\nu_\omega$ is pure,  $\ker{H^{GNS}}$ is one-dimensional which easily follows from the commutant relation above. Thus the projection condition of Theorem 2 in \cite{NachtergaeleSims} is satisfied.  
\end{proof}

We are now ready to state our third main result
\begin{thm}[Theorem~\ref{thmx:Gapless_Hamiltonian}]\label{thm:IID_AKLT_Gapless}
    Suppose that $\{\Theta_z\}_{z\in \Z}$ is a bi-infinite family of IID random variables on a probability space $(\Omega, \P)$. Assume the following
    \begin{enumerate}[label = \roman*.)]
        \item The family $\Theta_z$ has a nonatomic distribution supported in $[0, \tfrac{\pi}{4}]$ which avoids the endpoints. That is, additionally assume $\P[\Theta_x = 0] = 0$, second $\P[\Theta_x = \tfrac{\pi}{4}] = 0$. 
        \item For all $\delta>0$ small enough, assume that $p_0(\delta) := \P[\Theta_0 <\delta] >0$ and that $\P[\Theta_0 = 0] = 0$.
    \end{enumerate}
    Let $\nu_\omega = \nu_{(\omega_j)_{j\in \Z}}$ where $\omega_j = \Theta_j(\omega) = \vartheta^j \Theta_0(\omega)$ be the state constructed as in equation~(\ref{eqn:delta_state}) where $\vartheta$ is the ergodic shift. Then the following hold.
        \begin{enumerate}[label = \alph*.)]
            \item The state $\nu_\omega$ is a translation covariant state with small correlations. In particular, the triple $(\frac{1}{2}\tr,E_{a,\omega} ,\one)$ defines a transfer apparatus for $\nu_\omega$ on the trivial bundle $[0, \frac{\pi}{4}] \times \M_2$. 
            \item Furthermore $\nu_\omega$ admits a nearest neighbor parent Hamiltonian $h$ which is translation co-variant with respect to the ergodic shift. Thus Corollary~\ref{cor:spec_const} implies that the GNS Hamiltonian has a spectrum which is deterministic on a set $\Omega_s$ with $\P(\Omega_s) =1$. 
            \item Let $(\delta_n)_{n=1}^\infty$ be a sequence in $[0, \frac{\pi}{4})$ so that $\delta_n \downarrow 0$ as $n\to \infty$. Then, 
                \begin{equation}
                    \P\left(\forall \ell\in \mathbb N,\, \exists x\in \Z\colon\, |\cos(2\delta_\ell)|^\ell \le |\nu_\omega(S_x^z S_{x+\ell}^z)|\right) =1\,.
                \end{equation}
            \item Therefore, taking $a-c.)$ together implies that the spectral gap of $H^{GNS}$ vanishes on a set with probability one.  
        \item By Birkhoff's Ergodic theorem, the Lyapunov exponent \[
            \Lambda:=\lim_{\ell \to \infty}\log\left(|\nu_\omega(S_x^z\, S_{x+\ell}^z)|^{\tfrac{1}{\ell}}\right)<0\,.
        \] In particular, if $\Theta_x \sim f \tfrac{1}{\pi}dx$ is an $L^\infty$ density function, then we have the following quantitative estimate.
        \begin{equation}
            \Lambda = \E\log(\cos(2\omega)) \le -\frac{\pi \|f\|_\infty}{4}\,\log(2)\,.       
        \end{equation}
        \end{enumerate} 
\end{thm}
\begin{proof}
    a.) The covariance condition follows by construction of the probability space. Similarly, the fact that the trivial bundle implements $\psi_\omega$ here is by construction. b.) is immediate from Corollary~\ref{cor:Parent_Hamiltonian}, and Lemma~\ref{lem:parent_covariance}. For part c.) first, we note that by construction of the $\nu_{\omega}$ state, one has 
    \begin{align}
        \nonumber |\nu_{\omega}(S_x^z S_{x+\ell}^z)| &= \frac{1}{2}|\tr[E_{S^z_x, \omega_x}\circ E_{\one, \omega_{x+1}}\circ\cdots \circ E_{\one, \omega_{x+\ell -1}}\circ E_{S^z, \omega_{x+\ell}}(\one_{2\times 2})]|\\
        \nonumber &= \frac{1}{2}\cos^2(\omega_{x+\ell})|\tr[E_{S^z_x, \omega_x}\circ E_{\one, \omega_{x+1}}\circ\cdots \circ E_{\one, \omega_{x+\ell -1}}(\sigma^z))]|\\
        \nonumber &=\frac{1}{2}\cos(\omega_{x+\ell})^2\, \prod_{j=1}^{\ell-1} |\sin^2(\omega_{x+j}) - \cos^2(\omega_{x+j})|\, | \tr[E_{S^z_x, \omega_x}(\sigma^z)]|\\
        \label{pf:eqn:product_of_cosines}&=\cos^2(\omega_x)\cos^2(\omega_{x+\ell})\prod_{j=1}^{\ell-1} |\sin^2(\omega_{x+j}) - \cos^2(\omega_{x+j})| \,.
    \end{align} 
    
Now, since we assumed $\omega_0 \in [0, \frac{\pi}{4}]$, we know that $1\ge\cos(2\omega_0) \ge 0$ deterministically. Therefore, 
    \begin{equation}\label{pf:eqn:product_of_cosines_lower}
        |\nu_{\omega}(S_x^zS_{x+\ell}^z)| = \left(\frac{1 + \cos(2\omega_x)}{2}\right)\left(\frac{1 + \cos(2\omega_{x+\ell})}{2}\right) \prod_{j=1}^{\ell -1} |\cos(2\omega_{x+j})|\ge \prod_{j=0}^\ell \cos(2\omega_{x+j})\,.
    \end{equation}

    Now, consider the event $E_{x,\ell}(\delta) = [|\cos(2\delta)|^\ell \le |\nu_\omega(S_x^z S_{x+\ell}^z)|]$. Notice that the IID assumption implies the probability $\P E_{x, \ell}(\delta) \ge \P(\Theta_{x}, \Theta_{x+1}, \dots, \Theta_{x+\ell} <\delta) = p_0(\delta)^\ell>0$ for all $x\in \Z$, $\ell \in \mathbb{N}$ and $\delta>0$ sufficiently small. Since we have 
    \[
        \nu_\omega (S_x^z S_{x+\ell}^z) = \nu_\omega( \tau_x \tau_{-x} S_x^z S_{x+\ell}^z) = \nu_{\vartheta^x \omega}(S_0^z S_\ell^z)\,,
    \] the identity $E_{\ell, x}(\delta) = \vartheta^x E_{0,\ell}(\delta)$ holds. Let 
    \begin{equation}
        A_{\ell}(\delta) = \bigcup_{x\in \Z} E_{\ell,x}(\delta)\,,
    \end{equation} which we note is a $\vartheta$ invariant set with $\P(A_{\ell}(\delta))\ge p_0(\delta)^\ell >0$. So $\P A_{\ell}(\delta) = 1$ for all $\ell \ge 0$ and $\delta>0$ by ergodicity. Now, for any sequence $\delta_n \downarrow 0$ we have that 
    \[
        \P \bigcap_{\ell=1}^\infty A_\ell(\delta_\ell) =1\,,
    \] since this is a countable intersection of sets with full probability. 

    d.) Suppose that $H^{GNS}$ has a positive spectral gap on a set with probability one, say $\Omega_G$. Let \[\Omega_s = \bigcap_{\ell}(\Omega_G \cap A_{\ell}(2^{-\ell -1})).\] Then, by part c.), we have that 
        \[
            \P\Omega_s = \P[ \Omega_G \cap \bigcap_{\ell=0}^\infty A_\ell(2^{-\ell-1})] = 1\,.
        \] By Lemma~\ref{lem:tee_up_contradiction} combined with the definition of the $A_\ell$'s we get that for all length scales $\ell$ and almost every $\omega\in \Omega$, there is some $x\in \Z$ depending on $\omega$ and $\ell$ so that 
        \begin{equation}
            |\cos(2^{-{\ell}})|^\ell \le |\nu_\omega(S_{x}^z S_{x+\ell})| \le 3e^{-\mu_{\lambda}\ell}\,.
        \end{equation} Taking $\limsup$ as $\ell \to \infty$ of both sides (noting that the upper and lower bound are deterministic) yields a contradiction. Therefore it must be that $\mu_{\lambda} = 0$ on $\Omega_s$, but this is equivalent to $\gamma = 0$ almost surely, since $\spec(H^{GNS})$ is deterministic by Corollary~\ref{cor:spec_const}. 

    Part e.) follows by utilizing Birkhoff's ergodic theorem \cite[Theorem 6.2.1]{Durrett}. Explicitly, consider the function $g(x) = \log |\cos(2x)|$ for $x\in [0, \frac{\pi}{4}]$. Under the assumption that $\omega_0 \sim f(x) dx$ where $f$ is a continuous density function, then we note $g\in L^1(\P)$ since $f\in L^\infty(dx)$. Indeed, 
    \[
        \int_0^{\pi/4} |\log (\cos(2x))| f (x) dx \le \|f\|_{\infty} \int_0^{\pi/4} -\log|\cos(2x)| dx = \|f\|_{\infty} \frac{\pi}{4} \log(2) < \infty\,.
    \]
    Then, by~(\ref{pf:eqn:product_of_cosines}) and~(\ref{pf:eqn:product_of_cosines_lower}), we obtain the estimate
    \begin{align*}
         \frac{1}{\ell+1}\sum_{j=0}^{\ell}g(\omega_{x+j}) \le \log\left( |\nu_\omega(S_x^z\,S_{x+\ell}^z)|^{\tfrac{1}{\ell}}\right) \le \frac{1}{\ell-1}\sum_{j=1}^{\ell-1} g(\omega_{x+j})\,.
    \end{align*} which holds almost surely. Now, Birkhoff's ergodic theorem implies that the right- and left-hand sides of the above converge almost surely to a constant $\E[g(\omega)]\le - \frac{\|f\|_{\infty}\pi}{4}\log(2)$, as claimed. 
\end{proof}

\begin{rmk}
    We reiterate that the mechanism by which the spectral gap closes is the fact that the probability distribution of $\omega_0$ is allowed to take arbitrarily small values. A reasonable next question to ask is whether or not the spectral gap closes in the case when $\P[\pi/4 >\delta_1 \ge \omega_0 \ge \delta_2>0] = 1$, that is, if $\omega_0$ is valued in $[\delta_2, \delta_1]$ and is not allowed to venture toward degenerate values. It seems likely in this scenario that one could leverage that $gap(E_{\one, \omega})$ can  be explicitly computed in order to show that $gap(H^{GNS}_\omega)$ is open with large probability. However we shall defer a detailed discussion of this idea for later work \cite{EkbladMorenoNadalesRoonSchenker}. 
\end{rmk}

\subsection{Time Reversal Symmetry and an Index Calculation}

Another interesting question raised by our disordered AKLT model is whether or not the disorder affects the values of  topological indices. Many authors have worked on the classification of symmetry protected states via topological indices in recent years (see the sampling of works \cite{Bols_et_al, CarvalhodeRoeckJappens, KapustinSopenkoYang, Ogata_cmp, Ogata_cdm, Ogata_icm,  Sopenko_2d, Tasaki_prs, Tasaki_jmp, Tasaki_HeisenbergChain}). Many (if not all) of the aforementioned works are concerned with unique \emph{gapped} ground states of local Hamiltonians which are protected against the action of some group on a spin system. Due to their explicit nature, Matrix Product State representations have played a crucial role in many of the aforementioned works, and specifically, the AKLT model is known to be protected by several different on-site symmetries. In the following we will show that the matrix product approximation to our disordered AKLT state $\nu_\omega$ is invariant under \emph{Time Reversal} symmetry which can be represented as an action of $\Z_2$ on the spin chain (see \cite{Ogata_z2, Pollman_et_al_z2, Tasaki_book} for more details). Furthermore, we aim to show that $\nu_\omega$ has a nontrivial $\Z_2$-index with high probability. 

Recall that the matrices which generate $\nu_\omega$ at site $j$ can be rewritten in terms of the canonical $2\times 2$ Pauli matrices $\sigma^z, \sigma^x, \sigma^y$ in the following way:
    \[
      X_{-1}(\omega_j) = - \cos(\omega_j) \sigma^-\,;\,\,\,\,\,\,\,\,\,  X_{0}(\omega_j)=- \sin(\omega_j) \sigma^z\,;\,\,\,\,\,\,\,\,\,   X_{1}(\omega_j)  =\cos(\omega_j) \sigma^+ ;
    \] where we recall $\sigma^\pm = \sigma^x \pm i\sigma^y$. To show that our disordered AKLT state is time-reversal invariant, we find it convenient to compute the coefficients explicitly first. 
\begin{lem}\label{lem:coefficient_expansion}
    Let $\nu_{\omega}$ be the state from Notation~\ref{note:pre_disorder_AKLT_state}, where $\omega$ is a bi-infinite sequence in $(0, \frac{\pi}{2})\cup(\frac{\pi}{2}, \pi)$. Let $L\ge 2$, and let Let $\Sigma = (s_{-L}, \dots, s_{L+1})$ be a spin-z configuration on the interval $[-L, L+1]$ (i.e. a choice of basis vector at each site).  Let $J_\Sigma\subset [-L, L+1]$ denote the set 
    \begin{equation}
        J_\Sigma = \{x\in [-L, L+1] \colon s_x \neq 0\}\,.
    \end{equation} Let it be understood that $J_\Sigma^c = [-L, L+1] \setminus J_\Sigma$.  Then, the vector approximates $\Gamma^{[-L, L+1]}_\omega(\one)$ from \eqref{eqn:gamma_mn} have the following coefficients:
    \begin{enumerate}
        \item In the case that there are no excitations (i.e. $\zeta = (0, 0, \dots, 0)$), then the associated coefficient is
        \begin{equation}\label{eqn:zero_configuration}
            \tr[X_0(\omega_{L+1})\cdots X_0(\omega_{-L})] = 2\, \prod_{x=-L}^{L+1} \sin(\omega_x)\,.
        \end{equation}
        \item More generally, if there is $k\ge 1$ so that $|J_\Sigma|=2k$ and the signs of $\sigma_x$ alternate as $x\in J_\Sigma$ increases, then
        \begin{equation}\label{eqn:admissible_configuration}
            \tr[X_{s_{L+1}}(\omega_{L+1}) \cdots X_{s_{-L}} (\omega_{-L})] =  
                (-1)^{P_\Sigma} \prod_{x\in J_\Sigma}s_x\cos(\omega_{x}) \prod_{y\in J_\sigma^c} \sin(\omega_y)\,, 
        \end{equation} where $P_\Sigma$ is the signed sum of the lengths between excitations given by
        \begin{equation}
            P_\Sigma = \delta(+, s_{x_N})(L+1 - x_N) + \sum_{j=1}^{N-1}\delta_{+,j+1}(x_{j+1} - x_{j}) + \delta_{+,1} (x_1 - (-L))\,,
        \end{equation}  where $\delta_{+,j}=\delta(+1, s_{x_j})$ is the Kronecker delta.
        \item Otherwise, the coefficient vanishes. 
    \end{enumerate} 
\end{lem}

\begin{proof}
    First, consider the case of the zero configuration $\zeta = (0, \dots, 0)$. Then each matrix in the trace $\tr[X_{\zeta_{L+1}}(\omega_{L+1}) \dots X_{\zeta_{-L}}(\omega_{-L})]$ is proportional to $-\sigma^z$. Therefore, we may write 
    \[
        \tr[X_{\zeta_{L+1}}(\omega_{L+1}) \dots X_{\zeta_{-L}}(\omega_{-L})]  = \tr[X_{0}(\omega_{L+1}) \dots X_{0}(\omega_{-L})]= \tr[(-\sigma^z)^{2L+2}] \prod_{x=-L}^{L+1} \sin(\omega_x) = 2  \prod_{x=-L}^{L+1} \sin(\omega_x)\,.
    \]
    More generally,  for the sites $y\in J_\Sigma^c$, the associated matrix $X_{s_y}$ is proportional to $-\sigma^z$, or proportional to one of  $\sigma^-$ or $\sigma^+$.  Write $J_{\Sigma} = \{x_1<x_2 < \cdots < x_N\}$ in increasing order.  One can factor out the constants to obtain 
        \[
         \tr[X_{s_{L+1}}(\omega_{L+1}) \cdots X_{s_{-L}} (\omega_{-L})] = \left(\prod_{y\in J^c_\Sigma} \sin(\omega_y)\right) \left( \prod_{x\in J_\Sigma} s_x \cos(\omega_x)\right)\tr Y_{J_\Sigma}\,,\] where we we have set 
         \begin{equation}\label{eqn:calc_coeff_Y_J}
         Y_J = (-\sigma^z)^{L+1 - x_N}\,\sigma^{s_{x_n}}\, (-\sigma^z)^{x_{N}-1 - (x_{N-1}+1)}\sigma^{s_{x_{N-1}}}\cdots (-\sigma^z)^{x_2 - x_1} \sigma^{s_{x_1}}(-\sigma^z)^{x_1 - (-L)}\,.
         \end{equation}

    Now, notice that for any power $p$, one has 
    \[
        \begin{split}
            (-\sigma^z)^p \sigma^+ &= (-1)^p \sigma^+\,,\\
            (-\sigma^z)^p \sigma^-& = \sigma^-\,.
        \end{split}
    \] Therefore, it is immediate that $Y_{J_\Sigma}$ vanishes if there is a pair of subsequent equal indices $s_{x_j} = s_{x_{j+1}}$, since this would imply an $(\sigma^-)^2$ or $(\sigma^+)^2$ appears in $Y_{J_\Sigma}$. 

    Continuing in our reduction, we group the matrices in equation~(\ref{eqn:calc_coeff_Y_J}) from left to right to see that
    \begin{equation}\label{eqn:Y_Jeval}
        Y_{J_\Sigma} = (-1)^{\delta(+, \sigma_{x_N})(L+1 - x_N) + \sum_{j=1}^{N-1}\delta(+1, \sigma_{x_{j+1}})(x_{j+1} - x_{j})} \tr[E_\Sigma (-\sigma^z)^{x_1 - L}]\,,
    \end{equation} where $E_\Sigma$ is some alternating product of $\sigma^+$ and $\sigma^-$. By straightforward computation, we see that $|J_\Sigma|$ cannot be odd, otherwise, $E_\Sigma$ is the zero matrix. 

    To recapitulate, we have shown that $\tr[X_{i_{L+1}}(\omega_{L+1}) \cdots X_{i_{-L}}(\omega_{-L})]$ equals zero unless $|J_\Sigma|=2k$ for some $0\le k \le L+1$, and the nonzero values of $\Sigma$ must be alternating. 

    This means there are only two options for the value of $\tr[E_\Sigma (-\sigma^z)^{x_1 - (-L)}]$. Either 
    \[ \tr[E_\sigma (\sigma^z)^{x_1-(-L)}]= \left\{ 
        \begin{array}{cc}
           \tr[E_{11}(-\sigma^z)^{x_1-(-L)}] = (-1)^{x_1-(-L)} & \text{ if $x_1 = +1$} \\
            \tr[E_{22}(-\sigma^z)^{x_1-(-L)}] = 1 & \text{ if $x_1 = -1$}\\
        \end{array}\right.\,,
    \] where $E_{11} = \llbracket \delta_{1,j}\delta_{1,k}\rrbracket_{j,k=1}^2$ and $E_{22} = \llbracket \delta_{2,j}\delta_{2,k}\rrbracket_{j,k=1}^2$ are the canonical matrix units.
\end{proof}

\begin{cor}\label{cor:TR_symmetry}
    Let $\nu_\omega$ be the disordered AKLT state we construct above. Then $\nu_\omega$ is invariant under time-reversal symmetry. 
\end{cor}

\begin{proof}
Recall that the weak* limit $\lim_{L \to \infty}  \< \Gamma^{[-L, L+1]}_\omega(\one)| \cdot | \Gamma^{[-L, L+1]}_\omega(\one)\> =\nu_\omega$ almost surely by Theorem~\ref{lem:gamma_formulas} part c.). Therefore, it is sufficient to show that the coefficients of the MPS $\Gamma^{[-L, L+1]}_\omega(\one)$ is time-reversal symmetric for any $L$.

 Let $F:\M_2 \to \M_2$ be the quantum channel given by $F(x) = Z^*xZ$ where the unitary $Z= \begin{bmatrix}
        0 & 1\\ -1 & 0\\
    \end{bmatrix}.$ One readily checks that 
    \[
        \left\{ \begin{array}{c}
        F (X_{\pm 1}(\omega)) = - X_{\mp 1}(\omega)\\
        F(X_0(\omega)) =  X_0(\omega)
        \end{array} \right.\,,
    \] for all $\omega$ in range. 
    
    Lemma~\ref{lem:coefficient_expansion} shows that the nonzero coefficients of $\Gamma^{[-L, L+1]}_\omega(\one) \in \mathcal{H}^{[-L, L+1]}$ are functions of spin configurations $\Sigma$ which have alternating sign and obey $|J_\Sigma|=2k$ for some $0\le k \le L+1$. For such a $\Sigma$, one has
    \begin{align*}
        \tr\left[X_{s_{L+1}}(\omega_{L+1}) \cdots X_{s_{-L}}(\omega_{-L})\right] &= \tr\left[Z^*ZX_{\sigma_{L+1}}(\omega_{L+1})\cdot Z^*Z X_{s_{L}}(\omega_{L}) \cdots Z^*ZX_{s_{-L}}(\omega_{-L})\right]\\
        &= \tr[F(X_{s_{L+1}}(\omega_{L+1})) \cdots F(X_{s_{-L}}(\omega_{-L}))] \\
        &= (-1)^{|J_\Sigma|} \tr[X_{-s_{L+1}}(\omega_{L+1}) \cdots X_{-s_{-L}}(\omega_{-L})]\\
        &= \tr[X_{-s_{{L+1}}}(\omega_{L+1}) \dots X_{-s_{{-L}}}(\omega_{-L})]\,.
    \end{align*} 

Under the action of the time-reversal symmetry, any vector in the canonical basis expansion of $\Gamma^{[-L, L+1]}_\omega(\one)$ is mapped to 
    \[
        \tr\left[X_{i_{L+1}}(\omega_{L+1}) \cdots X_{i_{-L}}(\omega_{-L})\right] |i_{-L} \dots i_{L+1}\> \longmapsto \tr\left[X_{-i_{L+1}}(\omega_{L+1}) \cdots X_{-i_{-L}}(\omega_{-L})\right]|(-i_{-L}), \dots ,(-i_{L+1})\>\,,
    \] but as we demonstrated above the coefficients are equal, hence either both or neither vector appears in the sum. Thus $\Gamma^{[-L, L+1]}_\omega(\one)$ is invariant under time-reversal symmetry for every nondegenerate bi-infinite sequence of values $(\omega_j)_{j\in \Z}$ in $(0, \pi/2) \cup (\pi/2, \pi)$. 
\end{proof}


Recent works of Ogata \cite{Ogata_z2} generalized and extended the works of Pollmann et al. \cite{Pollman_et_al_z2, Pollmann_et_al_entanglement} and Tasaki \cite{Tasaki_prs, Tasaki_book} showing that for time-reversal protected pure states on the spin chain, there is a $\Z_2$-topological index. A key step in Ogata's method is the fact that injective matrix product states have a unique MPS representations (up to unitary conjugation and a phase)\cite{FannesNachtergaeleWerner_pure, Perez-Garcia_et_al}, and that furthermore, they obey Matsui's \emph{split property} \cite{Matsui_split, Matsui_entropy}.  Ogata shows in \cite[Theorem 6.1]{Ogata_cmp} that her index agrees with the index of \cite{Pollman_et_al_z2} in the case of a matrix product state. 

While it is generally known that string order parameters correspond to local symmetries of MPS \cite{PG_StringOrder} (which again utilizes the uniqueness of injective MPS representations) we are not currently aware if there is a uniqueness theorem for \emph{ergodically disordered} injective matrix product states. It is also an open question as to whether the index calculation in \cite{Tasaki_prs, Tasaki_HeisenbergChain} is the same as the index from \cite{Ogata_z2} (though both Tasaki's index and Ogata's index numerically agree for the AKLT model). For these reasons, we seek to explicitly compute the $\Z_2$-index of $\nu_\omega$ using the methods presented in \cite{Tasaki_prs}, without appealing to the methods of \cite{Ogata_z2,Pollman_et_al_z2,Pollmann_et_al_entanglement}, and \cite[Section 8.3.4]{Tasaki_book}.

We shall take a few lines to recall Tasaki's method from \cite{Tasaki_prs}. Recall that the \emph{Affleck-Lieb twist operator} is a pure tensor of local unitary operators in the spin chain defined as follows: for $L>1$ an integer, set 
\begin{equation}
    T_L := \bigotimes_{j\in [-L, L+1]} \exp\left\{-2\pi i\, \frac{j+L}{2L+1}\, S^{z}_j\right\}\,,
\end{equation} with $S^z_j$ the $j$-site spin-z matrix for a spin 1. Tasaki then goes on to show the following
\begin{thm}[Theorem 1 \cite{Tasaki_prs}]
    Suppose $\varphi$ is the unique \emph{gapped} ground state of a local Hamiltonian on the spin-1 chain. Then, the value $\varphi(T_\ell)$ is real and has constant sign for $\ell$ sufficiently large. In particular, the quantity $\mathscr{O}(\varphi):= \lim_{\ell \to \infty} \frac{\varphi(T_\ell)}{|\varphi(T_\ell)|}\in \{\pm 1\}\,.$
\end{thm}

Even though our disordered state $\nu_\omega$ is gapless, we show in Theorem~\ref{thm:Z2-index} below, that the $\Z_2$-index $\mathscr{O}(\nu_\omega) = -1$ almost surely. 
\begin{thm}[Theorem~\ref{thmx:TasakiIndex}]\label{thm:Z2-index}
    Let $\nu_\omega$ be the disordered AKLT state as above. Then, for $\P$-a.e. $\omega$, we have 
        \begin{equation}
            \lim_{L\to \infty} \nu_\omega(T_L)=-1\,.
        \end{equation}
\end{thm}
\begin{proof}
To see how the proof works, let us start with a technical claim which we will prove in Section~\ref{sec:proof_of_THM_D_claim} below. 
\begin{claim}\label{claim:THM_D}
Let $\epsilon, L>0$. There is a positive random variable $g(\omega, \epsilon, L)$ so the following upper bound holds
\begin{equation}
    |\langle \Gamma_L |T_L |\Gamma_L\> + \|\Gamma_L\|^2| \le 4\pi \epsilon + g(\omega, \epsilon, L)\,,
\end{equation} and for any $0<\epsilon <\E[\sin^2(\omega)]$ fixed, $g$ almost surely vanishes as $L\to \infty$. 
\end{claim}

With this in place, the argument is as follows. Use the triangle inequality to obtain 
\begin{align*}
    |\nu_\omega(T_L) +1| \le |\nu_\omega(T_L) - \<\Gamma_L|T_L|\Gamma_L\rangle| + |\<\Gamma_L|T_L|\Gamma_L\rangle + \|\Gamma_L\|^2_{\mathcal{H}}|+|1- \|\Gamma_L\|_\mathcal{H}^2|\,.
\end{align*} Recall from Lemma~\ref{lem:gamma_formulas} part c.) that the following inequalities hold almost surely,

\begin{equation}
   \left. \begin{array}{c}
        | \<\Gamma_L|T_L|\Gamma_L\> - \nu_\omega(T_L)|\\
        |\|\Gamma_L\|^2 - 1| 
    \end{array} \right\}\le \|\Phi_{[-L,L+1]} - \mathcal{T}\|\,.
\end{equation} Thus, letting $L \to \infty$ we get that $\limsup_{L\to \infty}|\nu_\omega(T_L)+1| \le 4\pi \epsilon$ for any $\epsilon >0$ sufficiently small. Letting $\epsilon \downarrow 0$ now completes the proof. 
\end{proof}


\subsection{Proof of Claim~\ref{claim:THM_D}}\label{sec:proof_of_THM_D_claim}

We will now construct a disordered classical stochastic process that models the coefficients of $\<\Gamma_L|T_L|\Gamma_L\>$. We take a few lines to discuss the heuristics behind our argument. Owing to equation~(\ref{eqn:stoch_expansion}) below, we observe that for a given \emph{un-signed spin configuration} $\Sigma\in \{0,1\}^{2L+2}$, an excitation at $x\in [-L, L+1]$ picks up a coefficient of $\cos^2(\omega_x)$, whereas a zero picks up $\sin^2(\omega_x)$. We can think of $\Sigma$ as a record of (e.g.) a string of light-bulbs arranged on $\Z$ where the (random) probability that the bulb at site $x$ is ``on," is $c_x^2 := \cos^2(\omega_x)$. Let $\mathcal{M}_0 = \{0,1\}$ and denote by $q^{(x)}_\omega$ the random measure given by $q^{(x)}\{1\} = c_x^2$. Then the collection $\{q^{(x)}\}_{x\in \Z}$ defines a bi-infinite sequence of random, binary measures. 

Now, set $\mathcal{M} = \{0,1\}^\Z$ and define $\mathcal{M}_L = \{0,1\}^{2L+2}$. Define the sequence of measures 
\begin{equation}
    \mathcal{Q}_\omega^{(L)}[x_{-L}, \dots, x_{L+1}] := \prod_{j=-L}^{L+1} q^{(x_j)}_\omega(\{x_j\})= \prod_{\substack{j\in [-L,L+1] \colon\\ x_j = 1}} c_{x_j}^2(\omega)\, \prod_{\substack{j\in [-L,L+1] \colon\\ x_j = 0}} s_{x_j}^2(\omega)\,,
\end{equation} and observe this is a consistent family of product measures. By Kolmogorov's Theorem, there is a unique measure $\mathcal{Q}_\omega$ on $\mathcal{M}$ whose finite dimensional distributions agree with the $\mathcal{Q}^{(L)}_\omega$'s. Denote by $\cale_\omega$ expectation with respect to $\calq_\omega$. Notice that with respect to the random measure $\mathcal{Q}_\omega$, the coordinate projections $\pi_x(\Sigma) := s_x$ are independent but \emph{not} identically distributed. However, identical distribution is restored when we consider the coordinate projections with respect to the annealed measure $\E \circ \calq_\omega$.

\begin{note}
Given a $\Sigma$ in $\mathcal{M}$, we are interested in the number of excitations occurring in a finite interval $[-L, L+1]$. We write
\begin{equation}
    N_L(\Sigma) = \sum_{j=-L}^{L+1} \pi_j(\Sigma)\,. 
\end{equation} We will also require information about the location of the excitations to the right of $-L$ (provided $\Sigma$ is not identically zero). In symbols we can recursively define
\begin{equation}
    \begin{split}
        x_1(\Sigma) &= \min\{ j\ge -L \, \colon \pi_j(\Sigma) = 1\}\,,\\
        x_k(\Sigma) &= \min\{ j\ge -L\, \colon \pi_j(\Sigma) = 1 \text{ and } x_k > x_{k-1}\}\quad \forall 1\le k \le N_L\,.
    \end{split}
\end{equation}
\end{note}

We now show a technical lemma relating the expectation $\langle \Gamma_L |T_L |\Gamma_L\rangle $ to the quenched process $\calq_\omega$. The upshot of Lemma~\ref{lem:pre_simulation} (below) is that $\<\Gamma^{[-L, L+1]}(\one)|T_L|\Gamma^{[-L,L+1]}(\one) \rangle$ can be expanded in terms of a $\calq_\omega$-measurable random variable. By a slight abuse of notation, we will write $\Sigma$ for a signed configuration in part ii.) of the Lemma below. 
\begin{lem}\label{lem:pre_simulation}
    Let $(\omega_j)_{j\in \Z}$ be a bi-infinite sequence of IID random variables taking values in $[0, \frac{\pi}{4}]$. Assume that $\P[\delta< \omega_0] >0$ for all $\delta>0$ in range, and that $\P$ is non-atomic.  
    Let $L>0$ and denote by $\Sigma = (s_x)_{x\in [-L, L+1]} \in \{-1, 0, +1\}^{2L+2}$ an arbitrary signed spin configuration. Let $k\in [-L, L+1]$ be an integer and let $A_{2k}$ be the set of those signed spin configurations in $[-L, L+1]$ with alternating signs and $2k$-many excitations. Then the following hold. 
    \begin{enumerate}[label = \roman*.)]
        \item Let $c^2 := \E \cos^2(\omega_0)$ and $s^2:= \E \sin^2(\omega_0)$. Then, $\frac{1}{2}<c^2<1$ and $0<s^2<\frac{1}{2}$ and furthermore $s^2+c^2 = 1$. 
        \item Let $\Sigma\in A_{2k}$, and let $J_\Sigma = \{x\in [-L, L+1] \colon s_x \neq 0\}$. We shall write $J_\Sigma = \{x_1 < x_2 < \cdots< x_{2k}\}$ in increasing order with the understanding that the $x_m$ depend on $\Sigma$. Then, Tasaki's formula holds:
        \begin{equation*}
            \sum_{x=-L}^{L+1} \frac{x+L}{2L+1} s_x = - s_{x_1}\sum_{m=1}^{k} \frac{x_{2m}- x_{2m-1}}{2L+1}\,.
        \end{equation*}
        \item Let $L>0$ and set $\Gamma_L = \Gamma^{[-L, L+1]}_\omega(\one)$ for convenience. Then the following holds almost surely:
        \begin{equation}\label{eqn:simulation}
            \langle \Gamma_L |T_L|\Gamma_L\rangle = 2 \cale_\omega\left( \chi_{\{N_L \text{ even}\}}(\Sigma) \cos\left(2\pi \sum_{x=-L}^{L+1} \frac{x_{2m}(\Sigma) - x_{2m-1}(\Sigma)}{2L+1}\right)\right)\,.
        \end{equation} A similar expression holds for $\<\Gamma_L|\Gamma_L\>$. Note that here the expansion is in terms of unsigned spin configurations $\Sigma\in \mathcal{M}$.
    \end{enumerate}
\end{lem}
\begin{proof}
    The proofs of parts i.) and ii.) are straightforward. To see part iii.) use equation~(\ref{eqn:admissible_configuration}) to compute directly the following:
    \[
        \<\Gamma_L|T_L|\Gamma_L\> = \sum_{\substack{\Sigma\in A_{2k}\\  k=0, \dots, L+1}} \prod_{x\in J_\Sigma} \cos^2(\omega_x) \prod_{y\in J^c_\Sigma} \sin^2(\omega_y)\, \exp\left\{-2\pi i \sum_{x=-L}^{L+1} \frac{s_x(x+L)}{2L+1}\right\} \,.
    \] But notice that the summand associated to a given configuration $\Sigma$ has a coefficient equal to the complex conjugate of the summand associated to its time-reversal, $-\Sigma$. Thus, 
    \begin{equation}\label{eqn:stoch_expansion}
        \begin{split}
            \<\Gamma_L|T_L|\Gamma_L\>         &=2\sum_k\sum_{\substack{\Sigma \in A_{2k}\colon\\ s_{x_1} = 1}} \prod_{x\in J_\Sigma} \cos^2(\omega_x) \prod_{y\in J^c_\Sigma} \sin^2(\omega_y)\, \cos\left(2\pi \sum_{m=0}^{k} \frac{x_{2m}(\Sigma) - x_{2m-1}(\Sigma)}{2L+1}\right)\,.
        \end{split}
    \end{equation} In the equation above, we see that the signed components of $\Sigma$ do not play any role, and so we may regard $\Sigma$ instead as an unsigned configuration. Taking quenched expectation and using the definition of $\calq_\omega$, and the fact that cosine is even yields equation~(\ref{eqn:simulation}).
\end{proof}

 Now we want to show that the average length between excitations occupies about half the length of the interval $[-L, L+1]$. To do so, we will develop an ``Annealed" Large Deviations estimate after the style of Cram\'er's theorem. Since the lengths between unsigned excitations depend on the random variable $N_L$, we first show that we can control the deviations of $N_L$ from its annealed mean. 
 
\begin{lem}\label{lem:annealed_est_M}
    Let $\epsilon >0$ and write $c^2 = \E \cos^2(\omega_0)$. Consider the event \begin{equation}
        M_{\epsilon, L} := \left[ \left|\frac{1}{2L+2} \sum_{x=-L}^{L+1}\pi_x(\Sigma) - c^2\right| > \epsilon \right]\,.
    \end{equation} Then, there is some $\alpha>0$ and a coefficient $C_{\epsilon,L}$ satisfying $\lim_{\epsilon \downarrow 0} C_{\epsilon,L} = 1$ for all $L>0$, so that 
    \begin{equation}\label{eqn:annealed_LDP}
        \E[\mathcal{Q}_\omega(M_{\epsilon, L})] \le 2C_{\epsilon,L}  e^{-\alpha \epsilon^2(2L+2)}\,.
    \end{equation} Moreover, one may take $\alpha = (2c^2 s^2)^{-1}$ and one may find a function $h(\epsilon)$ with the property that $\lim_{\epsilon \downarrow 0} h(\epsilon) = 0$ for which $C_\epsilon = e^{-(2L+2) \epsilon^3 h(\epsilon)}$.
\end{lem}
\begin{proof} The proof strategy is largely similar to the well-known proofs of Cramer's theorem. The non-standard part is dealing with the fact that the coordinate projections are not $\calq_\omega$-identically distributed. 

    First, let $N_L(\Sigma)  = \sum_{j=-L}^{L+1} \pi_j(\Sigma)$ denote the number of excitations in the interval $[-L, L+1]$ for the configuration $\Sigma$. We then see that $M_{\epsilon, L}$ can be regarded as the event that $N_L>(c^2 + \epsilon)(2L+1)$ or $N_L < (c^2 - \epsilon) (2L+1)$. By symmetry, it is sufficient to prove~(\ref{eqn:annealed_LDP}) for $\E\calq_\omega [N_L > (c^2 + \epsilon) (2L+1)]$ 

    For any $t>0$, we may use Chebychev's inequality to obtain 
    \begin{align*}
        \calq_\omega[N_L>(c^2 + \epsilon) (2L+2) ] &= \calq_\omega[ tN_L > (c^2+\epsilon) t (2L+2)]\le e^{-(c^2 + \epsilon) t (2L+2)} \cale_\omega(\exp\{ t \sum_{j=-L}^{L+1} \pi_j(\Sigma)\})\\
        &=e^{-(c^2 + \epsilon) t (2L+2)} \prod_{j=-L}^{L+1}(s_j^2 + e^t c_j^2)\,.
    \end{align*} Using the fact that the $c_j^2$ are mutually independent in $j$ with respect to $\P$, we obtain 
    \[
        \E\calq_\omega[ N_L > (c^2+\epsilon) (2L+1)] \le \inf_{t>0} e^{-(2L+1)[(c^2 + \epsilon) t - \log(s^2 + e^t c^2) } =: \inf_{t>0} e^{-(2L+1) \Lambda(t, \epsilon)}\,,
    \] where we have set $\Lambda(t, \epsilon) = (c^2 + \epsilon) t - \log(s^2 + e^t c^2)$. Basic computation shows that the right-hand side of the above is achieved at the critical point for $\Lambda(t)$ depending on $\epsilon$. That is, $\inf_{t>0} e^{-(2L+1) \Lambda(t, \epsilon)} =  e^{-(2L+1) \Lambda(t_*, \epsilon)}$ where  
    \[
        t_*(c+\epsilon) := \log\left(\frac{s^2 (c^2 + \epsilon)}{c^2(s^2 + \epsilon) } \right)\,,
    \] and furthermore that 
    \[
        \begin{split}
            \Lambda^*(\epsilon)&:= \Lambda(t_*,\epsilon) = (c^2 + \epsilon)\log\left(\frac{s^2 (c^2 + \epsilon)}{c^2(s^2 + \epsilon) } \right)  - \log\left( \frac{s^2}{s^2-\epsilon} \right)\,,\\
            \frac{d}{d\epsilon} \Lambda^*(\epsilon) &= \log\left(\frac{s^2 (c^2 +\epsilon)}{c^2 (s^2 - \epsilon) } \right)\,.
        \end{split}
    \] Notice that $\Lambda^*(0) = 0$ and $(\Lambda^*)'(0) = 0$. Therefore, by Taylor's theorem, we can write $\Lambda^*(\epsilon) = \alpha \epsilon^2 + o(\epsilon^3)$ as $\epsilon \downarrow 0$, with $\alpha = (\Lambda^*)''(0) = (2c^2 s^2)^{-1}$. Putting all the above steps together yields 
    \begin{equation}
        \E \calq_\omega[M_{\epsilon, L}] \le C_{\epsilon,L}\, e^{-\alpha \epsilon^2 (2L+2)}\,,
    \end{equation} where $C_{\epsilon,L} = \exp\{ -(2L+2)\, o(\epsilon^3)\}$ hence $\lim_{\epsilon \downarrow 0} C_{\epsilon,L} =1$
\end{proof}
Using the Lemma above, we can prove that the limit as $L\to \infty$ of the `quenched' measure $\calq_\omega(M_{\epsilon,L})$ vanishes almost surely.
\begin{cor}\label{cor:quenched_limit_M}
    For every $\epsilon >0$ and $\P$-a.e. $\omega$, the limit $\lim_{L \to \infty} \calq_\omega(M_{\epsilon, L}) = 0$. 
\end{cor}
\begin{proof}
    By the Dominated Convergence Theorem combined with~(\ref{eqn:annealed_LDP}), we find that $\E \sum_{L=1}^\infty \calq_\omega [M_{\epsilon, L}] <\infty$. It is a standard trick using Fubini's theorem to convert $L^1(\P)$ convergence of this series into almost sure convergence of the summands. 
\end{proof}

To continue our analysis, we introduce a related set of events which control the average distance between excitation-pairs. Let
\begin{equation}
    D_{\epsilon, L} = \left\{\Sigma \colon \left|\frac{1}{2L+1} \sum_{k=1}^{N_L/2} x_{2k}(\Sigma) - x_{2k-1}(\Sigma) - \frac{1}{2} \right|>\epsilon \right\}\,.
\end{equation} In words, the event $D_{\epsilon, L}$ is the set of configurations where the average distance between every other excitation-pair is far from $\tfrac{1}{2}$. Of course, with respect to even the annealed measure $\P \circ \calq_\omega$, the $x_j$ are not independent. However, the differences are uncorrelated with respect to $\E \circ \cale_\omega$, and this will allow us to perform the same annealed Chebyshev trick as in the proof of Lemma~\ref{lem:annealed_est_M}. Now, we need to show that $\E\circ \calq_\omega (D_{\epsilon,L})$ decays sufficiently fast in $L$, and we do so in Lemma~\ref{lem:annealed_est_D} below. Note that by Lemma~\ref{lem:annealed_est_M}, it is sufficient to show that $\E \circ \calq_\omega( D_{\epsilon,L} \cap M_{\epsilon, L}^c)$ is decaying. 

An instructive calculation to perform before we prove our next Lemma is to evaluate $\E\circ \cale_\omega \{t(x_2 - x_1)\}$ which can be done as follows. First, notice that for any $y\ge -L$ and any $k\ge 1$ we may rewrite the event \[[x_1 = y, x_2 - x_1 = k]= [\pi_{-L}(\Sigma) = 0, \pi_{-L+1}(\Sigma) = 0, \dots, \pi_y(\Sigma) = 1, \pi_{y+1}(\Sigma) = 0, \dots, \pi_{y+k}(\Sigma) = 1],\] the latter being an intersection of $\P \circ \calq_\omega$-independent events. Therefore, 
    \begin{align*}
        \E\circ \cale_\omega\, \exp\{ t(x_2 - x_1)\} &= \sum_{y\ge -L} \sum_{k_1 \ge 1} e^{tk_1} \P\circ \calq_\omega [y =1, y+k_1 = 1] 
        = \sum_{y\ge -L} \sum_{k_1 \ge 1} e^{tk_1} (s^2)^{y+L}\, c^2\, (s^2)^{k_1 -1}\, c^2\\
        &= \frac{c^2  e^t}{1-s^2 e^t}\,,
    \end{align*} whenever $0<t<-\log(s^2)$ (recall that $s^2 \in (0, \tfrac12)$). Notice that this is the moment generating function of a geometric random variable with success parameter $c^2$ and mean $\mu = (c^2)^{-1}$. Recall that its logarithm is the cumulant generating function. 
    \begin{lem}\label{lem:legendre}
        Let $\kappa(t) = \log(c^2 e^t (1-s^2 e^t)^{-1})$ be the geometric CGF as above, which is defined for all $t\in (\log(s^2), -\log(s^2))$. Let 
        \begin{equation}
            \mathcal{L}_*(x) = \sup_t\, tx - \kappa(t)\,.
        \end{equation} Then the following hold. 
        \begin{enumerate}[label = \alph*.)]
            \item The function $\mathcal{L}_*$ is defined for all $x\in (1, \infty)$. Furthermore, one can explicitly compute 
            \begin{equation}
                \mathcal{L}_*(x) = x \log(\tfrac{1}{s^2}(1-x^{-1})) - \log(\tfrac{c^2}{s^2}
(x-1)) \,.            \end{equation}
            \item One has that $\mathcal{L}_*(\mu) = 0$, and $\mathcal{L}_*(x)$ is increasing and strictly positive whenever $x>\mu$. On the other hand, $\mathcal{L}_*(x)$ is strictly positive and decreasing whenever $1<x<\mu$. 
        \end{enumerate}
    \end{lem}
The proof follows by some elementary calculus. Turning toward our application, we have the following.

\begin{lem}\label{lem:annealed_est_D}
     For any $s^2>\epsilon>0$, any $L>0$, there are functions $\xi^\pm(\epsilon, L)$ so that following upper bound holds: 
    \begin{equation}\label{eqn:annealed_D_est}
        \E \calq_\omega(D_{\epsilon,L}) \le e^{-(c^2+\epsilon)(L+1)\mathcal{L}_*(\mu + \xi^+)} + e^{-(c^2-\epsilon)(L+1) \mathcal{L}_*(\mu - \xi^-)} \to 0 \text{ as $L\uparrow \infty$}\,.
    \end{equation} Moreover, the functions $\xi^{\pm}(\epsilon,L)$ are increasing in $L$ and for every $\epsilon$ in range, $\lim_{L\to \infty} \xi^{\pm} >0$.
\end{lem}
\begin{proof}
    We note as in the proof of Lemma~\ref{lem:annealed_est_M}, that the event $M_{\epsilon, L}^c$ induces a deterministic upper and lower bound on $N_L$. Observe that $D_{\epsilon, L}$ is a disjoint union of events $D^{\pm}_{\epsilon, L}$ where the average distance between excitation pairs is lower bounded by $(\tfrac12 + \epsilon)(2L+1)$ (respectively, upper bounded by $(\tfrac12 - \epsilon)(2L+1)$). Putting these two facts together, we see that 
    \begin{equation}
        \begin{split}
            D_{\epsilon,L}^+\cap M_{\epsilon, L}^c &\subset \left\{ \frac{1}{2L+1} \sum_{m=1}^{(c^2 + \epsilon)(L+1)} (x_{2m}-x_{2m-1}) > \epsilon + \tfrac{1}{2} \right\} =: F_{\epsilon, L}^+\,,\\
            D_{\epsilon, L}^- \cap M_{\epsilon, L}^c &\subset \left\{\frac{1}{2L+1} \sum_{m=1}^{(c^2 - \epsilon)(L+1)} (x_{2m}-x_{2m-1}) <  \tfrac{1}{2} - \epsilon \right\} =: F^-_{\epsilon, L}\,.
        \end{split}
    \end{equation} The aim now is to establish summable upper estimates for $\E \calq_\omega(F^{\pm}_{\epsilon, L} )$. 
    
     Using the same reasoning as in the lines preceding the statement of Lemma~\ref{lem:legendre}, we obtain
        \begin{align*}
            \E \calq_\omega (F_{\epsilon,L}^+) &\le e^{-t\,(2L+1) (\epsilon + \tfrac 12)}\E\, \cale_\omega \exp \left\{ t \sum_1^{(c^2+\epsilon)(L+1)} x_{2m} - x_{2m-1} \right\} = e^{-t\, (2L+1) (\epsilon + \tfrac 12)}\left( \frac{c^2  e^t}{1-s^2 e^t}\right)^{(c^2 + \epsilon)(L+1)} \,.
         \end{align*} Rewrite this as 
         \begin{equation}
             \E\circ \calq_\omega(F^{+}_{\epsilon, L}) \le \exp\{ -(c^2+\epsilon)(L+1)\ \mathcal{L}(t, \mu + \xi^+)\}\,,
         \end{equation} where we have set
         \begin{equation}\label{eqn:defofM}
             \mathcal{L}(t, x) = tx  - \log\left(\frac{c^2 e^t}{1-s^2e^t}\right)\, \quad\text{ for all $x>1$ }\,,
        \end{equation}
        and, 
        \begin{equation}
             \xi^+ := \frac{(2c^2 -1)\epsilon L - c^2(\tfrac12 -\epsilon)-\epsilon}{c^2(c^2+\epsilon)(L+1)} = \frac{(2L+1)(\epsilon + \tfrac12)}{(L+1)(c^2+\epsilon)} - \mu\,.
        \end{equation} Notice that $\xi^+$ is increasing in $L$ for every fixed $\epsilon >0$ with a strictly positive limit. We may optimize over $t$ to find that there is $t_* \in (0, - \log(s^2))$ depending on $\xi^+$ so that 
         \begin{equation}
             \E\circ \calq_\omega [F^+_{\epsilon, L}] \le e^{-(c^2+\epsilon)(L+1)\mathcal{L}_*(\mu + \xi^+)}\,, 
         \end{equation} where $\mathcal{L}_*= \mathcal{L}(t_*, \mu+\xi^+)$ as in Lemma~\ref{lem:legendre} above. Since $\lim_{L \to \infty} \xi^+ >0$ for every $\epsilon$, we can apply Lemma~\ref{lem:legendre} again to see that the resulting bound obeys the following.
         \begin{equation}
             \E\circ \calq_\omega(F^{+}_{\epsilon,L}) \le e^{-(c^2+\epsilon)(L+1) \mathcal{L}_*(\mu + \xi^+)} \to 0 \text{ as }L\to \infty\,. 
         \end{equation}

         In the $F_{\epsilon,L}^-$ case one may play the same game with Chebyshev's inequality to obtain 
         \begin{equation}
             \E\circ \calq_\omega(F_{\epsilon,L}^-) \le \exp\{ -(c^2 - \epsilon)(L+1) \mathcal{L}(-t, \mu - \xi^-)\}\,,
         \end{equation} where we have set 
         \[
            \xi^- := \frac{(2c^2 -1)\epsilon L - s^2\epsilon+ \tfrac{c^2}{2}}{c^2(c^2-\epsilon)(L+1)} = \mu - \frac{(\tfrac12 - \epsilon)(2L+1)}{(c^2-\epsilon)(L+1)} \,.
         \] Notice that $\xi^-$ is increasing in $L$ and that whenever $\epsilon<s^2$, the limit $\lim_{L\to \infty} \xi^- < 1-\mu$. Hence $\mu - \xi^- >1$ for all $L$ large enough. Optimizing in $-t$ this time, we find 
         \begin{equation}
             \E \circ \calq_\omega(F^{-}_{\epsilon ,L}) \le e^{-(c^2-\epsilon)(L+1) \mathcal{L}_*(\mu - \xi^-)}\,,
         \end{equation} where $\mathcal{L}_*$ is as before. By Lemma~\ref{lem:legendre} once again, we find that $\mathcal{L}_*(\mu - \xi^-)$ monotonically increases to some nonzero limit. Hence, 
         \begin{equation}
             \E \circ \calq_\omega(F_{\epsilon, L}^-) \le e^{-(c^2 - \epsilon)(L+1) \mathcal{L}_*(\mu - \xi^-)} \to 0 \text{ as } L \to \infty\,,
         \end{equation} as claimed. 
\end{proof}

\begin{cor}\label{cor:quenched_limit_D}
    For any $c^2>\epsilon>0$ and $\P$-a.e. $\omega$, the limit $\lim_{L\to \infty} \calq_\omega D_{\epsilon, L} =0$.
\end{cor}
\begin{proof}
    Using the same notation as in Lemma~\ref{lem:annealed_est_D}, we employ Lemma~\ref{lem:legendre} part b.) to conclude that the following constants are strictly positive: \[
       K_\epsilon^\pm := \lim_{L\to \infty} \mathcal{L}_*(\mu \pm \xi^{\pm}) 
    \] Now, let $\eta >0$ be small. Since the convergence $\mathcal{L}_*(\mu \pm \xi^{\pm}) \to K_\epsilon ^{\pm}$ is monotone increasing in both cases, there is some $L_0$ sufficiently large so that the following bounds hold for all $L\ge L_0$: 
    \[
            e^{-(c^2\pm \epsilon)(L+1)\mathcal{L}_*(\mu\pm \xi^{\pm})} \le e^{-(c^2\pm\epsilon)(L+1)(K_{\epsilon}^{\pm} - \eta)}\,.
    \]  
    The claim now follows because we have shown that the upper bound in~(\ref{eqn:annealed_D_est}) is summable in $L$ for all $s^2>\epsilon>0$. Using the same trick as in Corollary~\ref{cor:quenched_limit_M} yields the almost sure convergence. 
\end{proof}

We are now ready to prove Claim~\ref{claim:THM_D}.
\begin{proof}[Proof of Claim~\ref{claim:THM_D}]
    Let us write $f = \chi_{\{N_L \text{ even }\}}\cos\left(\frac{2\pi}{2L+1}\sum_{m=1}^{N_L/2} [x_{2m}-x_{2m-1}] \right)$.  Utilizing our representation of $ \langle \Gamma_L |T_L|\Gamma_L\rangle$ as the expectation of $f$ with respect to $\calq_\omega$, we may split the quenched expected values into into a sum of three terms 
    \[
        |\langle \Gamma_L|T_L|\Gamma_L\rangle + \|\Gamma_L\|^2| \le  |\cale_\omega(f+1; M_{\epsilon,L})| + |\cale_\omega( f+1; M_{\epsilon, L}^c, D_{\epsilon, L})| + |\cale_\omega(f+1; M_{\epsilon, L}^c, D_{\epsilon, L}^c)|\,.
    \] By Corollaries~\ref{cor:quenched_limit_M} and~\ref{cor:quenched_limit_D}, the first two summands decay to zero almost surely as $L\to \infty$. Turning to the last term, we see that
    \begin{align*}
        |\cale_\omega(f +1; D_{\epsilon, L}^c, M_{\epsilon, L}^c)| =|\cale_\omega(f - \cos(\pi); D_{\epsilon, L}^c, M_{\epsilon, L}^c)|& \le  \cale_\omega(|f - \cos(\pi )|; D_{\epsilon,L}^c, M_{\epsilon, L}^c)\\
        &\le 2\pi \cale_\omega \left(| \sum_{m=1}^{N_L/2} \frac{x_{2m} - x_{2m-1}}{2L+1} - \frac{1}{2}|;D_{\epsilon, L}^c, M_{\epsilon, L}^c \right) < 2 \pi \epsilon\,.
    \end{align*} where we have used the mean value theorem to estimate the integrand in the second to last inequality. Therefore, 
    \[
        | \langle \Gamma_L | T_L | \Gamma_L\rangle + \|\Gamma_L\|^2 | \le 4\pi \epsilon + 4\calq_\omega(M_{\epsilon,L}) + 4\calq_\omega(M_{\epsilon, L}^c, D_{\epsilon, L})\,,
    \] as claimed. 
\end{proof}

\appendix
\section{Appendix: Computing the intersection of Ground State Spaces}\label{apdx:mathematica}
In this appendix, we record the Wolfram code we used to compute the null spaces of  several (relatively) large matrices in the proof of Lemma~\ref{lem:delta_MPS_injectivity}. We will also compute the orthonormal basis comprising the Parent Hamiltonian from Corollary~\ref{cor:Parent_Hamiltonian}.

To make the exposition more legible, we adopt the convention that $\ell_{x,y} := \phi_{x,y}(j)$ and $r_{x,y} := \phi_{x,y}(j+1)$, where $\phi_{x,y}$ is defined in equation~(\ref{eqn:gj_basis}) above. We will list the computations in descending order of the total spin. Then, we consider the possibilities for the additional vector in $\C^3\ket{m}$ allowing $m\in\{-1,0,+1\}$.  So there are only certain allowed combinations with the basis we found for $\mathcal{G}_{12}$ in Lemma~\ref{lem:delta_MPS_injectivity}. Namely, those whose spin indices in the computational basis add to $0$ (e.g. $\ket{-0+}$). 

\subsection{Spin $2$ Intersection}\label{apdx:spin2}
First, for spin $+2$, the quantum numbers $k$ and $m$ and $k'$, $m'$ which yield $+2$ are $k=k'=m=m'=+1$. The resulting subspaces are therefore $L(2) = \mathcal{G}_{1,2}^{(1)}\otimes \C\ket{+}$ and $R(2) = \C\ket{+}\otimes \mathcal{G}_{2,3}^{(1)}$. Each is one-dimensional spanned by $\ket{\ell_{21},+}= \ket{\phi_{21}(j),+}$ and $\ket{+,r_{21}}=\ket{+, \phi_{21}(j+1)}$ respectively. Expanding both vectors as in equation~(\ref{eqn:spin_+2_basis}), we obtain
    \begin{equation}\label{eqn:spin_+2_basis}
        \begin{split}
            \ket{\phi_{21}(1), +} & = c_1 s_2 \ket{+0+} - s_1 c_2 \ket{0++}\,, \\
            \ket{+, \phi_{21}(2)}&= c_2 s_3 \ket{++0}- s_2 c_3 \ket{+0+},\,
        \end{split}
    \end{equation} which are linearly independent since $\omega_1, \omega_2, \omega_3 \in (0, \frac{\pi}{2})\cup (0, \pi)$ whence  $c_1s_2 \neq 0$, and $c_2s_3\neq 0$. The case of spin $-2$ is similar: the left-hand side is spanned by $\ket{\phi_{12}(1), -} = c_1s_2\ket{-0-} - s_1 c_2 \ket{0--}$ and the right hand side is spanned by $\ket{-, \phi_{12}(2)} = c_2s_3\ket{--0}-s_2c_3\ket{-0-}$, which are linearly independent.

\subsection{Spin 1 Intersection}\label{apdx:spin1}

Begin by initializing several variables and the associated assumptions. 
\begin{lstlisting}[language=Mathematica, style=mystyle]
In[1]:= $Assumptions = b1>= 0 && b1<= Pi && b2>=0 && b2<= Pi && b3>=0 && b3<= Pi;
In[2]:= c1 = Cos[b1]; c2 = Cos[b2]; c3 = Cos[b3] s1 = Sin[b1]; s2 = Sin[b2]; s3 = Sin[b3];
\end{lstlisting}

    Now, in the case that $k+m= +1$, the allowed values of are as follows: $(k,m) = (1,0), (0,1)$. We re-organize this information by $m$-value in the following tables. \begin{center}
        \begin{tabular}{c|c}
            $m$& basis vector of $L(+1)$ with $\sigma=1$\\
            \hline \hline
            $-$&  $\emptyset$\\\hline
            $0$&  $v_1:= |\phi_{21}(1), 0\> = -s_j c_{j+1} |0+0\> + c_j s_{j+1} |+00\>$\\\hline
            $+$ &  \begin{tabular}{@{}c@{}}
                    $v_2: = |\phi_{11}(1), +\> = s_j s_{j+1} |00+\> - c_jc_{j+1}|+-+\>$\\
                    $v_3 := |\phi_{22}(1), +\> = s_j s_{j+1} |00+\> - c_j c_{j+1} |-++\>$
            \end{tabular}\\
        \end{tabular}
    \end{center}

    Similarly for $R(+1)$, we get a family of basis vectors: 
    \begin{center}
            \begin{tabular}{c|c}
            $m'$& basis vector of $R(+1)$ with $\sigma=1$ \\
            \hline \hline
            $-$&  $\emptyset$\\\hline
            $0$&  $w_1:= |0, \phi_{21}(2)\> = -s_{j+1} c_{j+2} |00+\> + c_{j+1} s_{j+2} |0+0\>$\\\hline
            $+$ &  \begin{tabular}{@{}c@{}}
                    $w_2: = |+,\phi_{11}(2)\> = s_{j+1} s_{j+2} |+00\> - c_{j+1}c_{j+2}|++-\>$\\
                    $w_3 := |+,\phi_{22}(2)\> = s_{j+1} s_{j+2} |+00\> - c_{j+1} c_{j+2} |+-+\>$
            \end{tabular}\\
        \end{tabular}
    \end{center}
        Now, expanding in the total spin-z basis and writing down the corresponding entries yields to Table~\ref{tab:s_+1-matrix}.
    \begin{table}[h]
        \centering
        \begin{tabular}{c||c|c|c|c|c|c}
             & $v_1$ & $v_2$ & $v_3$ & $-w_1$ & $-w_2$ & $-w_3$\\
             \hline \hline
             $\ket{-++}$& & & $-c_jc_{j+1}$ & & \\
             $\ket{0+0}$& $-s_jc_{j+1}$ & & & $-c_{j+1}s_{j+2}$ &\\
             $\ket{00+}$& & $s_js_{j+1}$ & $s_js_{j+1}$ & $s_{j+1}c_{j+2}$ & \\
             $\ket{+00}$& $c_js_{j+1}$ & & & & $-s_{j+1}s_{j+2}$ & $-s_{j+1}s_{j+2}$ \\
             $\ket{+-+}$& & $-c_jc_{j+1}$ & & & &$c_{j+1}c_{j+2}$\\
             $\ket{++-}$& & & & & $c_{j+1}c_{j+2}$ &\\
             \hline
        \end{tabular}
        \caption{Spin +1 overlap matrix}
        \label{tab:s_+1-matrix}
    \end{table}  
We now enter this into Mathmatica by defining a $6\times 6$ matrix in the order of the entries we listed above:
\begin{lstlisting}[language=Mathematica, style=mystyle]
In[5]:= Q = {{0,0,-c1*c2, 0, 0,0},
	           {-s1*c2,0,0,-c2*s3,0,0},
	           {0, s1*s2, s1*s2, s2*c3,0,0},
	           {c1*s2,0,0,0, -s2*s3, -s2*s3},
	           {0,-c1*c2,0,0,0,c2*c3},
	           {0,0,0,0, c2*c3, 0}
            };
In[6]:= NullSpace[Q]
Out[6]= {{{Sec[b1] Sin[b3],Cos[b3] Sec[b1],0,-Tan[b1],0,1}}
\end{lstlisting}  In symbols, the null space of the matrix (i.e. the intersection) is spanned by the vector: 
        \begin{equation}
            \nu_1:= 2s_{3} v_1 + c_{3}v_2 = 2s_1 w_1 - c_1 w_3\,.
        \end{equation} This shows the dimension of this subspace is one provided that the angles are nondegenerate.

Working with the spin $-1$ case is similar: the resulting $6\times 6$ matrix given in Table~\ref{tab:s_-1-matrix} is the result of taking the appropriate tensor products corresponding to quantum numbers $k$ and $m$ adding to $-1$ (resp. $k'+m'=-1$ for the right hand side), then expanding according to the total spin-z basis.  
\begin{table}[h]
    \centering
        \begin{tabular}{c||c|c|c|c|c|c}
             & $\ket{\ell_{11}, -}$ & $\ket{\ell_{22}, -}$ & $\ket{\ell_{21},0}$ & $-\ket{-, r_{11}}$ & $-\ket{-, r_{22}}$ & $-\ket{0,r_{21}}$\\
             \hline \hline
             $\ket{--+}$& & & & & $c_{j+1}c_{j+2}$  \\
             $\ket{-00}$& & & $c_js_{j+1}$ & $-s_{j+1}s_{j+2}$ & $-s_{j+1}s_{j+2}$\\
             $\ket{-+-}$&  &$-c_jc_{j+1}$ & & $c_{j+1}c_{j+2}$ &  & \\
             $\ket{0-0}$&  & & $-s_jc_{j+1}$ & &  &  $-c_{j+1}s_{j+2}$\\
             $\ket{00-}$& $s_{j}s_{j+1}$ & $s_js_{j+1}$ & & & & $s_{j+1}s_{j+2}$\\
             $\ket{+--}$& $-c_{j}c_{j+1}$ & & & & &\\
             \hline
        \end{tabular}
        \caption{Spin -1 overlap matrix}
        \label{tab:s_-1-matrix}
        \end{table} 
        
        We employ the same commands to compute the null space:
\begin{lstlisting}[language=Mathematica, style=mystyle]
In[7]:= L =  {
	       {0, 0, 0, 0, c2*c3, 0},
	       {0, 0, c1*s2, -s2*s3, -s2*s3, 0},
	       {0, -c1*c2, 0, c2*c3, 0, 0},
	       {0, 0, -s1*c2, 0, 0, -c2*s3},
	       {s1*c2,s1*s2, 0, 0, 0, s2*c3},
	       {-s1*c2, 0, 0, 0, 0, 0}
	       };
In[8]:= NullSpace[L]
Out[8]= {{0, -Cos[b3] Csc[b1], -Csc[b1] Sin[b3], -Cot[b1], 0, 1}}
\end{lstlisting}

Once again, this matrix has a one-dimensional null space generated by the vector 
        \begin{equation}
            \nu_{-1} = -c_{3} |\ell_{22}, -\> - s_{3}|\ell_{12}, 0\>  = -c_1|-, r_{11}\> - s_1|0,r_{12}\>\,.
        \end{equation}  

Therefore, we have shown that $L(+1) \cap R(+1)$ is one-dimensional for all the allowed values of $\omega_j$. Similarly, the dimension of $L(-1) \cap (R(-1)$ is one dimensional for all the values of $\omega$ we allow. 

We now turn toward computing the intersection of the subspaces $L(0) \cap R(0)$ which is slightly trickier. 

\subsection{Spin 0 Intersection}
 Below, we compute the null space of the vectors corresponding to the spin 0 subspace of the intersection. There are seven kets that appear in the expansions, due to the allowed values of $k$ and $m$. They are: \[\ket{-0+},\ket{-+0},\ket{0-+},\ket{000},\ket{0+-},\ket{+-0},\ket{+0-}\] We label the rows of a $7 \times 8$ matrix $M$ (in this order) by these kets given in Table~\ref{tab:s_0-matrix}. 

         \begin{table}[h]
            \centering
            \begin{tabular}{c||c|c|c|c|c|c|c|c}
                  & $\ket{\ell_{12},-}$ & $\ket{\ell_{11},0}$ & $\ket{\ell_{22},0}$ & $\ket{\ell_{12},+}$ & $-\ket{-, r_{12}}$ &$-\ket{0, r_{11}}$ & $-\ket{0,r_{22}}$ &$ -\ket{+, r_{21}}$ \\
                  \hline \hline
                 $\ket{-0+}$&  &  &  & $c_1c_2$ & $c_2s_2$ & &  &  \\
                 $\ket{-+0}$&  &  & $-c_1c_2$ &  & $-c_2s_3$ & &  &  \\
                 $\ket{0-+}$&  &  &  & $-c_2s_1$ &  & & $c_2c_3$ &  \\
                 $\ket{000}$&  & $s_1s_2$ & $s_1s_2$ &  & & $-s_2s_3$ & $-s_2s_3$ &  \\
                 $\ket{0+-}$& $-c_{2}s_1$ &  &  &  & & $c_2c_3$ &  &  \\
                 $\ket{+-0}$&  & $-c_1c_2$ &  &  & & &  & $-c_2s_3$ \\
                 $\ket{+0-}$& $c_1s_2$ &  &  &  & & &  & $c_3s_2$ \\
                 \hline
            \end{tabular}
            \caption{Spin 0 overlap matrix}
            \label{tab:s_0-matrix}
        \end{table}

Once again, we set up the Null-space command in Mathematica and enter the data from Table~\ref{tab:s_0-matrix}. Keeping careful track of the rows leads us to conclude that the null-space is spanned by two vectors. 
\begin{lstlisting}[language=Mathematica, style=mystyle]
In[3]:= M = {
	   {0, 0, 0,c1*s2, s2*c3,  0, 0, 0},
	   {0, 0, -c1*c2, 0, -c2*s3, 0, 0, 0},
	   {0, 0, 0, -s1*c2, 0, 0, c2*c3, 0},	
	   {0, s1*s2, s1*s2, 0, 0, -s2*s3, -s2*s3, 0},
	   {-s1*c2, 0,  0, 0, 0, c2*c3, 0, 0},
	   {0, -c1*c2, 0, 0, 0, 0, 0, -c2*s3},
	   {c1*s2, 0, 0, 0, 0, 0, 0, s2*c3}
    };
In[4]:= NullSpace[M]
Out[4]= {
    {-Cos[b3] Sec[b1],-Sec[b1] Sin[b3],0,0,0,-Tan[b1],0,1},
    {0,0,Csc[b1] Sin[b3],Cos[b3] Csc[b1],-Cot[b1],0,1,0}
    }
\end{lstlisting}
We see that the nullspace is two-dimensional with the following vectors which span the intersection:
        \begin{equation}
            \begin{split}
                v_0^{1}&:= -c_3s_1\ket{\ell_{21},-} -s_1s_3\ket{\ell_{11},0}  = s_1^2\ket{0,r_{11}} - c_1s_1\ket{+,r_{12}}\,\\
                v_0^{2}&:= c_1s_3\ket{\ell_{22},0} +c_1c_3\ket{\ell_{12},+} = c_1^2\ket{-,r_{21}} - c_1s_1 \ket{0,r_{22}}\,.
          \end{split}
        \end{equation} Whenever the angles are not zeroes of the trigonometric functions, the dimension does not drop. Therefore, the dimension of the intersection is two. 

\subsection{Basis for the Parent Hamiltonian}
Now, we turn toward computing the parent Hamiltonian in terms of the basis of the ground state space. This described in \cite{Perez-Garcia_et_al}: we need to compute a basis for the ortho-compliment of $G_{12}$ as defined above. So we shall use the basis vectors from before, and then compute their null space.

\begin{lstlisting}[language=Mathematica, style=mystyle]
In[9]:= L =  {
	       {0, 0, 0, 0, c2*c3, 0},
	       {0, 0, c1*s2, -s2*s3, -s2*s3, 0},
	       {0, -c1*c2, 0, c2*c3, 0, 0},
	       {0, 0, -s1*c2, 0, 0, -c2*s3},
	       {s1*c2,s1*s2, 0, 0, 0, s2*c3},
	       {-s1*c2, 0, 0, 0, 0, 0}
	       };
In[10}:= V = NullSpace[P]
Out[10]:= {{0,0,0,0,0,0,0,0,1},
            {0,0,0,0,0,Cot[b1] Tan[b2],0,1,0},
            {0,0,1,0,Cot[b1] Cot[b2],0,1,0,0},
            {0,Cot[b2] Tan[b1],0,1,0,0,0,0,0},
            {1,0,0,0,0,0,0,0,0}
        }
\end{lstlisting}
Now that we have the vectors which span the null space, we compute the Grahm-Schmidt Orthogonalization below. The \url{FullSimplify} command reduces the algebraic expressions as much as the Wolfram Language will allow. 

\begin{lstlisting}[language=Mathematica, style=mystyle, escapeinside={(*}{*)}]
In[11]:=Orthogonalize[V] // FullSimplify
Out[11]:= {{0, 0, 0, 0, 0, 0, 0, 0, 1},
           {0, 0, 0, 0, 0, (*$\frac{\cot (\text{b1}) \tan (\text{b2})}{\sqrt{| \cot (\text{b1}) \tan (\text{b2})| ^2+1}}$*), 0, (*$\frac{1}{\sqrt{| \cot (\text{b1}) \tan (\text{b2})| ^2+1}}$*), 0},
            {0, 0, (*$\frac{1}{\sqrt{| \cot (\text{b1}) \cot (\text{b2})| ^2+2}}$*), 0, (*$\frac{\cot (\text{b1}) \cot (\text{b2})}{\sqrt{| \cot (\text{b1}) \cot (\text{b2})| ^2+2}}$*), 0, (*$\frac{1}{\sqrt{| \cot (\text{b1}) \cot (\text{b2})| ^2+2}}$*), 0, 0},
            {0, (*$\frac{\tan (\text{b1}) \cot (\text{b2})}{\sqrt{| \cot (\text{b2}) \tan (\text{b1})| ^2+1}}$*), 0, (*$\frac{1}{\sqrt{| \cot (\text{b2}) \tan (\text{b1})| ^2+1}}$*), 0, 0, 0, 0, 0},
            {1, 0, 0, 0, 0, 0, 0, 0, 0}
            }
\end{lstlisting} The outer product of these five vectors (after they have been matched with the appropriate vectors in the spin-z basis), forms the projection that defines the nearest neighbor Hamiltonian. 

\section*{Conflict of Interest Statement}
On behalf of all authors, the corresponding author states there are no conflicts of interest.

\section*{Data Availibility Statement}
This manuscript has no associated data.

\bibliographystyle{plain}
\bibliography{refs}

\begin{thebibliography}{10}

\bibitem{Abanin_et_al}
{Abanin, D. A. and De Roeck, W. and Huveneers, F.}
\newblock Theory of many-body localization in periodically driven systems.
\newblock {\em Ann. Physics}, 372:1--11, 2016.

\bibitem{AKLT}
I.~Affleck, T.~Kennedy, E.~H. Lieb, and H.~Tasaki.
\newblock Valence bond ground states in isotropic quantum antiferromagnets.
\newblock {\em Communications in Mathematical Physics}, 115(3):477--528, 1988.

\bibitem{AizenmanGraf}
M.~Aizenman and G.~M. Graf.
\newblock Localization bounds for an electron gas.
\newblock {\em J. Phys. A}, 31(32):6783--6806, 1998.

\bibitem{AizenmanWarzel}
M.~Aizenman and S.~Warzel.
\newblock {\em Random operators}, volume 168 of {\em Graduate Studies in
  Mathematics}.
\newblock American Mathematical Society, Providence, RI, 2015.
\newblock Disorder effects on quantum spectra and dynamics.

\bibitem{Beaudry_et_al}
A.~Beaudry, M.~Hermele, M.~J. Pflaum, M.~Qi, D.~D. Spiegel, and D.~T. Stephen.
\newblock A classifying space for phases of matrix product states, 2025.
\newblock arXiv:2501.14241.

\bibitem{Bhatia_MA}
R.~Bhatia.
\newblock {\em Matrix analysis}, volume 169 of {\em Graduate Texts in
  Mathematics}.
\newblock Springer-Verlag, New York, 1997.

\bibitem{Bhatia_PDM}
R.~Bhatia.
\newblock {\em Positive definite matrices}.
\newblock Princeton Series in Applied Mathematics. Princeton University Press,
  Princeton, NJ, 2007.

\bibitem{BolsDeRoeck}
A.~Bols and W.~De~Roeck.
\newblock Asymptotic localization in the {B}ose-{H}ubbard model.
\newblock {\em J. Math. Phys.}, 59(2):021901, 28, 2018.

\bibitem{Bols_et_al}
A.~Bols, W.~De~Roeck, M.~De~Wilde, and B.~de~O.~Carvalho.
\newblock Classification of locality preserving symmetries on spin chains,
  2025.
\newblock arXiv:2503.15088.

\bibitem{BugerolLacroix}
P.~Bougerol and J.~Lacroix.
\newblock {\em Products of random matrices with applications to
  {S}chr\"{o}dinger operators}, volume~8 of {\em Progress in Probability and
  Statistics}.
\newblock Birkh\"{a}user Boston, Inc., Boston, MA, 1985.

\bibitem{BratteliRobinsonI}
O.~Bratteli and D.~W. Robinson.
\newblock {\em Operator algebras and quantum statistical mechanics. 1}.
\newblock Texts and Monographs in Physics. Springer-Verlag, New York, second
  edition, 1987.
\newblock $C^\ast$- and $W^\ast$-algebras, symmetry groups, decomposition of
  states.

\bibitem{BridgemanChubb}
J.~C. Bridgeman and C.~T. Chubb.
\newblock Hand-waving and interpretive dance: an introductory course on tensor
  networks.
\newblock {\em Journal of Physics A: Mathematical and Theoretical},
  50(22):223001, may 2017.

\bibitem{BrownOzawa}
N.~P. Brown and N.~Ozawa.
\newblock {\em {$C^*$}-algebras and finite-dimensional approximations},
  volume~88 of {\em Graduate Studies in Mathematics}.
\newblock American Mathematical Society, Providence, RI, 2008.

\bibitem{CarmonaLacroix}
R.~Carmona and J.~Lacroix.
\newblock {\em Spectral theory of random {S}chr\"{o}dinger operators}.
\newblock Probability and its Applications. Birkh\"{a}user Boston, Inc.,
  Boston, MA, 1990.

\bibitem{Chen_et_al}
L.~Chen, R.~J. Garcia, K.~Bu, and A.~Jaffe.
\newblock Magic of random matrix product states.
\newblock {\em Phys. Rev. B}, 109:174207, May 2024.

\bibitem{Chirvasitu}
A.~Chirvasitu.
\newblock Small banach bundles and modules, 2024.

\bibitem{Cirac_et_al}
J.~I. Cirac, D.~P\'erez-Garc\'{\i}a, N.~Schuch, and F.~Verstraete.
\newblock Matrix product states and projected entangled pair states: Concepts,
  symmetries, theorems.
\newblock {\em Rev. Mod. Phys.}, 93:045003, Dec 2021.

\bibitem{Conway_FA}
J.~Conway.
\newblock {\em A course in functional analysis}, volume~96 of {\em Graduate
  Texts in Mathematics}.
\newblock Springer-Verlag, New York, second edition, 1990.

\bibitem{Conway_OT}
J.~Conway.
\newblock {\em A course in operator theory}, volume~21 of {\em Graduate Studies
  in Mathematics}.
\newblock American Mathematical Society, Providence, RI, 2000.

\bibitem{CarvalhodeRoeckJappens}
B.~de~Oliveira~Carvalho, W.~De~Roeck, and T.~Jappens.
\newblock Classification of symmetry protected states of quantum spin chains
  for continuous symmetry groups, 2024.
\newblock arXiv:2409.01112.

\bibitem{DeRoeckHuveneerset_al}
W.~De~Roeck, F.~Huveneers, B.~Meeus, and A.~O. Pro\'sniak.
\newblock Rigorous and simple results on very slow thermalization, or
  quasi-localization, of the disordered quantum chain.
\newblock {\em Phys. A}, 631:Paper No. 129245, 20, 2023.

\bibitem{DiximierDouady}
J.~Dixmier and A.~Douady.
\newblock Champs continus d'espaces hilbertiens et de {$C\sp{\ast}
  $}-alg\`ebres.
\newblock {\em Bull. Soc. Math. France}, 91:227--284, 1963.

\bibitem{DupreGillette}
M.~J. Dupr\'e and R.~M. Gillette.
\newblock {\em Banach bundles, {B}anach modules and automorphisms of
  {$C\sp{\ast} $}-algebras}, volume~92 of {\em Research Notes in Mathematics}.
\newblock Pitman (Advanced Publishing Program), Boston, MA, 1983.

\bibitem{Durrett}
R.~Durrett.
\newblock {\em Probability---theory and examples}, volume~49 of {\em Cambridge
  Series in Statistical and Probabilistic Mathematics}.
\newblock Cambridge University Press, Cambridge, fifth edition, 2019.

\bibitem{EffrosRuan}
E.~G. Effros and Z-J. Ruan.
\newblock {\em Operator spaces}, volume~23 of {\em London Mathematical Society
  Monographs. New Series}.
\newblock The Clarendon Press, Oxford University Press, New York, 2000.

\bibitem{Eisner_et_al}
T.~Eisner, B.~Farkas, M.~Haase, and R.~Nagel.
\newblock {\em Operator theoretic aspects of ergodic theory}, volume 272 of
  {\em Graduate Texts in Mathematics}.
\newblock Springer, Cham, 2015.

\bibitem{EkbladMorenoNadalesRoonSchenker}
O.~Ekblad, E.~Moreno-Nadales, E.~B. Roon, and J.~H. Schenker.
\newblock Parent hamiltonians for ergodic matrix product states.
\newblock \textit{In Preparation}.

\bibitem{Fanizzaet_al}
M.~Fanizza, J.~Lumbreras, and A.~Winter.
\newblock Quantum theory in finite dimension cannot explain every general
  process with finite memory.
\newblock {\em Communications in Mathematical Physics}, 405(2):50, 2024.

\bibitem{FannesNachtergaeleWerner_abundance}
M.~Fannes, B.~Nachtergaele, and R.~F. Werner.
\newblock Abundance of translation invariant pure states on quantum spin
  chains.
\newblock {\em Lett. Math. Phys.}, 25(3):249--258, 1992.

\bibitem{FannesNachtergaeleWerner_entropy}
M.~Fannes, B.~Nachtergaele, and R.~F. Werner.
\newblock Entropy estimates for finitely correlated states.
\newblock {\em Ann. Inst. H. Poincar\'e{} Phys. Th\'eor.}, 57(3):259--277,
  1992.

\bibitem{FannesNachtergaeleWerner}
M.~Fannes, B.~Nachtergaele, and R.~F. Werner.
\newblock Finitely correlated states on quantum spin chains.
\newblock {\em Communications in Mathematical Physics}, 144(3):443--490, 1992.

\bibitem{FannesNachtergaeleWerner_pure}
M.~Fannes, B.~Nachtergaele, and R.~F. Werner.
\newblock Finitely correlated pure states.
\newblock {\em J. Funct. Anal.}, 120(2):511--534, 1994.

\bibitem{FellDoran}
J.~M.~G. Fell and R.~S. Doran.
\newblock {\em Representations of {$^*$}-algebras, locally compact groups, and
  {B}anach {$^*$}-algebraic bundles. {V}ol. 1}, volume 125 of {\em Pure and
  Applied Mathematics}.
\newblock Academic Press, Inc., Boston, MA, 1988.
\newblock Basic representation theory of groups and algebras.

\bibitem{FernandezGonzalez_et_al}
C.~Fern\'andez-Gonz\'alez, N.~Schuch, M.~M. Wolf, J.~I. Cirac, and
  D.~P\'erez-Garc\'ia.
\newblock Frustration free gapless {H}amiltonians for matrix product states.
\newblock {\em Comm. Math. Phys.}, 333(1):299--333, 2015.

\bibitem{Gierz}
G.~Gierz.
\newblock {\em Bundles of topological vector spaces and their duality},
  volume~57 of {\em Queen's Papers in Pure and Applied Mathematics}.
\newblock Springer-Verlag, Berlin-New York, 1982.
\newblock With an appendix by the author and Klaus Keimel, Lecture Notes in
  Mathematics, 955.

\bibitem{Gutman_1}
A.~E. Gutman.
\newblock Banach bundles in the theory of lattice-normed spaces. {I}.
  {C}ontinuous {B}anach bundles.
\newblock {\em Siberian Adv. Math.}, 3(3):1--55, 1993.
\newblock Siberian Advances in Mathematics.

\bibitem{GutmanKoptev}
A.~E. Gutman and A.~V. Koptev.
\newblock On the notion of the dual of a {B}anach bundle.
\newblock {\em Siberian Adv. Math.}, 9(1):46--98, 1999.

\bibitem{Hamana}
M.~Hamana.
\newblock Injective envelopes of operator systems.
\newblock {\em Publ. Res. Inst. Math. Sci.}, 15(3):773--785, 1979.

\bibitem{HamzaSimsStolz}
E.~Hamza, R.~Sims, and G.~Stolz.
\newblock Dynamical localization in disordered quantum spin systems.
\newblock {\em Comm. Math. Phys.}, 315(1):215--239, 2012.

\bibitem{Hayden_et_al}
P.~Hayden, S.~Nezami, X-L. Qi, N.~Thomas, M.~Walter, and Z.~Yang.
\newblock Holographic duality from random tensor networks.
\newblock {\em Journal of High Energy Physics}, 2016(11):9, 2016.

\bibitem{Heikkinen}
N.~Heikkinen.
\newblock The canonical forms of matrix product states in infinite-dimensional
  hilbert spaces, 2025.
\newblock arXiv:2502.12934.

\bibitem{HuNechita}
M.~Hu and I.~Nechita.
\newblock Canonical partial ordering from min-cuts and quantum entanglement in
  random tensor networks, 2025.

\bibitem{Husemoller}
D.~Husemoller.
\newblock {\em Fibre bundles}, volume~20 of {\em Graduate Texts in
  Mathematics}.
\newblock Springer-Verlag, New York, third edition, 1994.

\bibitem{JauslinLemm}
I.~Jauslin and M.~Lemm.
\newblock Random translation-invariant {H}amiltonians and their spectral gaps.
\newblock {\em {Quantum}}, 6:790, September 2022.

\bibitem{KadisonRingroseI}
R.~V. Kadison and J.~R. Ringrose.
\newblock {\em Fundamentals of the theory of operator algebras. {V}ol. {I}},
  volume~15 of {\em Graduate Studies in Mathematics}.
\newblock American Mathematical Society, Providence, RI, 1997.
\newblock Elementary theory, Reprint of the 1983 original.

\bibitem{KapustinSopenkoYang}
A.~Kapustin, N.~Sopenko, and B.~Yang.
\newblock A classification of invertible phases of bosonic quantum lattice
  systems in one dimension.
\newblock {\em J. Math. Phys.}, 62(8):Paper No. 081901, 16, 2021.

\bibitem{KavrukPaulsenTodorov}
A.~S. Kavruk, V.~I. Paulsen, I.~G. Todorov, and M.~Tomforde.
\newblock Quotients, exactness, and nuclearity in the operator system category.
\newblock {\em Adv. Math.}, 235:321--360, 2013.

\bibitem{KirchbergWassermann}
E.~Kirchberg and S.~Wassermann.
\newblock {$C^\ast$}-algebras generated by operator systems.
\newblock {\em J. Funct. Anal.}, 155(2):324--351, 1998.

\bibitem{LancienPerezGarcia}
C.~Lancien and D.~P\'erez-Garc\'ia.
\newblock Correlation length in random {MPS} and {PEPS}.
\newblock {\em Ann. Henri Poincar\'e}, 23(1):141--222, 2022.

\bibitem{Matsui_split}
T.~Matsui.
\newblock The split property and the symmetry breaking of the quantum spin
  chain.
\newblock {\em Comm. Math. Phys.}, 218(2):393--416, 2001.

\bibitem{Matsui_entropy}
T.~Matsui.
\newblock Boundedness of entanglement entropy and split property of quantum
  spin chains.
\newblock {\em Rev. Math. Phys.}, 25(9):1350017, 31, 2013.

\bibitem{MovassaghSchenker_PRX}
R.~Movassagh and J.~Schenker.
\newblock Theory of ergodic quantum processes.
\newblock {\em Phys. Rev. X}, 11:041001, Oct 2021.

\bibitem{MovassaghSchenker}
R.~Movassagh and J.~Schenker.
\newblock An ergodic theorem for quantum processes with applications to matrix
  product states.
\newblock {\em Comm. Math. Phys.}, 395(3):1175--1196, 2022.

\bibitem{Munkres}
J.~R. Munkres.
\newblock {\em Topology}.
\newblock Prentice Hall, Inc., Upper Saddle River, NJ, second edition, 2000.

\bibitem{Murphy}
G.~J. Murphy.
\newblock {\em {$C^*$}-algebras and operator theory}.
\newblock Academic Press, Inc., Boston, MA, 1990.

\bibitem{Nachtergaele}
B.~Nachtergaele.
\newblock The spectral gap for some spin chains with discrete symmetry
  breaking.
\newblock {\em Comm. Math. Phys.}, 175(3):565--606, 1996.

\bibitem{NachtergaeleOgataSims}
B.~Nachtergaele, Y.~Ogata, and R.~Sims.
\newblock Propagation of correlations in quantum lattice systems.
\newblock {\em J. Stat. Phys.}, 124(1):1--13, 2006.

\bibitem{NachtergaeleReschke}
B.~Nachtergaele and J.~Reschke.
\newblock Slow propagation in some disordered quantum spin chains.
\newblock {\em J. Stat. Phys.}, 182(1):Paper No. 12, 28, 2021.

\bibitem{NachtergaeleSims}
B.~Nachtergaele and R.~Sims.
\newblock Lieb-{R}obinson bounds and the exponential clustering theorem.
\newblock {\em Comm. Math. Phys.}, 265(1):119--130, 2006.

\bibitem{NelsonRoon}
B.~Nelson and E.~B. Roon.
\newblock Ergodic quantum processes on finite von {N}eumann algebras.
\newblock {\em J. Funct. Anal.}, 287(4):Paper No. 110485, 63, 2024.

\bibitem{Ogata_z2}
Y.~Ogata.
\newblock A {${\Bbb {Z}}_2$}-index of symmetry protected topological phases
  with time reversal symmetry for quantum spin chains.
\newblock {\em Comm. Math. Phys.}, 374(2):705--734, 2020.

\bibitem{Ogata_cmp}
Y.~Ogata.
\newblock A {$\Bbb{Z}_2$}-index of symmetry protected topological phases with
  reflection symmetry for quantum spin chains.
\newblock {\em Comm. Math. Phys.}, 385(3):1245--1272, 2021.

\bibitem{Ogata_cdm}
Y.~Ogata.
\newblock Classification of symmetry protected topological phases in quantum
  spin chains.
\newblock In {\em Current developments in mathematics 2020}, pages 41--104.
  Int. Press, Boston, MA, 2022.

\bibitem{Ogata_icm}
Y.~Ogata.
\newblock Classification of gapped ground state phases in quantum spin systems.
\newblock In {\em I{CM}---{I}nternational {C}ongress of {M}athematicians.
  {V}ol. 5. {S}ections 9--11}, pages 4142--4161. EMS Press, Berlin, [2023]
  \copyright 2023.

\bibitem{Paulsen}
V.~Paulsen.
\newblock {\em Completely bounded maps and operator algebras}, volume~78 of
  {\em Cambridge Studies in Advanced Mathematics}.
\newblock Cambridge University Press, Cambridge, 2002.

\bibitem{Perez-Garcia_et_al}
D.~Perez-Garcia, F.~Verstraete, M.~M. Wolf, and J.~I. Cirac.
\newblock Matrix product state representations.
\newblock {\em Quantum Inf. Comput.}, 7(5-6):401--430, 2007.

\bibitem{PG_StringOrder}
D.~P\'erez-Garc\'{\i}a, M.~M. Wolf, M.~Sanz, F.~Verstraete, and J.~I. Cirac.
\newblock String order and symmetries in quantum spin lattices.
\newblock {\em Phys. Rev. Lett.}, 100:167202, Apr 2008.

\bibitem{Petersen}
K.~Petersen.
\newblock {\em Ergodic theory}, volume~2 of {\em Cambridge Studies in Advanced
  Mathematics}.
\newblock Cambridge University Press, Cambridge, 1989.
\newblock Corrected reprint of the 1983 original.

\bibitem{Pollman_et_al_z2}
F.~Pollmann, E.~Berg, A.~M. Turner, and M.~Oshikawa.
\newblock Symmetry protection of topological phases in one-dimensional quantum
  spin systems.
\newblock {\em Phys. Rev. B}, 85:075125, Feb 2012.

\bibitem{Pollmann_et_al_entanglement}
F.~Pollmann, A.~M. Turner, E.~Berg, and M.~Oshikawa.
\newblock Entanglement spectrum of a topological phase in one dimension.
\newblock {\em Phys. Rev. B}, 81:064439, Feb 2010.

\bibitem{ResendeSantos}
P.~Resende and J.~P. Santos.
\newblock Open quotients of trivial vector bundles.
\newblock {\em Topology Appl.}, 224:19--47, 2017.

\bibitem{RoonSchenker_BGA}
E.~B. Roon and J.~H. Schenker.
\newblock Disordered ground states of ergodic quantum spin systems, 2026.
\newblock arXiv:2603.19475.

\bibitem{Rudin}
W.~Rudin.
\newblock {\em Functional Analysis}.
\newblock International Series in Pure and Applied Mathematics. McGraw-Hill,
  Inc., New York, second edition, 1991.

\bibitem{Sopenko_2d}
N.~Sopenko.
\newblock An index for two-dimensional {SPT} states.
\newblock {\em J. Math. Phys.}, 62(11):Paper No. 111901, 13, 2021.

\bibitem{Souissi}
A.~Souissi.
\newblock Matrix product states as observations of entangled hidden markov
  models, 2025.
\newblock arXiv: 2502.12641.

\bibitem{SouissiAndolsi}
A.~Souissi and A.~Andolsi.
\newblock A hidden quantum markov model framework for entanglement and
  topological order in the aklt chain, 2025.
\newblock arXiv:2512.18642.

\bibitem{SouissiAndolsi_II}
A.~Souissi and A.~Andolsi.
\newblock Aklt state is indeed the observation process of a causal hidden
  quantum markov model, 2026.
\newblock arXiv:2605.24431.

\bibitem{Spiegel}
D.~D. Spiegel.
\newblock {\em A C*-Algebraic Approach to Parametrized Quantum Spin Systems and
  Their Phases in One Spatial Dimension}.
\newblock PhD thesis, University of Colorado at Boulder, 2023.

\bibitem{Stollmann}
P.~Stollmann.
\newblock {\em Caught by Disorder: Bound states in random media}.
\newblock Progress in Mathematical Physics. {Birkh\"auser Boston, MA}, 2001.

\bibitem{Stolz}
G.~Stolz.
\newblock Aspects of the mathematical theory of disordered quantum spin chains.
\newblock In {\em Analytic trends in mathematical physics}, volume 741 of {\em
  Contemp. Math.}, pages 163--197. Amer. Math. Soc., [Providence], RI, [2020]
  \copyright 2020.

\bibitem{Takeda}
Z.~Takeda.
\newblock Inductive limit and infinite direct product of operator algebras.
\newblock {\em Tohoku Math. J. (2)}, 7:67--86, 1955.

\bibitem{TakesakiI}
M.~Takesaki.
\newblock {\em Theory of operator algebras. {I}}, volume 124 of {\em
  Encyclopaedia of Mathematical Sciences}.
\newblock Springer-Verlag, Berlin, 2002.
\newblock Reprint of the first (1979) edition, Operator Algebras and
  Non-commutative Geometry, 5.

\bibitem{Tasaki_prs}
H.~Tasaki.
\newblock Topological phase transition and ${\mathbb{z}}_{2}$ index for $s=1$
  quantum spin chains.
\newblock {\em Phys. Rev. Lett.}, 121:140604, Oct 2018.

\bibitem{Tasaki_book}
H.~Tasaki.
\newblock {\em Physics and mathematics of quantum many-body systems}.
\newblock Graduate Texts in Physics. Springer, Cham, [2020] \copyright 2020.

\bibitem{Tasaki_jmp}
H.~Tasaki.
\newblock Rigorous index theory for one-dimensional interacting topological
  insulators.
\newblock {\em J. Math. Phys.}, 64(4):Paper No. 041903, 19, 2023.

\bibitem{Tasaki_HeisenbergChain}
Hal Tasaki.
\newblock The ground state of the s=1 antiferromagnetic heisenberg chain is
  topologically nontrivial if gapped, 2024.

\bibitem{vanLuijk_Schmidt_Rank}
L.~van Luijk, R.~Schwonnek, A.~Stottmeister, and R.~F. Werner.
\newblock The {S}chmidt rank for the commuting operator framework.
\newblock {\em Comm. Math. Phys.}, 405(7):Paper No. 152, 46, 2024.

\bibitem{vanLuijk_Review}
L.~van Luijk, A.~Stottmeister, and H.~Wilming.
\newblock The large-scale structure of entanglement in quantum many-body
  systems, 2025.
\newblock arXiv:2503.03833.

\bibitem{VerstraeteCirac}
F.~Verstraete and J.~I. Cirac.
\newblock Matrix product states represent ground states faithfully.
\newblock {\em Phys. Rev. B}, 73:094423, Mar 2006.

\bibitem{Vidal}
G.~Vidal.
\newblock Efficient classical simulation of slightly entangled quantum
  computations.
\newblock {\em Phys. Rev. Lett.}, 91:147902, Oct 2003.

\end{thebibliography}
\end{document}